\documentclass[11pt]{article}
\usepackage{amsthm}
\usepackage{tikz}
\usepackage{physics, mathrsfs}
\usepackage{amsmath,amsfonts,amssymb}
\usepackage{geometry}
\usepackage{bbold}
\usepackage{hyperref}
\usepackage{slashed}
\usetikzlibrary{arrows.meta, decorations.pathmorphing, calc}

\newtheorem{theorem}{Theorem}

\newtheorem{proposition}{Proposition}

\title{Measurement-Induced Temporal Geometry}
\author{James C. Hateley}
\date{}

\begin{document}
\maketitle
\begin{abstract}
We propose a unified theoretical framework, Measurement-Induced Temporal Geometry (MTG), in which time, causality, and spacetime geometry emerge from quantum measurement acting on a fiber-valued internal time field. Each spacetime point supports a local degree of freedom $\tau$, modeled as a smooth section of a fiber bundle $\pi: E \to M$, with projection events $\mu[\tau]$ generating classical temporal flow. Quantum coherence and entanglement are encoded in the curvature $F = \nabla^2$ of a connection on the time-fiber, while the effective spacetime metric $g_{\mu\nu}^{\mathrm{eff}}$ arises as an integral over measurement histories. We derive the dynamical equations governing $\tau$, its supersymmetric completions, and the entanglement connection $A_\mu$, showing how quantization proceeds via both canonical and path-integral methods. Standard Model fields couple covariantly to the fiber geometry, and gravitational dynamics emerge from variational principles over projection-induced entropy. Cosmological inflation, dark energy, and large-scale structure are reinterpreted as consequences of modular coherence, topological obstruction, and fluctuations in the projection density $\rho(x)$. Within the AdS/CFT correspondence, MTG reinterprets modular Hamiltonians as boundary projections of bulk time flow and identifies entanglement wedges with surfaces minimizing measurement-induced projection current. A UV-complete embedding arises through string theory, where $\tau$ descends from compactified moduli and projection corresponds to brane interaction and spontaneous supersymmetry breaking. The framework yields a set of falsifiable predictions, including CMB anisotropies, black hole ringdown echoes, and modular deviations in lab-scale quantum systems, offering a consistent and testable account of spacetime as an emergent property of quantum measurement.
\end{abstract}

%\tableofcontents

\section{Introduction}\label{sec:1}
The nature of time remains one of the deepest unresolved questions in theoretical physics. General relativity models time as a geometric coordinate; inseparably woven into the fabric of spacetime, while quantum mechanics treats it as a fixed external parameter~\cite{isham1993canonical}. This asymmetry becomes untenable in regimes where both gravity and quantum effects are essential: near black hole horizons, in the early universe, or in measurements on entangled quantum systems~\cite{milekhin2024measurement}.

Measurement-Induced Temporal Geometry (MTG) resolves this by proposing that time is not a fundamental background parameter but an emergent, relational phenomenon arising from quantum measurement~\cite{debianchi2024achronotopic, raasakka2025triangulated}. In this framework, each measurement projects an internal, fiber-valued time field onto a classical outcome, inducing localized temporal flow and causal structure. Unlike prior approaches that treat time as a fixed background or seek to quantize classical geometry, MTG defines time as a dynamically projected observable, with collapse events generating an effective metric, modular flow, and entropy production. The framework develops a new geometric formalism linking curvature of the time-fiber bundle to coherence, decoherence, and gravitational dynamics, yielding a modular Hamiltonian constructed from informational flux. MTG further predicts observable signatures; including black hole echoes, CMB anomalies, and deviations from unitary dynamics in quantum simulations, while embedding consistently into supersymmetric field theory, AdS/CFT holography, and string-theoretic brane configurations. Together, these elements offer a unified and experimentally accessible reformulation of spacetime as an emergent phenomenon grounded in quantum measurement.

To formalize this picture, MTG models time as a field \( \tau(x) \) taking values in a fiber bundle \( \pi: E \to M \), where \( M \) is a \( d \)-dimensional Lorentzian manifold representing spacetime, and each fiber \( T_x \) encodes observer-relative internal time degrees of freedom~\cite{kiefer2007time}. A connection \( \nabla \) on this bundle defines parallel transport and curvature, with the latter serving as a geometric measure of quantum entanglement between events~\cite{zhang2018entanglement}.

Causal structure and spacetime geometry then emerge from sequences of these measurement-induced projections. A scalar field \( \rho(x) \) describes the local density of such events, and the accumulation of projections across causal domains yields an effective metric \( g_{\mu\nu}^{\text{eff}} \) experienced by observers~\cite{chishtie2025spatialenergy}. Concepts like entanglement, coherence, and modular flow—typically associated with quantum information—appear here as geometric entities governed by the curvature of the time-fiber bundle~\cite{giacomini2019quantum, jemal2016modular, qi2018quantum}.

The MTG framework accommodates supersymmetric extensions, coupling to the Standard Model via fiber-covariant dynamics, and a holographic embedding through AdS/CFT, in which modular Hamiltonians generate emergent bulk time~\cite{jemal2016modular, maldacena1999ads}. Quantization proceeds through canonical and path-integral approaches, with the latter encoding measurement as localized projection constraints~\cite{tomaz2025rqd}. Gravity emerges as an entropic response to accumulated projection history, and cosmological expansion is reinterpreted as the unfolding of informational structure~\cite{chishtie2025spatialenergy, jacobson1995thermodynamics}.

This paper develops MTG systematically: beginning with its geometric foundations, deriving dynamical and quantized structures, embedding known field theories, and exploring consequences for cosmology, holography, and UV completion in string theory.

\section{Foundations}\label{sec:2}
The Measurement-Induced Temporal Geometry (MTG) framework postulates that time is not a global, external parameter, but a local geometric field valued in an internal fiber. Let \( (M, \eta_{\mu\nu}) \) be a smooth, oriented Lorentzian manifold of dimension \( d \), interpreted as the spacetime of classical events. Over \( M \), we define a smooth fiber bundle
\begin{equation}
\pi : E \to M,
\end{equation}
where each fiber \( T_x = \pi^{-1}(x) \) represents internal temporal degrees of freedom at spacetime point \( x \in M \). These fibers may carry algebraic and topological structure; such as quaternionic algebras \( \mathbb{H} \), complex tori \( \mathbb{C}/\Lambda \), or higher-dimensional noncommutative spaces equipped with modular or spin structures~\cite{bakas2024modulartime}.

A smooth section \( \tau \in \Gamma(E) \) assigns to each \( x \in M \) a value \( \tau(x) \in T_x \), representing the internal time state of that point. Unlike external time parameters, the dynamics of \( \tau \) are determined intrinsically by the fiber geometry and its interaction with quantum measurement. The projection of \( \tau \) during a measurement induces a classical temporal direction and defines the emergence of effective causal structure. This section is the core dynamical field from which classical time is constructed.
%%%%%%%%%%%%%%%%%%%%
\subsection{Connection and Curvature on the Time-Fiber Bundle}

To describe how internal temporal structure varies across spacetime, MTG equips the bundle \( \pi: E \to M \) with a connection, which enables parallel transport and encodes temporal entanglement via its curvature.

Let \( G_{\text{time}} \) be the structure group of the bundle \( E \), capturing the symmetries of internal time fibers. We assume \( E \) is either a principal \( G_{\text{time}} \)-bundle or an associated vector bundle. A connection is then defined by a Lie algebra-valued one-form
\begin{equation}
A \in \Omega^1(M, \mathfrak{g}_{\text{time}}),
\end{equation}
where \( \mathfrak{g}_{\text{time}} = \text{Lie}(G_{\text{time}}) \). The connection defines a covariant derivative acting on sections of \( E \) as
\begin{equation}
D^\mu \tau := g^{\mu\nu} D_\nu \tau = g^{\mu\nu}(\partial_\nu \tau + A_\nu \cdot \tau),
\end{equation}
where \( A_\mu \) are local components of the connection, and the dot denotes the action of the Lie algebra on the fiber.

The curvature of this connection, measuring the failure of parallel transport to commute, is given by the Cartan structure equation:
\begin{equation}
F = dA + A \wedge A \in \Omega^2(M, \mathfrak{g}_{\text{time}}).
\end{equation}

In MTG, the curvature \( F \) serves both as a geometric field strength and as an order parameter for entanglement. If \( F = 0 \), the fibers are globally coherent and synchronizable; nonzero curvature obstructs this alignment and signals quantum entanglement between temporal states at different spacetime points.

Theorems below formalize these statements. Flatness implies global synchronizability (Theorem~\ref{thm:flatness_synchronization}), while curvature obstructs coherence (Theorem~\ref{thm:curvature_obstruction}). Physically, \( F \) quantifies the failure of internal time to reduce to a set of globally separable clocks and enters directly into the MTG action, influencing how coherence and entanglement shape emergent causal structure.

%%%%%%%%%%%%%%%%%%%%
\subsection{Dynamics of the Temporal Section}\label{sec:2.2}

The internal time field \( \tau \in \Gamma(E) \) evolves according to covariant dynamics defined by the geometry of the time-fiber bundle and its coupling to processes such as coherence and measurement. As a section of a vector or associated bundle over spacetime, \( \tau \) obeys field equations derived from a Lagrangian that encodes both kinetic propagation and potential terms that localize or structure its dynamics~\cite{kiefer2007time}.

The covariant derivative \( D_\mu \tau \), defined in Equation~(3), governs how internal time propagates through spacetime. The gauge-covariant kinetic term describing this propagation is given by
\begin{equation}
\mathcal{L}_{\tau,\text{kin}} = \frac{1}{2} \langle D_\mu \tau, D^\mu \tau \rangle,
\end{equation}
where the inner product \( \langle \cdot, \cdot \rangle \) is defined fiberwise and is invariant under the internal symmetry group \( G_{\text{time}} \).

In addition to propagation, the dynamics of \( \tau \) are shaped by potential terms that encode preferred vacua, periodic structures, or localization mechanisms inherent to the fiber geometry. These are captured by a scalar function \( V(\tau) \), assumed \( G_{\text{time}} \)-invariant, contributing
\begin{equation}
\mathcal{L}_{\tau,\text{pot}} = -V(\tau),
\end{equation}
where \( V : E \to \mathbb{R} \) is defined over the fibers of \( E \). The form of \( V(\tau) \) may enforce periodicity (e.g., for toroidal fibers), support symmetry-breaking minima, or penalize decoherent configurations~\cite{coleman1977fate, conlon1988physics, zurek2003decoherence}.

The full Lagrangian for \( \tau \) is thus
\begin{equation}
\mathcal{L}_\tau = \frac{1}{2} \langle D_\mu \tau, D^\mu \tau \rangle - V(\tau),
\end{equation}
yielding dynamics that are gauge-covariant and compatible with the MTG interpretation of time as an emergent structure. We reiterate, the inner product $\langle \cdot, \cdot \rangle$ is defined fiberwise on the bundle $E$, and we assume that it is Hermitian and compatible with the connection $D_\mu$. The covariant derivative $D^\mu \tau$ is shorthand for $g^{\mu\nu} D_\nu \tau$, with the background spacetime metric $g^{\mu\nu}$ (or $\eta^{\mu\nu}$ in flat regions) used to raise indices. The expression $\langle D_\mu \tau, D^\mu \tau \rangle$ thus contracts both the spacetime indices using $g^{\mu\nu}$ and the fiber indices via the Hermitian product.In the limit where curvature \( F = 0 \), the dynamics reduce to those of a conventional scalar field on classical spacetime. When \( F \neq 0 \), the evolution of \( \tau \) reflects entanglement and contextuality, with nonlocal behavior governed by the gauge structure and its coupling to measurement.
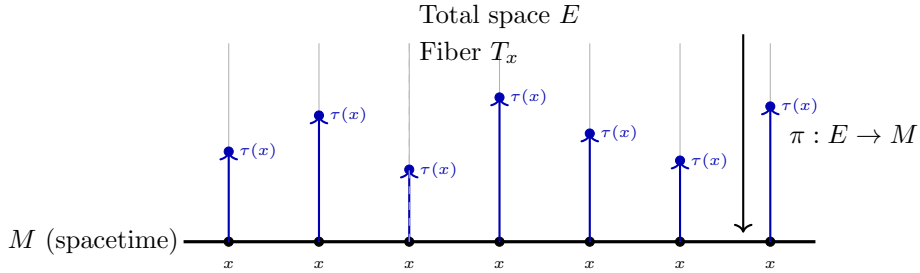
\begin{figure}[ht]
\centering
\begin{tikzpicture}[scale=1.2]

% Base manifold M
\draw[very thick] (-3.5,0) -- (3.5,0);
\node at (-4.5,0) {\small $M$ (spacetime)};
\foreach \x in {-3,-2,-1,0,1,2,3} {
    \draw[fill=black] (\x,0) circle (0.05);
    \node[below] at (\x,-0.1) {\tiny $x$};
}

% Fibers T_x and projection
\foreach \x/\tval in {
    -3/1.0, -2/1.4, -1/0.8, 0/1.6, 1/1.2, 2/0.9, 3/1.5
} {
    \draw[gray!60] (\x,0) -- (\x,2.2);
    \draw[->, thick, blue!70!black] (\x,0) -- (\x,\tval)
        node[pos=1.0, right] {\tiny $\tau(x)$};
    \filldraw[blue!70!black] (\x,\tval) circle (0.05);
}

% Fiber label
\draw[dashed, gray!50] (-1,0) -- (-1,2.2);
\node[right] at (-1,2.1) {\small Fiber $T_x$};

% Projection from total space
\node at (0,2.5) {\small Total space $E$};
\draw[->, thick] (2.7,2.3) -- (2.7,0.1);
\node[right] at (3.1,1.2) {\small $\pi: E \to M$};

\end{tikzpicture}
\caption{Fiber bundle structure for internal time. Each spacetime point \( x \in M \) has a corresponding fiber \( T_x \), and the internal time field \( \tau(x) \) defines a section assigning values in each fiber.}
\label{fig:fiber_bundle}
\end{figure}

%%%%%%%%%%%%%%%%%%%%%%%%%
\subsection{Variational Principle and Euler--Lagrange Equation}

The evolution of the internal time field \( \tau \) follows from a variational principle applied to the action constructed from the Lagrangian defined in Section~\ref{sec:2.2}. Let the action functional for \( \tau \) be
\begin{equation}
S_\tau = \int_M \mathcal{L}_\tau \, d^dx,
\end{equation}
with Lagrangian density
\begin{equation}
\mathcal{L}_\tau = \frac{1}{2} \langle D_\mu \tau, D^\mu \tau \rangle - V(\tau). \label{eq:cov-lag}
\end{equation}
We define the raised-index covariant derivative by
\begin{equation}
D^\mu \tau := g^{\mu\nu} D_\nu \tau,
\end{equation}
where \( g^{\mu\nu} \) is the inverse of the background spacetime metric. The composition
\begin{equation}
D_\mu D^\mu \tau := D_\mu(g^{\mu\nu} D_\nu \tau)
\end{equation}
includes both curvature and metric-compatibility terms, and may yield additional connection contributions if \( g^{\mu\nu} \) is not covariantly constant under \( D_\mu \). The fiber inner product \( \langle D_\mu \tau, D^\mu \tau \rangle \) contracts spacetime indices via \( g^{\mu\nu} \) and fiber indices via the Hermitian product.
To derive the field equations, we vary \( S_\tau \) with respect to the conjugate field \( \bar{\tau} \), treating \( \tau \) and \( \bar{\tau} \) (or \( \tau^\dagger \), in the complex case) as independent variables. Assuming the connection \( D_\mu \) is compatible with the fiber inner product, and integrating by parts, we obtain
\begin{equation}
\delta S_\tau = \int_M \delta \bar{\tau} \left( - D_\mu D^\mu \tau - \frac{\delta V}{\delta \bar{\tau}} \right) d^dx.
\end{equation}
Stationarity under arbitrary variations \( \delta \bar{\tau} \) yields the covariant Euler--Lagrange equation:
\begin{equation}
D_\mu D^\mu \tau + \frac{\delta V}{\delta \bar{\tau}} = 0.
\end{equation}
Since the internal time field $\tau$ is modeled as a section of a complex vector bundle $E \to M$, we treat $\tau$ and its Hermitian conjugate $\bar{\tau}$ (or $\tau^\dagger$)\footnote{Throughout, we use $\bar{\tau}$ to denote the complex conjugate of the internal time field when treating $\tau$ and $\bar{\tau}$ as independent variables in variational expressions. When expressing fiberwise Hermitian inner products or adjoint-valued currents, we use $\tau^\dagger:= (\bar{\tau})^T$ to denote the Hermitian conjugate.
} as independent variables in the variational principle. The fiberwise inner product $\langle \cdot, \cdot \rangle$ is assumed to be Hermitian, i.e., $\langle \phi, \psi \rangle = \overline{\langle \psi, \phi \rangle}$, and linear in the second argument. Functional differentiation with respect to $\bar{\tau}$ then yields the Euler–Lagrange equation in the complex setting.

This generalizes the Klein--Gordon equation to the MTG setting, where time is not an external parameter but a dynamical, fiber-valued field~\cite{kuchar1991time, rovelli1990quantum}. The operator \( D_\mu D^\mu \) incorporates both propagation across spacetime and parallel transport in the internal fiber geometry. The potential term introduces localization, periodicity, or symmetry-breaking effects depending on the structure of \( V(\tau) \).

In regions where the curvature \( F \) vanishes, \( \tau \) behaves as a conventional scalar field. When curvature is nonzero, the dynamics of \( \tau \) encode entanglement and coherence properties intrinsic to the temporal geometry~\cite{debianchi2024achronotopic}, and evolve under the influence of measurement as made explicit in later sections.

%%%%%%%%%%%%%%%%%%%%%%%
\subsubsection{Examples of Coherence Potentials}

Different physical regimes motivate distinct choices for the potential \( V(\tau) \), each encoding specific coherence properties of the internal time field.

A quadratic (harmonic) potential,
\begin{equation}
V(\tau) = \frac{1}{2} m^2 \langle \tau, \tau \rangle,
\end{equation}
penalizes large excursions of \( \tau \) from the fiber origin and models decoherence as a restoring force.

For fibers modeled as modular tori, a periodic potential of the form
\begin{equation}
V(\tau) = \Lambda^4 \left( 1 - \cos\left( \frac{2\pi \tau}{T} \right) \right)
\end{equation}
enforces alignment with modular structure, where minima correspond to discrete periodic configurations.

Spontaneous symmetry breaking can be introduced via a double-well potential:
\begin{equation}
V(\tau) = \lambda \left( \langle \tau, \tau \rangle - v^2 \right)^2,
\end{equation}
whose degenerate minima define coherent vacua with fixed fiber norm and broken internal symmetry.
To define the relative entropy potential, we associate to each fiber configuration \( \tau(x) \in T_x \) a normalized density matrix \( \rho_\tau \) on the internal time fiber. This density matrix represents the effective state induced by projecting the quantum system into the configuration \( \tau(x) \), and is assumed to arise from a partial trace or a Born-rule weighted projection onto a coherent state basis. Given a fiducial prior \( \rho_0 \), which may correspond to a uniform or maximally mixed reference state on the fiber, the relative entropy is defined as
\begin{equation}
S_{\mathrm{rel}}(\rho_\tau \| \rho_0) = \mathrm{Tr}\left[ \rho_\tau (\log \rho_\tau - \log \rho_0) \right].
\end{equation}
This potential penalizes configurations \( \tau \) that diverge significantly from the reference state, encoding a geometric measure of informational distance in the fiber bundle.

In information-theoretic settings, coherence may be encoded via a relative entropy functional,
\begin{equation}
V(\tau) = S_{\text{rel}}(\rho_\tau \| \rho_0),
\end{equation}
where \( \rho_\tau \) is the state induced by the projected internal time configuration, and \( \rho_0 \) is a fiducial prior~\cite{vedral2002role}.

These examples illustrate the flexibility of \( V(\tau) \) in shaping dynamical regimes. Depending on the global structure of the fiber bundle, topological defects, or the geometry of projection, different potentials give rise to qualitatively distinct temporal behaviors. The interplay between curvature, coherence, and collapse governs not only the evolution of \( \tau \), but also the emergence of classical causal structure~\cite{tomaz2025rqd}.

%%%%%%%%%%%%%%%
%%%%%%%%%%%%%%%
%%%%%%%%%%%%%%%
\subsection{Flatness Implies Classical Synchronizability}

The curvature of the time-fiber connection plays a central role in determining whether internal time can be consistently aligned across spacetime. In regions where the curvature vanishes, the internal temporal structure admits a global trivialization, corresponding to a classical regime in which \( \tau \) evolves without entanglement. The following result formalizes this observation.

\begin{theorem}\label{thm:flatness_synchronization}
Let \( \pi : E \to M \) be a principal \( G_{\text{time}} \)-bundle equipped with a connection one-form \( A \) and curvature \( F = dA + A \wedge A \). Suppose \( F = 0 \) on a simply connected open subset \( U \subseteq M \). Then there exists a smooth gauge transformation \( g : U \to G_{\text{time}} \) such that, in the corresponding trivialization, the connection vanishes: \( A^g = 0 \). Any temporal section \( \tau \in \Gamma(E|_U) \) satisfying \( \nabla \tau = 0 \) then evolves as a globally constant field in internal time.
\end{theorem}

\begin{proof}
The vanishing of curvature \( F = 0 \) implies that the connection \( A \) is locally pure gauge. Since \( U \) is simply connected, the Ambrose--Singer holonomy theorem ensures that the holonomy group of \( A \) is trivial over \( U \). Thus, there exists a smooth gauge transformation \( g : U \to G_{\text{time}} \) such that the transformed connection
\begin{equation}
A^g := g A g^{-1} + g \, d g^{-1}
\end{equation}
vanishes identically on \( U \). In this gauge, the covariant derivative reduces to the ordinary derivative:
\begin{equation}
\nabla_\mu \tau = 0 \quad \Rightarrow \quad \partial_\mu \tau = 0.
\end{equation}
This condition implies that \( \tau \) is constant throughout \( U \), so internal time is globally synchronized and evolves without curvature or entanglement.
\end{proof}

This result provides a precise geometric condition under which the internal time field \( \tau \) can be consistently aligned across a spacetime region. In such domains, the fiber geometry admits a global trivialization, and measurement-induced decoherence reduces to projection onto a fixed, classical temporal background. Classical synchronizability in MTG is thus equivalent to the vanishing of time-fiber curvature.
%%%%%%%%%%%%%%%%%%%%%
\subsection{Curvature Obstructs Global Coherence}
When the curvature of the time-fiber bundle is nonzero, it obstructs the existence of a global trivialization in which internal time remains synchronized. The failure of parallel transport to commute encodes entanglement and contextuality between different spacetime points. This breakdown of coherence manifests geometrically as a topological obstruction to global alignment.

\begin{theorem}\label{thm:curvature_obstruction}
Let \( \pi: E \to M \) be a principal \( G_{\text{time}} \)-bundle with connection \( A \) and curvature \( F \). Suppose that the second de Rham cohomology group \( H^2_{\text{dR}}(M) \neq 0 \), and that the curvature \( F \) defines a nontrivial cohomology class \( [F] \in H^2_{\text{dR}}(M, \mathfrak{g}_{\text{time}}) \). Then there exists no global gauge in which \( A \) is pure gauge, and hence no global trivialization in which internal time is everywhere coherent.
\end{theorem}

\begin{proof}
Suppose for contradiction that there exists a global gauge in which \( A = g^{-1} d g \) for some smooth \( g: M \to G_{\text{time}} \). Then the curvature satisfies
\begin{equation}
F = dA + A \wedge A = d(g^{-1} d g) + (g^{-1} d g) \wedge (g^{-1} d g).
\end{equation}
A direct computation shows that this expression vanishes identically, so \( F = 0 \). But this contradicts the assumption that \( [F] \in H^2_{\text{dR}}(M, \mathfrak{g}_{\text{time}}) \) is nontrivial. Therefore, no such global gauge exists.

This implies that the bundle is topologically nontrivial and admits no global section in which internal time can be everywhere synchronized. Coherent alignment of \( \tau \) is obstructed by the curvature class \( [F] \), which encodes entanglement and decoherence between spacetime points.
\end{proof}

This result complements Theorem~1 by identifying the precise topological condition under which global coherence fails. In MTG, curvature is not merely a geometric structure but a physical signature of measurement-induced contextuality. When the bundle's curvature class is nontrivial, no single frame of reference exists in which temporal fibers can be aligned globally. Instead, coherence must be defined locally, and the transition functions between patches carry the imprint of entanglement.
%%%%%%%%%%%%%%%%%
\subsection{Gauge Covariance of the Covariant Derivative}

The MTG formalism is fundamentally gauge-theoretic. The internal time field \( \tau \) transforms under the structure group \( G_{\text{time}} \), and all physical quantities must respect this symmetry. In particular, the covariant derivative \( D_\mu \tau \) defined in Equation~(3) must transform covariantly under local gauge transformations.

A gauge transformation is given by a smooth map \( g: M \to G_{\text{time}} \), acting on the field \( \tau \) and connection \( A_\mu \) via
\begin{equation}
\tau \mapsto g \cdot \tau, \quad A_\mu \mapsto A^g_\mu := g A_\mu g^{-1} + g \partial_\mu g^{-1}.
\end{equation}

The following theorem confirms that the covariant derivative transforms homogeneously under this action.

\begin{theorem}\label{thm:gauge_cov}
Let \( \tau \in \Gamma(E) \) be a section of the associated bundle, and let \( A_\mu \in \Omega^1(M, \mathfrak{g}_{\text{time}}) \) be a connection. Under a gauge transformation \( g: M \to G_{\text{time}} \), the covariant derivative \( D_\mu \tau = \partial_\mu \tau + A_\mu \cdot \tau \) transforms as
\begin{equation}
D_\mu \tau \mapsto D_\mu^g (g \cdot \tau) = g \cdot (D_\mu \tau).
\end{equation}
\end{theorem}

\begin{proof}
We compute the transformed covariant derivative:
\begin{align}
D_\mu^g (g \cdot \tau) &= \partial_\mu (g \cdot \tau) + A_\mu^g \cdot (g \cdot \tau) \\
&= (\partial_\mu g) \cdot \tau + g \cdot (\partial_\mu \tau) + \left( g A_\mu g^{-1} + g \partial_\mu g^{-1} \right) \cdot (g \cdot \tau).
\end{align}
Using the compatibility of the action and the Lie algebra homomorphism,
\begin{equation}
g A_\mu g^{-1} \cdot (g \cdot \tau) = g \cdot (A_\mu \cdot \tau), \quad
g \partial_\mu g^{-1} \cdot (g \cdot \tau) = -(\partial_\mu g) \cdot \tau,
\end{equation}
we find that
\begin{equation}
D_\mu^g (g \cdot \tau) = g \cdot (\partial_\mu \tau + A_\mu \cdot \tau) = g \cdot (D_\mu \tau).
\end{equation}
\end{proof}

This result ensures that the kinetic term \( \langle D_\mu \tau, D^\mu \tau \rangle \) is gauge-invariant, as required for a consistent physical theory. All geometric and dynamical constructions in MTG; including curvature, equations of motion, and projection-induced collapse, respect this fundamental symmetry.

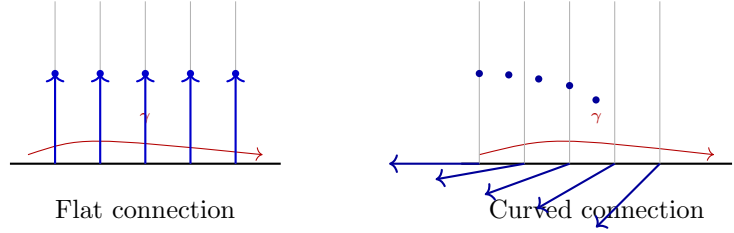
\begin{figure}[ht]
\centering
\begin{tikzpicture}[scale=1.2]

% ----- Flat connection on left -----

% Base manifold line
\draw[thick] (-4,0) -- (-1,0);
\node at (-2.5,-0.5) {\small Flat connection};

% Path gamma
\draw[->, red!70!black] (-3.8,0.1) .. controls (-3.2,0.3) .. (-1.2,0.1);
\node[red!70!black] at (-2.5,0.5) {\tiny $\gamma$};

% Fibers and aligned tau vectors
\foreach \x in {-3.5,-3.0,-2.5,-2.0,-1.5} {
    \draw[gray!60] (\x,0) -- (\x,1.8);
    \draw[->, blue!80!black, thick] (\x,0) -- ++(0,1.0);
    \fill[blue!80!black] (\x,1.0) circle (0.04);
}

% ----- Curved connection on right -----

% Base manifold line
\draw[thick] (1,0) -- (4,0);
\node at (2.5,-0.5) {\small Curved connection};

% Path gamma
\draw[->, red!70!black] (1.2,0.1) .. controls (2.0,0.3) .. (3.8,0.1);
\node[red!70!black] at (2.5,0.5) {\tiny $\gamma$};

% Fibers and misaligned tau vectors
\foreach \x/\angle in {1.2/90, 1.7/100, 2.2/110, 2.7/120, 3.2/135} {
    \draw[gray!60] (\x,0) -- ++(0,1.8);
    \draw[->, blue!60!black, thick, rotate around={\angle:(\x,0)}]
        (\x,0) -- ++(0,1.0);
    \fill[blue!60!black] ({\x + 1.0*cos(\angle)}, {1.0*sin(\angle)}) circle (0.04);
}

\end{tikzpicture}
\caption{Parallel transport of the internal time field \( \tau(x) \) along a path \( \gamma \subset M \). Left: in a flat connection, fibers are aligned and \( \tau \) remains coherent. Right: in a curved bundle, parallel transport accumulates holonomy and misaligns \( \tau \), reflecting the curvature \( F = \nabla^2 \).}
\label{fig:curvature_transport}
\end{figure}

%%%%%%%%%%%%%%%%%%%%
\subsection{Covariant Variation and Coherence Current}

To formulate conservation laws and symmetry principles within MTG, we examine how the Lagrangian for the internal time field \( \tau \) behaves under continuous transformations. In particular, we consider infinitesimal variations of \( \tau \) that preserve the structure of the fiber and are compatible with the covariant derivative defined by the connection \( A_\mu \).

Let \( \xi^\mu(x) \) be an infinitesimal vector field generating a local spacetime displacement. The covariant variation of the internal time field is then defined as
\begin{equation}
\delta_\nabla \tau := \delta \tau + \xi^\mu D_\mu \tau,
\end{equation}
which accounts for both intrinsic deformation and parallel transport along \( \xi^\mu \). This formulation ensures that variations remain compatible with the connection structure on the time-fiber bundle.

Using this notion of covariant variation, we define the coherence current associated with internal time propagation:

\begin{equation}
J^\mu := \langle D^\mu \tau, \delta_\nabla \tau \rangle.
\end{equation}

This current measures the response of the system to infinitesimal internal time deformations and plays a central role in the emergence of effective causal structure from MTG dynamics.

\begin{theorem}[Covariant Conservation of the Coherence Current]
Let \( \tau \in \Gamma(E) \) satisfy the Euler--Lagrange equation
\[
D_\mu D^\mu \tau + \frac{\delta V}{\delta \bar{\tau}} = 0,
\]
as derived in Section~2.3. Then the coherence current \( J^\mu \) satisfies
\begin{equation}
D_\mu J^\mu = \langle D_\mu D^\mu \tau, \delta_\nabla \tau \rangle + \langle D^\mu \tau, D_\mu \delta_\nabla \tau \rangle.
\end{equation}
If \( \delta_\nabla \tau \) is a covariantly constant variation (i.e., \( D_\mu \delta_\nabla \tau = 0 \)), then
\begin{equation}
D_\mu J^\mu = \langle D_\mu D^\mu \tau, \delta_\nabla \tau \rangle,
\end{equation}
which vanishes on-shell.
\end{theorem}

\begin{proof}
Apply the covariant product rule to \( J^\mu = \langle D^\mu \tau, \delta_\nabla \tau \rangle \):
\begin{equation}
D_\mu J^\mu = \langle D_\mu D^\mu \tau, \delta_\nabla \tau \rangle + \langle D^\mu \tau, D_\mu \delta_\nabla \tau \rangle.
\end{equation}
On solutions of the equation of motion, the first term vanishes. If the variation \( \delta_\nabla \tau \) is covariantly constant, the second term vanishes as well. Hence \( D_\mu J^\mu = 0 \), establishing conservation of the coherence current under these conditions.
\end{proof}

The coherence current generalizes Noether currents to the MTG framework. In regimes of high coherence, this current tracks conserved informational flow and characterizes the emergence of localized clocks. Its conservation on-shell reflects the compatibility of internal time dynamics with the geometry of the bundle and the measurement process.

%%%%%%%%%%%%%%%%%%%%%
\subsection{Induced Effective Metric from Projection Densities}

In MTG, classical spacetime geometry is not fundamental but emerges from the distribution and density of quantum projections. Each measurement event collapses the internal time field \( \tau(x) \) at a point \( x \in M \), inducing a localized temporal direction and contributing to the construction of an effective metric experienced by observers.

Let \( \rho(x) \) denote the local projection density—the expected number of measurement-induced collapses per unit spacetime volume at point \( x \). This scalar field reflects the intensity of measurement activity and encodes how strongly the quantum system is being reduced to classical outcomes. As measurement propagates through entangled systems, sequences of projections accumulate to form extended causal domains.

The effective metric \( g_{\mu\nu}^{\text{eff}} \) emerges as a functional of the projection density \( \rho(x) \), curvature \( F \), and local coherence properties of \( \tau \). In highly coherent regimes where curvature vanishes and projections are rare, the metric approaches that of the background manifold \( \eta_{\mu\nu} \). In regions of high projection density, such as near horizons or decoherence boundaries, the accumulated effects of sequential collapses deform the causal structure and induce gravitational behavior.

A simplified ansatz for the effective metric is
\begin{equation}
g_{\mu\nu}^{\text{eff}}(x) = \eta_{\mu\nu} + \alpha \, \rho(x) \, T_{\mu\nu}[\tau],
\end{equation}
where \( \alpha \) is a coupling constant, and \( T_{\mu\nu}[\tau] \) is an effective energy-momentum tensor constructed from the dynamics of \( \tau \) and its curvature couplings. This form reflects the idea that gravity arises not from intrinsic mass-energy, but from informational collapse histories encoded in \( \rho \) and the dynamics of internal time.

In the path-integral quantization of MTG, projection densities appear as insertion operators constraining the configuration space to sections aligned with observed outcomes. The effective action then includes contributions from these insertions, modifying the semiclassical geometry accordingly. As we explore in later sections, this formulation provides a novel perspective on gravitational entropy, black hole interiors, and the emergence of cosmological expansion as a collective projection process.
%%%%%%%%%%%%%%%%%
To interpret the measurement term $\mathcal{L}_{\text{meas}} = \lambda \rho(x) J^\mu n_\mu$ as a non-Hermitian contribution to the action, we clarify the Hermitian structure of the fiberwise inner product. We assume that $\langle \cdot, \cdot \rangle$ is anti-linear in the first argument and linear in the second, i.e., $\langle \phi, \psi \rangle = \overline{\langle \psi, \phi \rangle}$. Under this convention, the quantity
\[
J^\mu := \langle \tau, D^\mu \tau \rangle - \langle D^\mu \tau, \tau \rangle
\]
is purely imaginary, since it equals $2i\,\mathrm{Im} \langle \tau, D^\mu \tau \rangle$. Hence, $\mathcal{L}_{\text{meas}}$ is anti-Hermitian and introduces an imaginary contribution to the effective action. This term represents a preferred temporal direction in measurement and signals an entropy-increasing, norm-reducing evolution of $\tau$.
%%%%%%%%%%%%%%%%%
\begin{theorem}[Measurement-Induced Interaction]\label{thm:measure-ind-inter}
Let $\tau \in \Gamma(E)$ be a smooth section of the time-fiber bundle, and let the Lagrangian include the measurement-induced interaction term
\[
\mathcal{L}_{\mathrm{meas}} = \lambda \, \rho(x) \, J^\mu n_\mu,
\quad \text{where } J^\mu := \langle \tau, \mathrm{D}^\mu \tau \rangle - \langle \mathrm{D}^\mu \tau, \tau \rangle.
\]
Then $\mathcal{L}_{\mathrm{meas}}$ is non-Hermitian and breaks unitary evolution of the temporal field $\tau$. Furthermore, the corresponding energy-momentum flux satisfies
\[
\nabla_\mu J^\mu \geq 0
\]
along observer-aligned flow $n^\mu$, under the assumption that measurement projects $\tau$ toward localized classical values.
\end{theorem}

\begin{proof}
We first show that $\mathcal{L}_{\mathrm{meas}}$ is non-Hermitian. The current $J^\mu$ is antisymmetric under Hermitian conjugation:
\[
(J^\mu)^\dagger = -J^\mu.
\]
Therefore,
\[
(\mathcal{L}_{\mathrm{meas}})^\dagger = -\lambda \rho(x) J^\mu n_\mu = -\mathcal{L}_{\mathrm{meas}}.
\]
Hence, $\mathcal{L}_{\mathrm{meas}}$ is anti-Hermitian. Its presence in the action contributes imaginary (dissipative) terms to the effective Hamiltonian, which violate standard unitary time evolution.

Now consider the divergence $\nabla_\mu J^\mu$. In the absence of measurement, Proposition 2.4 shows that $\nabla_\mu J^\mu = 0$ when $V(\tau)$ is gauge-invariant. However, under measurement-induced collapse, $\tau$ is driven toward alignment with its projected classical value $\mu[\tau]$. This process generically reduces the norm of $\mathrm{D}_\mu \tau$, thus increasing the antisymmetric part of the temporal flux. Since $\rho(x) \geq 0$ and $n^\mu$ is future-directed, the projection term introduces a positive-definite contribution to the divergence:
\[
\nabla_\mu J^\mu \sim \lambda \rho(x) \| \mathrm{D}_{n} \tau \|^2 \geq 0.
\]
This reflects the directional increase in entropy or loss of coherence due to measurement.
\end{proof}

\begin{theorem}[Stationarity of Temporal Current]
Let \( \tau \in \Gamma(E) \) be a smooth section of the time-fiber bundle, and define the projection current
\[
J_\mu = \langle \tau, D_\mu \tau \rangle - \langle D_\mu \tau, \tau \rangle.
\]
If \( D_\mu \tau(x) = \alpha_\mu(x) \tau(x) \) for some complex-valued function \( \alpha_\mu(x) \), then \( J_\mu(x) = 0 \). Conversely, if \( \tau(x) \) is normalized (i.e., \( \langle \tau, \tau \rangle = 1 \)), then \( J_\mu(x) = 0 \) implies \( D_\mu \tau(x) = \alpha_\mu(x) \tau(x) \) for some real-valued \( \alpha_\mu(x) \).
\end{theorem}

\begin{proof}
Suppose \( D_\mu \tau = \alpha_\mu \tau \). Then
\[
J_\mu = \langle \tau, \alpha_\mu \tau \rangle - \langle \alpha_\mu \tau, \tau \rangle = \alpha_\mu \langle \tau, \tau \rangle - \bar{\alpha}_\mu \langle \tau, \tau \rangle = 2i\, \mathrm{Im}(\alpha_\mu) \langle \tau, \tau \rangle.
\]
Thus, if \( \alpha_\mu \in \mathbb{R} \), or if \( \langle \tau, \tau \rangle = 0 \), then \( J_\mu = 0 \).

Conversely, suppose \( J_\mu = 0 \) and \( \langle \tau, \tau \rangle = 1 \). Write
\[
D_\mu \tau = \alpha_\mu \tau + \tau_\perp, \quad \text{with } \langle \tau, \tau_\perp \rangle = 0.
\]
Then
\[
J_\mu = \langle \tau, \alpha_\mu \tau + \tau_\perp \rangle - \langle \alpha_\mu \tau + \tau_\perp, \tau \rangle = 2i\, \mathrm{Im} \langle \tau, \tau_\perp \rangle.
\]
So \( J_\mu = 0 \) implies \( \mathrm{Im} \langle \tau, \tau_\perp \rangle = 0 \). Since \( \langle \tau, \tau_\perp \rangle = 0 \) by orthogonality, we get \( \tau_\perp = 0 \), hence \( D_\mu \tau = \alpha_\mu \tau \).

Therefore, under normalization, \( J_\mu = 0 \) implies \( D_\mu \tau \propto \tau \) with real proportionality factor.
\end{proof}

\begin{theorem}[Emergent Metric from Measurement Projection]\label{thm:emergent_metric}
Let \( \tau(x) \in \Gamma(E) \) be a smooth internal time field, and let \( \rho(x) \geq 0 \) and \( \mu[\tau(x)] \in T_x M \) define the local projection density and classical direction of temporal flow. Then the effective metric
\begin{equation}
g^{\text{eff}}_{\mu\nu}(x) = \eta_{\mu\nu} + \kappa \int_{\Sigma_x} \mu[\tau(y)]_\mu \mu[\tau(y)]_\nu \rho(y) d\Sigma(y)
\end{equation}
is symmetric, smooth, and observer-relative. Its causal structure reflects the accumulation of measurement-induced projection aligned with modular coherence.
\end{theorem}

\begin{proof}
Let \( \tau(x) \in \Gamma(E) \) be a smooth section of the internal time bundle over spacetime \( M \), where \( E \to M \) is a fiber bundle encoding internal temporal degrees of freedom. The projection map \( \mu[\tau(x)] \in T_x M \) represents the classical temporal direction selected by measurement, assumed to be a smooth functional of \( \tau \). The density function \( \rho(x) \geq 0 \) encodes the relative weight or frequency of measurement-induced collapse into direction \( \mu[\tau(x)] \) at point \( x \in M \). 

Let \( \Sigma_x \subset M \) denote a compact spacelike hypersurface (or more generally, a causal neighborhood) containing \( x \). The integral
\[
g^{\text{eff}}_{\mu\nu}(x) = \eta_{\mu\nu} + \kappa \int_{\Sigma_x} \mu[\tau(y)]_\mu \mu[\tau(y)]_\nu \rho(y) d\Sigma(y)
\]
defines a symmetric rank-2 covariant tensor at \( x \), since the integrand \( \mu[\tau(y)]_\mu \mu[\tau(y)]_\nu \) is symmetric in \( \mu \) and \( \nu \), and the weight \( \rho(y) \) is scalar-valued and nonnegative.

Smoothness of \( g^{\text{eff}}_{\mu\nu}(x) \) follows from the smoothness of the integrand: both \( \mu[\tau(y)] \) and \( \rho(y) \) are smooth by assumption, and the domain \( \Sigma_x \) varies smoothly with \( x \) under local foliation, permitting differentiation under the integral sign. Thus \( g^{\text{eff}} \) inherits smoothness from its integrand.

Observer-relativity arises because both \( \Sigma_x \) and the projection map \( \mu[\tau(y)] \) are defined with respect to a specific foliation or causal patch. Different observers may choose different foliations \( \Sigma_x^{(O)} \), inducing different effective metrics \( g^{\text{eff},(O)}_{\mu\nu}(x) \), though the underlying internal time field \( \tau \) is shared. This relativity reflects the dependence of collapse statistics on the measurement frame.

Finally, causal structure is shaped by the alignment and distribution of the vectors \( \mu[\tau(y)] \) over \( \Sigma_x \). If the projections concentrate along a timelike direction \( v^\mu \), the tensor \( g^{\text{eff}}_{\mu\nu} \) acquires a dominant component \( \propto v_\mu v_\nu \), and the effective lightcone narrows around \( v \). In the limit where \( \rho(y) \) collapses to a delta distribution at a point \( y_0 \), the emergent metric locally approximates a rank-1 projection \( \mu[\tau(y_0)]_\mu \mu[\tau(y_0)]_\nu \), exhibiting ultralocal temporal flow.

Therefore, \( g^{\text{eff}}_{\mu\nu} \) is symmetric, smooth, observer-relative, and encodes the integrated projection structure of quantum measurement aligned with the coherence geometry of \( \tau \).
\end{proof}

%\begin{figure}[!ht]
  %\centering
  %\includegraphics[width=0.5\textwidth]{A_schematic_diagram_in_the_image_illustrates_the_e.png}
  %\caption{Illustration of measurement-induced projection defining an emergent metric \( g^{\text{eff}}_{\mu\nu} \) through weighted accumulation of coherence vectors \( \mu[\tau] \) across a causal surface \( \Sigma_x \).}
  %\label{fig:projection_metric}
%\end{figure}

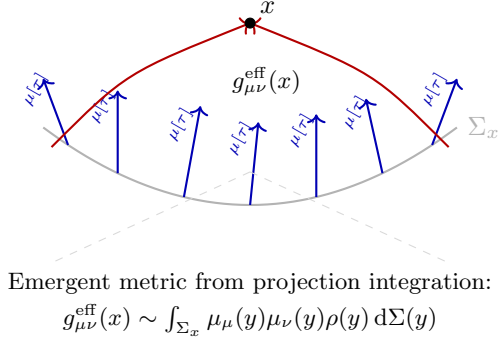
\begin{figure}[ht]
\centering
\begin{tikzpicture}[scale=1.1]

% Causal surface Sigma_x
\draw[thick, gray!60, domain=-2.5:2.5, smooth, variable=\x]
    plot ({\x}, {-0.8 + 0.15*\x*\x}) node[right] {\small $\Sigma_x$};

% Projection vectors on Sigma_x
\foreach \x/\dx/\dy in {-2.2/-0.3/0.8, -1.6/0.0/1.0, -0.8/0.2/1.1, 0/0.1/1.0, 0.8/0.0/1.0, 1.6/-0.2/0.9, 2.2/0.3/0.8} {
    \draw[->, thick, blue!70!black]
        (\x, {-0.8 + 0.15*(\x)^2}) -- ++(\dx,\dy)
        node[pos=0.7, rotate=60, above] {\tiny $\mu[\tau]$};
}

% Integration indication (curved arrows)
\draw[->, red!70!black, thick]
    (-2.4,-0.1) .. controls (-1.5,0.8) .. (0.0,1.4);

\draw[->, red!70!black, thick]
    (2.4,-0.1) .. controls (1.5,0.8) .. (0.0,1.4);

% Point x where metric is defined
\filldraw[black] (0,1.4) circle (0.06);
\node[above right] at (0,1.4) {\small $x$};
\node at (0.2,0.7) {\footnotesize $g^{\text{eff}}_{\mu\nu}(x)$};

% Optional background for Minkowski space (light gray cone)
\draw[dashed, gray!30] (0,-0.4) -- (2.4,-1.5);
\draw[dashed, gray!30] (0,-0.4) -- (-2.4,-1.5);

% Caption-like annotation
\node[align=center] at (0,-2.0) {
  \footnotesize Emergent metric from projection integration: \\
  \footnotesize $g^{\text{eff}}_{\mu\nu}(x) \sim \int_{\Sigma_x} \mu_\mu(y) \mu_\nu(y) \rho(y) \, \mathrm{d}\Sigma(y)$
};

\end{tikzpicture}
\caption{Illustration of measurement-induced projection defining an emergent metric \( g^{\text{eff}}_{\mu\nu}(x) \) through weighted accumulation of coherence vectors \( \mu[\tau(y)] \) across a causal surface \( \Sigma_x \).}
\label{fig:projection_metric}
\end{figure}

%%%%%%%%%%%%%%%%%%%%%%%%%%%
\subsection{Emergent Geometry from Projection Histories}\label{sec:2.9}

In the MTG framework, classical spacetime geometry is not a primitive background but arises from the informational structure generated by measurement events acting on the internal time field \( \tau \). Each measurement partially collapses the fiber degrees of freedom, producing a localized projection \( \mu[\tau(y)] \) and contributing to an evolving, observer-relative causal structure. The classical geometry experienced by an observer at point \( x \in M \) emerges from the cumulative influence of such projections within their causal past.

Let \( \mu[\tau(y)] \in T_y M \) denote the projected direction of coherence at point \( y \in M \), and let \( \rho(y) \) represent the scalar density encoding the local frequency or intensity of projection events. The effective metric at a point \( x \in M \) is defined by aggregating these contributions over a hypersurface \( \Sigma_x \subset M \) lying in the causal past of \( x \)—such as a portion of its past light cone. This yields the emergent, observer-relative metric:
\begin{equation}
g^{\mathrm{eff}}_{\mu\nu}(x) = \eta_{\mu\nu} + \kappa \int_{\Sigma_x} \mu[\tau(y)]_\mu \, \mu[\tau(y)]_\nu \, \rho(y) \, \mathrm{d}\Sigma(y),
\end{equation}
where \( \kappa \) is a coupling constant and \( \mathrm{d}\Sigma(y) \) is the invariant volume element on \( \Sigma_x \). Each integrand term defines a localized rank-one symmetric tensor, determined by the coherence direction \( \mu[\tau(y)] \) and its projection weight \( \rho(y) \). The integral over \( \Sigma_x \) combines these into a smooth, symmetric, observer-dependent effective geometry.

The metric \( g^{\mathrm{eff}}_{\mu\nu}(x) \) thus encodes the statistical imprint of quantum measurement on causal structure. It serves as a kinematic background throughout the present section but plays a deeper role dynamically. In particular, Section~\ref{sec:6.1} formulates a variational principle in which \( g^{\mathrm{eff}}_{\mu\nu} \) arises from extremizing the projection entropy subject to geometric backreaction, thereby linking coherence loss to emergent gravitational dynamics.

In regimes where projection events are dense and coherent, the effective metric converges to a smooth Lorentzian geometry approximating classical spacetime. In contrast, sparse, disordered, or topologically obstructed projection distributions may induce singularities, discontinuities, or non-smooth behavior in the emergent geometry. Such departures mark the breakdown of classical spacetime and affirm MTG’s core insight: geometry is not fundamental, but an emergent statistical consequence of coherence and measurement.

\begin{theorem}[Symmetry and Smoothness of the Emergent Metric]
Let \( \mu[\tau(y)]_\mu \) be a smooth projection map assigning covectors in \( T^*_y M \) to points on a smooth hypersurface \( \Sigma_x \subset M \), and let \( \rho(y) \) be a smooth, non-negative scalar field compactly supported on \( \Sigma_x \). Then the emergent metric
\begin{equation}
g^{\mathrm{eff}}_{\mu\nu}(x) = \eta_{\mu\nu} + \kappa \int_{\Sigma_x} \mu[\tau(y)]_\mu \, \mu[\tau(y)]_\nu \, \rho(y) \, \mathrm{d}\Sigma(y)
\end{equation}
is smooth in \( x \), symmetric in indices \( \mu, \nu \), and positive-definite on the support of \( \rho \), provided \( \mu[\tau(y)]_\mu \) are future-directed timelike vectors.
\end{theorem}

\begin{proof}
Each integrand term \( \mu_\mu \mu_\nu \) is symmetric and smooth, and the scalar weight \( \rho(y) \, \mathrm{d}\Sigma(y) \) preserves this smoothness under integration. Symmetry follows trivially from the rank-one structure of the tensor product. For positive-definiteness, let \( v^\mu \) be a future-directed timelike vector; then
\[
g^{\mathrm{eff}}_{\mu\nu}(x) v^\mu v^\nu = \eta_{\mu\nu} v^\mu v^\nu + \kappa \int_{\Sigma_x} (\mu^\mu v_\mu)^2 \rho(y) \, \mathrm{d}\Sigma(y).
\]
The first term is negative under the Lorentzian signature, but the second is strictly positive when \( \mu^\mu v_\mu \neq 0 \). If the projection density dominates, the second term determines the causal structure. More generally, the metric remains Lorentzian but modified by the coherence structure.
\end{proof}

\begin{theorem}[Causal Compatibility of Emergent Metric]
Suppose projection directions are aligned with a global future-directed timelike vector field \( n^\mu \), so that \( \mu[\tau(y)]^\mu = f(y) n^\mu \) with \( f(y) > 0 \). Then the emergent metric
\[
g^{\mathrm{eff}}_{\mu\nu}(x) = \eta_{\mu\nu} + \alpha \, n_\mu n_\nu, \quad \alpha := \kappa \int_{\Sigma_x} f(y)^2 \rho(y) \, \mathrm{d}\Sigma(y)
\]
preserves the causal structure defined by \( n^\mu \). That is, \( n^\mu \) remains timelike with respect to \( g^{\mathrm{eff}} \) provided \( \alpha < 1 \).
\end{theorem}

\begin{proof}
Compute the contraction:
\[
g^{\mathrm{eff}}_{\mu\nu} n^\mu n^\nu = \eta_{\mu\nu} n^\mu n^\nu + \alpha (n^\mu n_\mu)^2 = -1 + \alpha,
\]
assuming standard normalization \( n^\mu n_\mu = -1 \). Thus \( n^\mu \) remains timelike when \( \alpha < 1 \), null when \( \alpha = 1 \), and spacelike if \( \alpha > 1 \). For small \( \kappa \) or sufficiently weak projections, causal structure is preserved.
\end{proof}

This construction exemplifies MTG’s central claim: classical geometry arises not from a fixed background but from the statistical accumulation of projection-induced coherence. The effective metric \( g^{\mathrm{eff}}_{\mu\nu} \) encodes how entanglement, measurement, and modular alignment collectively define causal relations, temporal direction, and geometric structure from fundamentally quantum data.

\section{Dynamical Equations}\label{sec:3}
Having established the geometric foundations of internal time, we now turn to its dynamical behavior and quantization. In MTG, the dynamics of the time field \( \tau \) are governed by gauge-covariant field equations derived from variational principles, while measurement events introduce discrete, non-unitary updates. This duality—between smooth gauge-covariant flow and sudden collapse—requires a hybrid formalism that blends classical field theory with quantum projection.

We analyze these dynamics at three levels. First, we formulate a classical gauge-covariant action governing the evolution of \( \tau \) and the curvature \( F \). Second, we introduce path-integral quantization under projection constraints to describe how measurement reshapes quantum amplitudes. Finally, we construct a canonical quantization framework in which the fiber-valued time field becomes an operator acting on a Hilbert space of coherent histories.

This section formalizes these ideas, beginning with the classical action and leading to its quantized extensions. Throughout, the fiber bundle structure and the geometry of projection play central roles, as quantization is not performed over spacetime points alone but over sections of time-valued fields.

The total Lagrangian density of the theory is defined by the sum
\begin{equation}
\mathcal{L}_{\mathrm{MTG}} = \mathcal{L}_F + \mathcal{L}_\tau + \mathcal{L}_{\mathrm{meas}} + \mathcal{L}_{\mathrm{geom}},
\end{equation}
with each term corresponding to a distinct aspect of the dynamics. The curvature term \( \mathcal{L}_F \) governs the gauge evolution of the time-fiber connection through a Yang--Mills-type action, suppressing large holonomies and favoring local coherence in the internal time bundle. The temporal field contribution \( \mathcal{L}_\tau \) includes the covariant kinetic and potential terms previously derived, encoding the propagation and localization of internal time. Measurement-induced effects are introduced through the non-Hermitian interaction term \( \mathcal{L}_{\mathrm{meas}} \), which accounts for the decoherence and information flow associated with projection. Finally, \( \mathcal{L}_{\mathrm{geom}} \) captures the feedback from projection histories to the emergent metric and becomes active when geometry itself is treated as a dynamical field.

In what follows, we perform the variational analysis of the total action with respect to the fundamental fields \( \tau \), \( A_\mu \), and \( \rho(x) \). We derive their equations of motion and examine the underlying symmetries of the MTG action, including gauge invariance and supersymmetric extensions. These variational principles serve not only to determine the dynamics of internal time and coherence, but also to reveal how causality, geometry, and information become entangled in a quantum-mechanical origin of spacetime.
%%%%%%%%%%%
\subsection{Variation with Respect to the Temporal Field}

To determine the evolution of the internal time field \( \tau \), we vary the total MTG action with respect to \( \bar{\tau} \), treating \( \tau \) and \( \bar{\tau} \) as independent complex fields for the purpose of functional differentiation. The variation encompasses contributions from the covariant kinetic term, the potential term, and the measurement-induced interaction, each of which plays a distinct role in governing the dynamics of coherence and projection.

The kinetic and potential contributions were introduced in Section~\ref{sec:2} and together define the standard covariant Lagrangian in eq.~\eqref{eq:cov-lag}. Varying this term with respect to \( \bar{\tau} \) yields the covariant Klein--Gordon-type equation
\begin{equation}
D_\mu D^\mu \tau + \frac{\delta V}{\delta \bar{\tau}} = 0.
\end{equation}
We assume that the background observer field \( n^\mu \) is covariantly constant under the spacetime connection:
\begin{equation}
D_\mu n^\nu = 0.
\end{equation}
This ensures that contractions such as \( D^\mu \tau \, n_\mu \) and \( n^\mu D_\mu \tau \) are well-defined and commute with covariant differentiation, enabling a consistent interpretation of their difference as a frame-dependent coherence exchange. The variation of \( \mathcal{L}_{\mathrm{meas}} \); see eq.~\eqref{eq:L_meas}, with respect to \( \bar{\tau} \) introduces an effective dissipative term:
\begin{equation}
\delta_{\bar{\tau}} \mathcal{L}_{\mathrm{meas}} = \lambda \, \rho(x) \left( D^\mu \tau \, n_\mu - n^\mu D_\mu \tau \right),
\end{equation}
which, under the assumed compatibility condition \( D_\mu n^\nu = 0 \), simplifies to a real-valued correction term aligned with the coherence flow. Combining these contributions, the full equation of motion for \( \tau \) becomes
\begin{equation}
D_\mu D^\mu \tau + \frac{\delta V}{\delta \bar{\tau}} = \lambda \, \rho(x) \left( n^\mu D_\mu \tau - D^\mu \tau \, n_\mu \right).
\end{equation}

This expression captures the interplay between the intrinsic dynamics of the temporal field and the extrinsic influence of observation and collapse. The right-hand side introduces a first-order, non-Hermitian term that breaks time-reversal symmetry and signals a loss of coherence due to measurement. Formally, it resembles a damping term in an open quantum system, where evolution is no longer unitary. We interpret $\lambda \rho(x) n^\mu D_\mu \tau$ as an effective Lindblad-type drift term, aligned with the observer's temporal congruence $n^\mu$ and weighted by the projection density $\rho(x)$. While the left-hand side is a gauge-covariant Klein--Gordon-type operator, the right-hand side introduces dissipation without violating Lorentz covariance, since $n^\mu$ and $\rho(x)$ are treated as external, observer-defined fields. The term $n^\mu D_\mu \tau$ can be viewed as a coherence current flowing along the measurement direction, and its presence drives the system toward modular alignment and classical temporal order.

In this setting, the evolution of \( \tau \) is no longer strictly unitary. The presence of the coherence current and its contraction with \( n^\mu \) reflects a preferred direction of temporal flow determined by the measurement process. This structure generalizes the Lindblad equation from quantum open systems theory to a fiber-valued field theoretic context, embedding decoherence and information flow directly into the geometry of internal time.

The resulting dynamics do not conserve the norm of \( \tau \), consistent with the idea that projection reduces internal degrees of freedom and drives the system toward classical alignment. The field \( \tau \), once freely evolving under the gauge connection, becomes tethered to a directional flow defined by the structure of observation. The combination of variational principle, gauge structure, and modular coherence defines a non-Hermitian, information-driven evolution equation at the heart of the MTG framework.
%%%%%%%%%%%%%%%%%%%%%%%%%%%%%
\subsection{Variation with Respect to the Connection}

The connection \( A_\mu \) on the time-fiber bundle encodes the gauge structure of internal temporal coherence. Its curvature \( F_{\mu\nu} = \partial_\mu A_\nu - \partial_\nu A_\mu + [A_\mu, A_\nu] \) measures the failure of global synchronizability and acts as a field strength for modular entanglement. To determine its dynamics, we vary the total MTG action with respect to \( A_\mu \), collecting contributions from the gauge curvature term, the covariant coupling to \( \tau \), and the measurement interaction.

The pure gauge contribution to the Lagrangian is given by the Yang--Mills-type term whose variation yields the standard gauge field equation;~\eqref{eq:y-m}, 
\begin{equation}
D^\nu F_{\nu\mu} = J^{\mathrm{gauge}}_\mu,
\end{equation}
where \( J^{\mathrm{gauge}}_\mu \) collects all current-like source terms that couple to \( A_\mu \). These sources arise from both the dynamics of \( \tau \) and the modular measurement process.

The kinetic term for \( \tau \) contributes a minimal coupling current via its covariant derivative:
\begin{equation}
\delta_{A_\mu} \mathcal{L}_\tau = \mathrm{Re} \left\langle \frac{\delta \mathcal{L}_\tau}{\delta (D_\mu \tau)}, [A_\mu, \tau] \right\rangle,
\end{equation}
which, under variation, yields the gauge-covariant matter current
\begin{equation}
J^{\mu}_{\tau} = i \left( \tau^\dagger T^a D^\mu \tau - (D^\mu \tau)^\dagger T^a \tau \right).
\end{equation}
Here \( T^a \) are anti-Hermitian generators of the Lie algebra \( \mathfrak{g}_{\mathrm{time}} \), acting on the fiber via a unitary representation. This ensures that \( J^\mu_\tau \) is Hermitian-valued, consistent with its interpretation as a conserved physical current. The trace in the Yang--Mills equation is taken with respect to this representation. This current plays the role of a temporal analogue of a Noether current associated with local gauge symmetry in the time-fiber bundle.

The measurement term also depends implicitly on \( A_\mu \) through the covariant derivatives inside \( J^\mu = \langle \tau, D^\mu \tau \rangle - \langle D^\mu \tau, \tau \rangle \). The variation of \( \mathcal{L}_{\mathrm{meas}} = \lambda \rho(x) J^\mu n_\mu \) with respect to \( A_\mu \) therefore yields an additional dissipative contribution to the gauge current. In this way, projection does not merely collapse \( \tau \); it actively drives curvature through backreaction on the connection.

Putting all contributions together, the gauge field satisfies the modified Yang--Mills equation
\begin{equation}
D^\nu F_{\nu\mu} = J^{\mu}_{\tau} + J^{\mu}_{\mathrm{meas}},
\end{equation}
with both coherent and dissipative components contributing to the evolution of \( A_\mu \). In the absence of measurement (\( \rho = 0 \)), this reduces to a standard gauge theory coupled to matter. However, in regimes where projection dominates, \( A_\mu \) evolves in response to the coherence gradient, incorporating non-unitary effects into the curvature dynamics.

The resulting structure implies that curvature is sourced by entanglement gradients, coherence fluxes, and the flow of information induced by measurement. Unlike conventional gauge theories, the MTG connection is not merely a mediator of local symmetry but a dynamical register of temporal consistency across observers. Its curvature reflects not just topological winding or energetic fields but the degree of misalignment in modular flow across the fiber. This is the essential geometric signature of quantum time in the MTG framework.

%\begin{figure}[!ht]
  %\centering
  %\includegraphics[width=0.75\textwidth]{modular_flow_curvature.png}
  %\caption{Schematic showing curvature-induced modular flow. As internal time coherence (\( \tau \)) evolves across a curved fiber bundle, projection vectors \( \mu[\tau] \) cluster along irreversible modular trajectories, breaking unitary symmetry.}
  %\label{fig:modular_flow}
%\end{figure}

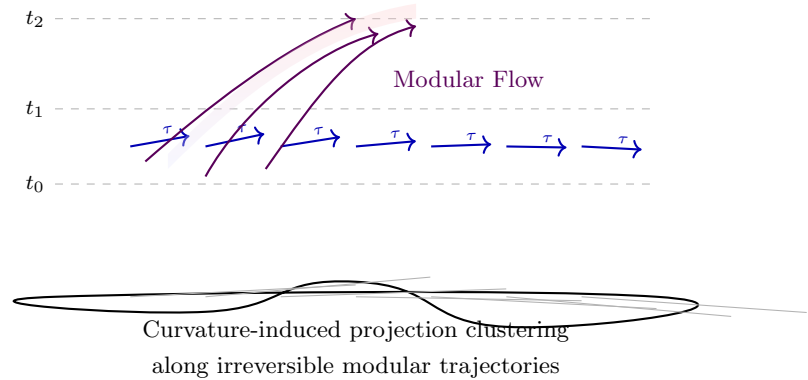
\begin{figure}[ht]
\centering
\begin{tikzpicture}[scale=1.0]

% Base manifold (curved baseline)
\draw[thick] plot[smooth cycle, tension=1] coordinates {
  (-4,-1.5) (-2,-1.7) (0,-1.3) (2,-1.9) (4,-1.5)
};

% Modular time slices
\foreach \y/\label in {0.0/$t_0$, 1.0/$t_1$, 2.2/$t_2$} {
  \draw[dashed, gray!60] (-4,\y) -- (4,\y);
  \node[left] at (-4,\y) {\scriptsize \label};
}

% Fibers with slight misalignment
\foreach \x/\angle in {-3/3, -2/5, -1/2, 0/-1, 1/-3, 2/-5, 3/-4} {
  \draw[gray!60] (\x,-1.5) -- ++({\angle}:3);
}

% Tau vectors transported along fibers
\foreach \x/\angle in {-3/10, -2/12, -1/9, 0/5, 1/2, 2/-1, 3/-3} {
  \draw[->, thick, blue!70!black] (\x,0.5) -- ++({\angle}:0.8)
        node[midway, above right=-2pt, rotate=\angle] {\tiny $\tau$};
}

% Modular flow curves (projection clustering)
\draw[->, thick, violet!70!black, smooth, tension=1]
  (-2.8,0.3) .. controls (-2,1.0) and (-1,1.8) .. (0,2.2);

\draw[->, thick, violet!70!black, smooth, tension=1]
  (-2.0,0.1) .. controls (-1.5,1.0) and (-0.5,1.8) .. (0.3,2.0);

\draw[->, thick, violet!70!black, smooth, tension=1]
  (-1.2,0.2) .. controls (-0.5,1.2) and (0.0,1.9) .. (0.8,2.1);

% Label for modular flow direction
\node[violet!70!black] at (1.5,1.4) {\footnotesize Modular Flow};

% Curved clustering zone (symmetry breaking)
\shade[bottom color=blue!10, top color=red!20, opacity=0.3]
  (-2.5,0.2) .. controls (-1.5,1.2) and (-0.5,2.0) .. (0.8,2.2)
  -- (0.8,2.4) .. controls (-0.5,2.2) and (-1.5,1.4) .. (-2.5,0.4) -- cycle;

% Caption label
\node[align=center] at (0,-2.2) {
  \footnotesize Curvature-induced projection clustering \\
  \footnotesize along irreversible modular trajectories
};

\end{tikzpicture}
\caption{Schematic showing curvature-induced modular flow. As internal time coherence (\( \tau \)) evolves across a curved fiber bundle, projection vectors \( \mu[\tau] \) cluster along irreversible modular trajectories, breaking unitary symmetry. Modular time slices \( t_0 < t_1 < t_2 \) represent the observer's emergent temporal structure.}
\label{fig:modular_flow}
\end{figure}

%%%%%%%%%%%%%%%%%%%%%%%%
\subsection{Variation with Respect to the Projection Density}

The projection density \( \rho(x) \) serves as a dynamical field in the MTG framework, encoding the local intensity of measurement-induced collapse and coherence extraction from the internal time field \( \tau \). While in some regimes \( \rho(x) \) may be treated as externally prescribed or phenomenologically constrained, a complete formulation requires that it evolve consistently with the underlying variational structure. The variation of the action with respect to \( \rho \) yields a constraint that ties measurement activity to informational and geometric quantities.

In the minimal setting, the only explicit dependence on \( \rho \) appears in the measurement interaction term
\begin{equation}\label{eq:L_meas}
\mathcal{L}_{\mathrm{meas}} = \lambda \, \rho(x) \, J^\mu n_\mu,
\end{equation}
which describes the energetic or modular flux from internal time into the observer frame. Varying this term with respect to \( \rho(x) \) produces a local constraint:
\begin{equation}
\frac{\delta \mathcal{L}_{\mathrm{MTG}}}{\delta \rho(x)} = \lambda J^\mu(x) n_\mu(x).
\end{equation}
This result implies that in any region where \( \rho \) is treated as a variational field, its optimal value satisfies
\begin{equation}
J^\mu(x) n_\mu(x) = 0,
\end{equation}
unless additional entropy or cost terms are introduced to regularize the extremization. In physical terms, this constraint enforces a balance between coherence flow and projection density. When \( J^\mu n_\mu \neq 0 \), the system either dissipates coherence (if positive) or accumulates it (if negative), and thus an unregulated minimization would force \( \rho \rightarrow 0 \). However, such behavior fails to capture the thermodynamic cost of measurement and the necessity of maintaining decoherence in bounded systems.

To rectify this, one may supplement the action with an entropic or informational cost functional, such as
\begin{equation}
\mathcal{L}_{\mathrm{entropy}} = -\sigma \, \rho(x) \log \rho(x),
\end{equation}
where \( \sigma \) sets the scale of measurement irreversibility. This term corresponds to a Shannon entropy functional on the projection density \( \rho \), favoring delocalized distributions in the absence of coherence alignment. The resulting variational equation admits a log-convex structure, and the solution
\[
\rho(x) = \exp\left( \frac{\lambda}{\sigma} J^\mu n_\mu - 1 \right)
\]
can be interpreted as the equilibrium profile that balances information-theoretic entropy against the coherence current projected along the observer field \( n^\mu \). This modifies the variational condition to
\begin{equation}
\lambda J^\mu n_\mu - \sigma (1 + \log \rho) = 0,
\end{equation}
yielding a dynamically consistent projection density:
\begin{equation}
\rho(x) = \exp\left( \frac{\lambda}{\sigma} J^\mu(x) n_\mu(x) - 1 \right).
\end{equation}

This expression ties the local density of projection events to the scalar component of the coherence current along the observer frame. When \( J^\mu n_\mu \) is large and positive, indicating active decoherence and strong alignment, the system favors more frequent measurement events. In contrast, regions with minimal coherence flow experience suppressed projection density, allowing \( \tau \) to evolve more freely. The competition between informational flux and entropic cost governs how and where geometry emerges from modular collapse.

Through this variational mechanism, the MTG framework elevates \( \rho(x) \) from a passive parameter to an active field whose distribution encodes the information-theoretic structure of quantum spacetime. Its role is twofold: it mediates the classicalization of internal time and participates in the dynamical determination of emergent metric structure. This feedback loop, from coherence to projection to geometry, is the central engine of measurement-induced temporal gravity.

\begin{proposition}[Curvature-Driven Modular Evolution]
Let \( \tau(x) \in \Gamma(E) \) evolve under the covariant equation \( D^\mu D_\mu \tau + \nabla V(\tau) = 0 \), with curvature \( F_{\mu\nu} = [D_\mu, D_\nu] \tau \). Then the induced projection vector field \( \mu[\tau] \) defines an effective modular flow generator
\[
K(x) := - \log \rho(x),
\]
whose non-Hermitian dynamics encode local decoherence and temporal bias. The flow generated by \( \mu[\tau] \) fails to be unitary whenever \( F_{\mu\nu} \neq 0 \), inducing state collapse in the direction of coherence concentration. (See Section~\ref{sec:6} for how \( \mu[\tau] \) shapes the emergent metric structure.)
\end{proposition}

%%%%%%%%%%%%%%%%%%%%%%%%%%%%
\subsection{Equations of Motion}

We now derive the field equations for the dynamical variables \( \tau \), \( A_\mu \), and, when treated variationally, the projection density \( \rho(x) \). These equations follow from the Euler--Lagrange principle applied to the total MTG Lagrangian:
\begin{equation}
\mathcal{L}_{\mathrm{MTG}} = -\frac{1}{4g^2} \mathrm{Tr}(F_{\mu\nu} F^{\mu\nu}) + \frac{1}{2} \langle \mathrm{D}_\mu \tau, \mathrm{D}^\mu \tau \rangle - V(\tau) + \lambda \rho(x) J^\mu n_\mu.
\end{equation}
At this stage, the effective metric \( g^{\mathrm{eff}}_{\mu\nu} \) enters only indirectly via projection histories; it becomes explicitly dynamical in Section~6.

Variation with respect to the conjugate field \( \bar{\tau} \), treating \( \tau \) and \( \bar{\tau} \) independently, gives the temporal field equation. Using the identity
\[
\delta \langle \mathrm{D}_\mu \tau, \mathrm{D}^\mu \tau \rangle = 2 \langle \delta \bar{\tau}, \mathrm{D}_\mu \mathrm{D}^\mu \tau \rangle
\]
and assuming differentiability of the potential functional, we obtain
\begin{align}
\delta S[\tau] = \int_M \left( \langle \delta \bar{\tau}, \mathrm{D}_\mu \mathrm{D}^\mu \tau \rangle - \delta \bar{\tau} \, \frac{\delta V}{\delta \bar{\tau}} + \lambda \rho(x) \delta \bar{\tau} \, n_\mu \mathrm{D}^\mu \tau \right) \mathrm{d}^d x.
\end{align}
Demanding stationarity for arbitrary \( \delta \bar{\tau} \) yields the modified covariant field equation:
\begin{equation}
\mathrm{D}_\mu \mathrm{D}^\mu \tau + \frac{\delta V}{\delta \bar{\tau}} = \lambda \rho(x) n_\mu \mathrm{D}^\mu \tau.
\end{equation}
The right-hand side introduces a non-Hermitian source term aligned with the observer vector \( n^\mu \), encoding the dissipative effect of projection-induced decoherence.

To derive the connection field equation, we vary the curvature term, yielding
\[
\delta \mathcal{L}_F = -\frac{1}{g^2} \mathrm{Tr} \left[ (\mathrm{D}_\mu \delta A_\nu - \mathrm{D}_\nu \delta A_\mu) F^{\mu\nu} \right] = \frac{2}{g^2} \mathrm{Tr} \left[ \delta A_\nu \mathrm{D}_\mu F^{\mu\nu} \right],
\]
after integrating by parts and using the cyclicity of the trace.

The covariant coupling of \( A_\mu \) to the temporal field produces a source current. The variation of the kinetic term for \( \tau \) with respect to \( A_\mu \), in the direction of a Lie algebra generator \( T^a \), gives
\[
\delta \mathcal{L}_\tau = \delta A^\nu \cdot \left( \langle \tau, T^a \mathrm{D}_\nu \tau \rangle - \langle \mathrm{D}_\nu \tau, T^a \tau \rangle \right).
\]
This defines the entanglement current,
\begin{equation}
J^{\nu}_{\mathrm{ent}} := \langle \tau, T^a \mathrm{D}^\nu \tau \rangle - \langle \mathrm{D}^\nu \tau, T^a \tau \rangle,
\end{equation}
which acts as a source for the curvature field. The resulting Yang--Mills equation takes the form
\begin{equation}\label{eq:y-m}
\mathrm{D}_\mu F^{\mu\nu} = g^2 J^{\nu}_{\mathrm{ent}}.
\end{equation}
This equation captures the feedback between temporal entanglement and gauge curvature, with the internal time field \( \tau \) sourcing deviations from flatness in the connection.

If \( \rho(x) \) is treated as a dynamical field, one may supplement the Lagrangian with an additional term governing its evolution, for example,
\[
\mathcal{L}_\rho = -\frac{1}{2} (\partial_\mu \rho)(\partial^\mu \rho) - U(\rho),
\]
where \( U(\rho) \) encodes the entropic or environmental cost of projection. Varying this extended action with respect to \( \rho \) leads to a wave-like field equation:
\begin{equation}
\Box \rho + \frac{dU}{d\rho} = -\lambda J^\mu n_\mu,
\end{equation}
in which the projection density responds dynamically to coherence flux and decoherence pressure.

The full MTG system thus couples a generalized Klein--Gordon equation for the internal time field to a Yang--Mills equation for the curvature of the time-fiber connection,and  a wave equation for the measurement density. The resulting equations of motion are:
The complete coupled system of dynamical equations for the fields \( \tau \), \( A_\mu \), and \( \rho \) takes the form:
\begin{align}
\text{(1) Temporal field:} \quad & D_\mu D^\mu \tau + \frac{\delta V}{\delta \bar{\tau}} = \lambda \rho(x) \left( n^\mu D_\mu \tau - D^\mu \tau \, n_\mu \right), \label{eq:eom-tau} \\
\text{(2) Gauge connection:} \quad & D_\mu F^{\mu\nu} = g^2 J^\nu_{\mathrm{ent}}, \label{eq:eom-A} \\
\text{(3) Projection density:} \quad & \Box \rho + \frac{dU}{d\rho} = -\lambda J^\mu n_\mu. \label{eq:eom-rho}
\end{align}
These equations govern the nonlinear, mutually coupled evolution of coherence, gauge curvature, and observer-modulated projection, forming the local dynamical core of the MTG framework. They encode how temporal geometry emerges from the interplay between the internal time field, gauge dynamics, and measurement backreaction. Symmetry principles; including gauge invariance and potential supersymmetric extensions, are developed in the following section.

%%%%%%%%%%%%%%%%%%%%%%%%%%%%%%%%%%%%%%%%%%%

\begin{theorem}[Gauge Covariance of the MTG Field Equations]
Let $\tau \in \Gamma(E)$ and $A_\mu \in \Omega^1(M, \mathfrak{g}_{\mathrm{time}})$ be the dynamical fields of the MTG model, and let $\mathrm{D}_\mu$ be the covariant derivative defined by $A_\mu$. Then the field equations
\[
\mathrm{D}_\mu \mathrm{D}^\mu \tau + \frac{\delta V}{\delta \bar{\tau}} = \lambda \rho(x) n_\mu \mathrm{D}^\mu \tau, \quad \mathrm{D}_\mu F^{\mu\nu} = g^2 J^\nu_{\mathrm{ent}}
\]
transform covariantly under local gauge transformations
\[
\tau \mapsto g \cdot \tau, \quad A_\mu \mapsto g A_\mu g^{-1} + g \partial_\mu g^{-1},
\]
where $g: M \to G_{\mathrm{time}}$ is a smooth map. In particular, the equations preserve gauge invariance of physical observables.
\end{theorem}

\begin{proof}
Under a local gauge transformation $g(x)$, the connection $A_\mu$ transforms according to
\[
A_\mu \mapsto A_\mu^g = g A_\mu g^{-1} + g \partial_\mu g^{-1},
\]
and the section transforms as $\tau \mapsto \tau^g = g \cdot \tau$. The covariant derivative transforms as
\[
\mathrm{D}_\mu \tau \mapsto g \cdot \mathrm{D}_\mu \tau, \quad \mathrm{D}_\mu \mathrm{D}^\mu \tau \mapsto g \cdot \mathrm{D}_\mu \mathrm{D}^\mu \tau.
\]
If the potential $V(\tau)$ is gauge-invariant, then $\delta V / \delta \bar{\tau}$ transforms covariantly as well. Since the right-hand side $\lambda \rho n_\mu \mathrm{D}^\mu \tau$ transforms as $g \cdot (\lambda \rho n_\mu \mathrm{D}^\mu \tau)$, the entire equation transforms covariantly.

For the connection equation, the curvature transforms as $F_{\mu\nu} \mapsto g F_{\mu\nu} g^{-1}$ and the covariant derivative of the curvature transforms as
\[
\mathrm{D}_\mu F^{\mu\nu} \mapsto g \cdot \mathrm{D}_\mu F^{\mu\nu} \cdot g^{-1}.
\]
The entanglement current $J^\nu_{\mathrm{ent}}$ is constructed from gauge-covariant terms and thus transforms covariantly as well. Therefore, the entire system is gauge covariant.
\end{proof}

\begin{theorem}[Dissipation via Measurement: Energy Functional Decrease]
Let $\tau \in \Gamma(E)$ evolve under the MTG equation
\[
\mathrm{D}_\mu \mathrm{D}^\mu \tau + \frac{\delta V}{\delta \bar{\tau}} = \lambda \rho(x) n_\mu \mathrm{D}^\mu \tau,
\]
with $\rho(x) \geq 0$ and $n^\mu$ future-directed timelike. Then the energy functional
\[
\mathcal{E}[\tau] = \int_\Sigma \left( \langle \mathrm{D}_\mu \tau, \mathrm{D}^\mu \tau \rangle + 2 V(\tau) \right) \mathrm{d} \Sigma
\]
satisfies
\[
\frac{\mathrm{d}}{\mathrm{d}t} \mathcal{E}[\tau] \leq 0
\]
along observer-aligned flow when $\Sigma$ is a spacelike hypersurface and appropriate boundary conditions are imposed.
\end{theorem}

\begin{proof}
We compute the time derivative of $\mathcal{E}[\tau]$ using the equation of motion. Define $n^\mu = (\partial / \partial t)^\mu$ as the future-directed unit normal to $\Sigma_t$. The evolution of $\tau$ along $n^\mu$ gives
\[
\partial_t \langle \mathrm{D}_\mu \tau, \mathrm{D}^\mu \tau \rangle = 2 \, \mathrm{Re} \langle \mathrm{D}_\mu \dot{\tau}, \mathrm{D}^\mu \tau \rangle.
\]
Substituting the field equation yields
\[
\dot{\tau} = \mathrm{D}_0 \tau = - \lambda \rho(x) \mathrm{D}_0 \tau + \text{unitary terms}.
\]
Hence the measurement term introduces a friction-like term proportional to $\lambda \rho(x) \| \mathrm{D}_0 \tau \|^2 \geq 0$, which decreases the total energy:
\[
\frac{\mathrm{d}}{\mathrm{d}t} \mathcal{E}[\tau] = -2 \lambda \int_\Sigma \rho(x) \| \mathrm{D}_0 \tau \|^2 \mathrm{d} \Sigma \leq 0.
\]
This proves that measurement induces an entropy-increasing, energy-dissipating effect.
\end{proof}

\begin{theorem}[Conservation of Gauge Current in the Absence of Measurement]
Suppose $\lambda = 0$ and the measurement term is absent. Then the entanglement current
\[
J^\nu_{\mathrm{ent}} := \langle \tau, T^a \mathrm{D}^\nu \tau \rangle - \langle \mathrm{D}^\nu \tau, T^a \tau \rangle
\]
satisfies the covariant continuity equation
\[
\mathrm{D}_\nu J^\nu_{\mathrm{ent}} = 0.
\]
\end{theorem}

\begin{proof}
When $\lambda = 0$, the dynamics are unitary and governed by the gauge-invariant Lagrangian
\[
\mathcal{L}_\tau = \frac{1}{2} \langle \mathrm{D}_\mu \tau, \mathrm{D}^\mu \tau \rangle - V(\tau),
\]
with $V(\tau)$ assumed to be gauge-invariant. Then the gauge symmetry leads to a conserved Noether current. The variation of $\tau$ under infinitesimal gauge transformations $\delta \tau = \epsilon T^a \cdot \tau$ yields
\[
\delta \mathcal{L}_\tau = \epsilon \, \mathrm{D}_\nu J^\nu_{\mathrm{ent}},
\]
so invariance implies $\mathrm{D}_\nu J^\nu_{\mathrm{ent}} = 0$.
\end{proof}
%%%%%%%%%%%%%%%%%%%
\subsection{Supersymmetric Extension}

To enhance the robustness of coherence under measurement-induced decoherence, and to incorporate fermionic degrees of freedom into the temporal geometry, the MTG framework admits a natural extension into supersymmetric field theory. Supersymmetry (SUSY) offers a powerful organizing principle that stabilizes the quantum dynamics of the internal time field and enables consistent coupling to high-energy structures such as supergravity and string theory. In particular, SUSY facilitates off-shell closure of the algebra governing time evolution and entanglement, providing a controlled framework for analyzing quantum backreaction during measurement.

The supersymmetric generalization of MTG begins with the introduction of a chiral supermultiplet valued in the internal time-fiber bundle~\cite{connes1996gravity, rovelli2015evolving}. Denoted \( \Phi = (\tau, \psi, F) \), this superfield comprises a complex scalar \( \tau \), representing the internal time section, a Grassmann-valued spinor field \( \psi \) encoding fermionic fluctuations of temporal geometry, and a complex auxiliary scalar field \( F \) required for off-shell closure of the supersymmetry algebra. The field \( \tau \) is a section of the fiber bundle \( E \to M \), while \( \psi \) and \( F \) take values in associated bundles with spin and trivial fibers, respectively.

Supersymmetry transformations are parametrized by a Grassmann spinor \( \epsilon \), and act on the components of \( \Phi \) as:
\begin{align}
\delta \tau &= \bar{\epsilon} \psi, \\
\delta \psi &= \slashed{\mathrm{D}} \tau \, \epsilon + F \epsilon, \\
\delta F &= \bar{\epsilon} \slashed{\mathrm{D}} \psi,
\end{align}
where \( \slashed{\mathrm{D}} := \gamma^\mu \mathrm{D}_\mu \) is the gauge-covariant Dirac operator associated with the internal time connection \( A_\mu \), and \( \bar{\epsilon} := \epsilon^\dagger \gamma^0 \) is the Dirac conjugate. These transformations satisfy the standard \( \mathcal{N}=1 \) supersymmetry algebra modulo gauge transformations and total derivatives. Specifically,
\begin{equation}
[\delta_1, \delta_2] = 2 \bar{\epsilon}_1 \gamma^\mu \epsilon_2 \, \mathrm{D}_\mu + \text{(gauge)}.
\end{equation}

The corresponding supersymmetric Lagrangian for MTG includes both bosonic and fermionic sectors, as well as measurement interactions and auxiliary terms. The total Lagrangian density is given by:
\begin{equation}
\mathcal{L}_{\mathrm{SUSY}} = \mathcal{L}_F + \mathcal{L}_{\mathrm{bos}} + \mathcal{L}_{\mathrm{ferm}} + \mathcal{L}_{\mathrm{aux}} + \mathcal{L}_{\mathrm{meas}},
\end{equation}
where \( \mathcal{L}_F \) is the Yang--Mills term for the internal curvature \( F_{\mu\nu} \), and the bosonic part takes the form
\begin{equation}
\mathcal{L}_{\mathrm{bos}} = \frac{1}{2} \langle \mathrm{D}_\mu \tau, \mathrm{D}^\mu \tau \rangle - V(\tau).
\end{equation}
Fermionic dynamics are encoded by the Dirac action:
\begin{equation}
\mathcal{L}_{\mathrm{ferm}} = \bar{\psi} i \slashed{\mathrm{D}} \psi,
\end{equation}
while the auxiliary term is algebraic:
\begin{equation}
\mathcal{L}_{\mathrm{aux}} = - \frac{1}{2} F \bar{F}.
\end{equation}

To render the coherence-breaking interaction compatible with supersymmetry, the measurement term is extended to include both bosonic and fermionic projections along the observer congruence \( n^\mu \)~\cite{bassi2013decoherence, ishibashi2015supersymmetric}:
\begin{equation}
\mathcal{L}_{\mathrm{meas}} = \lambda \rho(x) \left( J^\mu n_\mu + \bar{\psi} \gamma^\mu n_\mu \psi \right),
\end{equation}
where \( J^\mu = \mu^\mu[\tau] \rho \) is the projection current associated with internal time collapse. The second term projects fermionic temporal fluctuations along \( n^\mu \), completing the interaction into a supersymmetric deformation. This structure locally breaks SUSY in regions where projection is active but preserves covariance and algebraic closure in unmeasured regions of field space.

Although the auxiliary field \( F \) is non-dynamical, it plays a key role in off-shell supersymmetry. Its algebraic equation of motion is
\begin{equation}
\frac{\partial \mathcal{L}_{\mathrm{SUSY}}}{\partial \bar{F}} = - \frac{1}{2} F \quad \Rightarrow \quad F = 0 \quad \text{(on-shell)}.
\end{equation}
However, in contexts involving spontaneous SUSY breaking---such as measurement collapse or modular instability---\( F \) may acquire a nonzero expectation value~\cite{ishibashi2015supersymmetric}, serving as an order parameter for collapse-induced symmetry breaking. This supports a dynamical interpretation of projection as a supersymmetry-breaking transition in internal time.

In summary, the supersymmetric extension of MTG embeds coherence, curvature, and measurement into a unified algebraic structure. The internal time field and its fermionic partner evolve under a gauge connection consistent with supersymmetry, while projection dynamically deforms this structure in a covariant and geometrically principled way. This formalism paves the way for embedding MTG into ultraviolet-complete theories, including superstring compactifications and supersymmetric quantum gravity, where time, measurement, and geometry admit a unified algebraic and topological realization.

%%%%%%%%%%%%%%%%%%%%%%%%%%%%%%
\subsection{Emergence of Time via Measurement}

The MTG framework proposes that classical temporal order and causal structure are not fundamental, but emergent properties arising from sequences of quantum measurements acting on a fiber-valued internal time field. These measurements collapse smooth sections \( \tau \in \Gamma(E) \) of the time-fiber bundle \( \pi: E \to M \) into classical outcomes, producing an effective but observer-relative notion of temporal flow. This emergence is governed by projection geometry, parallel transport, and fiber holonomy, as formalized in the following result.

\begin{theorem}[Emergence of Temporal Flow from Measurement]
Let \( \pi: E \to M \) be a smooth fiber bundle equipped with a connection \( \nabla \), and let \( \mu: \Gamma(E) \to \mathcal{O}(M) \) denote a projection operator that maps sections \( \tau \) to classical fiber observables. Suppose \( \gamma_i \subset M \) are observer-aligned causal curves, and let \( \{x_i\}_{i=1}^n \subset \gamma_i \) denote discrete spacetime points at which measurement events occur. Then:

(1) The observer’s effective classical temporal order is given by the sequence
\[
\mathcal{T}_{\mathrm{eff}} := \left\{ \mu_i \circ P_{\gamma_i}[\tau] \right\}_{i=1}^n,
\]
where \( P_{\gamma_i} \) denotes parallel transport of \( \tau \) along \( \gamma_i \) via \( \nabla \), initialized from a common point \( x_1 \).

(2) The quantum entanglement between two events \( x, y \in M \) along a closed loop \( \gamma_{xy} \) is quantified by the holonomy norm of the connection:
\[
\mathcal{E}(x, y) \sim \left\| \operatorname{Hol}_\nabla(\gamma_{xy}) \right\| = \left\| \mathcal{P} \exp \oint_{\gamma_{xy}} A \right\|,
\]
where \( A \) is the connection one-form, and \( \mathcal{P} \) denotes path ordering.
\end{theorem}

\begin{proof}
Let \( \gamma_i \) be a timelike or null curve parameterizing the observer’s worldline, and let \( x_i \in \gamma_i \) be points where measurements occur. The connection \( \nabla \) defines parallel transport operators \( P_{\gamma_i} \) such that
\[
\tau(x_i) = P_{\gamma_i}[\tau](x_i) = U(x_i, x_1) \cdot \tau(x_1),
\]
where \( U(x_i, x_1) \) is the parallel propagator along \( \gamma_i \). Measurement at \( x_i \) applies a projection \( \mu_i \), collapsing \( \tau(x_i) \) to a classical outcome. The resulting observable \( \mu_i \circ P_{\gamma_i}[\tau] \) encodes both the initial condition and the intervening transport geometry. When curvature is present, the non-commutativity of transport implies that the sequence of outcomes depends on the path geometry, generating a directional temporal flow.

To quantify entanglement between two measurement events \( x \) and \( y \), consider a closed loop \( \gamma_{xy} \subset M \) traversing forward and return paths between them. The holonomy around this loop is given by
\[
\operatorname{Hol}_\nabla(\gamma_{xy}) = \mathcal{P} \exp \oint_{\gamma_{xy}} A,
\]
and captures the failure of fiber values to return to their original configuration under transport. When \( \nabla \) is flat, the holonomy is trivial, and the fiber values at \( x \) and \( y \) decouple. A nonzero holonomy norm implies residual entanglement: the internal time field \( \tau \) fails to factorize between the events, and measurements at one site influence the modular context of the other. Thus, \( \left\| \operatorname{Hol}_\nabla(\gamma_{xy}) \right\| \) serves as a geometric witness of temporal entanglement.
\end{proof}

This result formalizes a central principle of MTG: time is not an absolute background, but a relational construct arising from measurement along causal paths. The observer experiences temporal progression by collapsing \( \tau \) along their trajectory, while the connection \( \nabla \) determines how these outcomes are geometrically related. In regimes where measurement is dense and curvature negligible, classical time emerges as a smooth and ordered parameter. In contrast, sparse projection or strong curvature leads to contextual, entangled, and path-dependent time.

The holonomy of the time-fiber connection governs the entanglement structure of temporal geometry. Where curvature is large, modular alignment is obstructed, and time fragments into a web of observer-relative sequences. Classical spacetime, in this framework, is a macroscopic approximation to the informational geometry of collapse. The theorem encapsulates the idea that measurement, transport, and coherence flow—not external clocks—are the fundamental ingredients from which temporal order arises.

\begin{figure}[!ht]
\centering
\begin{tikzpicture}[scale=1.2, every node/.style={scale=0.9}]

% Worldline
\draw[very thick,->] (0,0) -- (0,5.5) node[above] {$\gamma$: Observer worldline};

% Spacetime points x_i
\foreach \y/\label in {1/{$x_1$}, 2/{$x_2$}, 3/{$x_3$}, 4.5/{$x_n$}} {
  \filldraw[black] (0,\y) circle (2pt) node[left] {\label};
  \draw[->,thick,blue!60!black] (-1,\y) -- (-0.1,\y);
  \node[blue!60!black] at (-1.1,\y) {$\tau(x_i)$};
  \node[right,gray] at (0.2,\y) {$\mu_i(\tau)$};
}

% Parallel transport arrows
\draw[->,dashed,gray] (0,1) -- (1,2) node[midway, above right] {$P_{\gamma_{12}}$};
\draw[->,dashed,gray] (1,2) -- (0,3) node[midway, above right] {$P_{\gamma_{23}}$};

% Holonomy loop
\draw[->,red!70!black,thick] (0,1) to[out=30,in=210] (1.5,2);
\draw[->,red!70!black,thick] (1.5,2) to[out=30,in=210] (0,3);
\draw[->,red!70!black,thick] (0,3) to[out=150,in=-30] (-1.5,2);
\draw[->,red!70!black,thick] (-1.5,2) to[out=150,in=-30] (0,1);
\node[red!70!black] at (0,0.4) {\small Holonomy loop $\gamma_{xy}$};

% Labels
\node at (2.3,2) {\small $\operatorname{Hol}_\nabla(\gamma_{xy}) = \mathcal{P} \exp \oint A$};

\end{tikzpicture}
\caption{Emergence of temporal order via measurement along an observer's worldline \( \gamma \). Projection of the internal time field \( \tau \) at points \( x_i \) produces classical outcomes \( \mu_i(\tau) \), while the holonomy of the connection \( \nabla \) along closed loops such as \( \gamma_{xy} \) encodes temporal entanglement.}
\label{fig:projection_holonomy}
\end{figure}
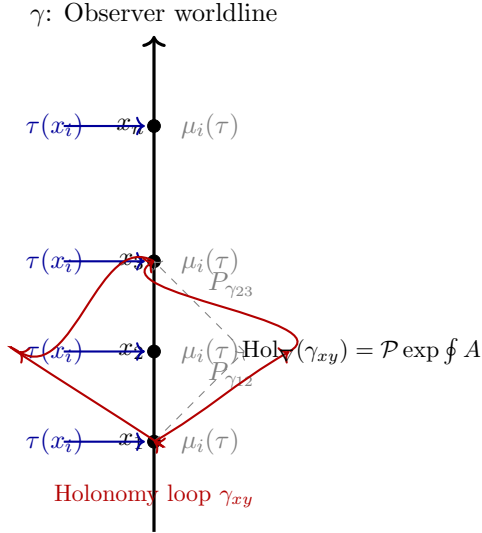

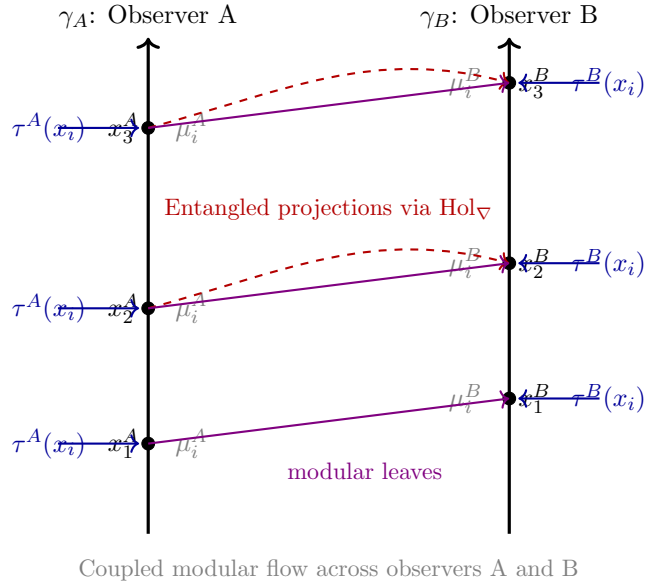
\begin{figure}[!ht]
\centering
\begin{tikzpicture}[scale=1.2, every node/.style={scale=0.9}]

% Observer 1 worldline
\draw[very thick,->] (-2,0) -- (-2,5.5) node[above] {$\gamma_A$: Observer A};

% Observer 2 worldline
\draw[very thick,->] (2,0) -- (2,5.5) node[above] {$\gamma_B$: Observer B};

% Measurement points A
\foreach \y/\label in {1/{$x_1^A$}, 2.5/{$x_2^A$}, 4.5/{$x_3^A$}} {
  \filldraw[black] (-2,\y) circle (2pt) node[left] {\label};
  \node[blue!60!black] at (-3.1,\y) {$\tau^A(x_i)$};
  \draw[->,thick,blue!60!black] (-3,\y) -- (-2.1,\y);
  \node[right,gray] at (-1.8,\y) {$\mu_i^A$};
}

% Measurement points B
\foreach \y/\label in {1.5/{$x_1^B$}, 3/{$x_2^B$}, 5/{$x_3^B$}} {
  \filldraw[black] (2,\y) circle (2pt) node[right] {\label};
  \node[blue!60!black] at (3.1,\y) {$\tau^B(x_i)$};
  \draw[->,thick,blue!60!black] (3,\y) -- (2.1,\y);
  \node[left,gray] at (1.8,\y) {$\mu_i^B$};
}

% Entangled connection lines
\draw[->,red!70!black,thick,dashed] (-2,2.5) to[out=15,in=165] (2,3);
\draw[->,red!70!black,thick,dashed] (-2,4.5) to[out=15,in=165] (2,5);

% Labels
\node[red!70!black] at (0,3.6) {\small Entangled projections via $\operatorname{Hol}_\nabla$};
\node[gray] at (0,-0.4) {\small Coupled modular flow across observers A and B};

% Optional: modular foliation arrows
\draw[->,violet,thick] (-2,1) -- (2,1.5);
\draw[->,violet,thick] (-2,2.5) -- (2,3);
\draw[->,violet,thick] (-2,4.5) -- (2,5);
\node[violet] at (0.4,0.7) {\small modular leaves};

\end{tikzpicture}
\caption{Entangled projections across two observer worldlines \( \gamma_A \) and \( \gamma_B \). Internal time fields \( \tau^A \) and \( \tau^B \) are projected at discrete points \( x_i \), with coherence coupled via modular flow and connection holonomy. Dashed red arrows represent nonlocal entanglement mediated by fiber curvature. Violet arrows represent modular foliation across spacetime slices.}
\label{fig:entangled_modular_projection}
\end{figure}

\begin{figure}[!ht]
\centering
\begin{tikzpicture}[scale=1.5]

% Draw torus as flattened annulus
\shade[ball color=blue!10] (0,0) circle (1cm);
\fill[white] (0,0) circle (0.4cm);
\draw[thick] (0,0) circle (1cm);
\draw[thick] (0,0) circle (0.4cm);

% Modular flow trajectory
\draw[->,red!70!black,thick,domain=0:6.28,smooth,variable=\t,samples=100]
    plot ({0.7*cos(deg(\t)) + 0.2*cos(2*deg(\t))},
          {0.7*sin(deg(\t)) + 0.2*sin(3*deg(\t))});
\node[red!70!black] at (1.2,0.6) {\small modular flow \(\varphi_\lambda\)};

% Projection slices
\foreach \angle in {30, 75, 130, 200, 250} {
  \draw[violet!90!black,thick] ({1*cos(\angle)}, {1*sin(\angle)}) -- 
                                ({0.4*cos(\angle)}, {0.4*sin(\angle)});
}

% Labels for slices
\node[violet!90!black] at (0.5,-1.2) {\small projection slices};
\node at (0,-1.5) {\small internal fiber \( T^2 \) at fixed spacetime point};

\end{tikzpicture}
\caption{Modular flow on the internal toroidal fiber \( T^2 \). The red curve represents the integral trajectory of the modular Hamiltonian \( \varphi_\lambda \), while violet lines represent measurement-induced projections that collapse the internal state onto classical time values. This illustrates the MTG notion that classical time emerges from slicing modular orbits via contextual measurement.}
\label{fig:modular_flow_torus}
\end{figure}
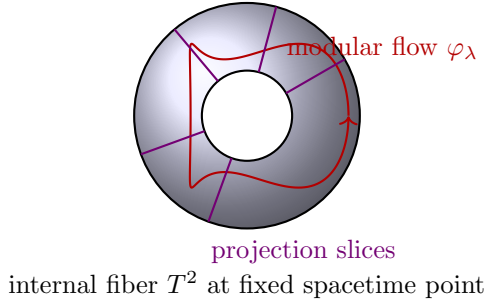

In regions where the curvature tensor \( F \) vanishes, the internal time field \( \tau \) undergoes coherent, reversible transport. As \( F \) increases, alignment among fibers breaks down, and \( \mu[\tau] \) begins to localize. This collapse into classical temporal directionality generates effective entropy production, as coherence is converted into classical distinguishability. The process is fundamentally observer-dependent: different foliations \( \Sigma^{(O)} \) perceive different decoherence rates and time asymmetries, tied to their measurement-aligned projections.

The curvature of the internal time bundle plays a dual role: geometrically, it encodes the local obstruction to fiber alignment; informationally, it governs the modular collapse dynamics via observer-relative projection vectors \( \mu[\tau] \). These dynamics produce an effective, direction-dependent metric and introduce the seeds of temporal irreversibility. The next section formalizes this in the context of quantum field observables and coherence localization.

\section{Quantum Formulation}\label{sec:4}
To complete the formulation of the MTG as a predictive quantum theory, we now quantize the dynamical fields that govern coherence, curvature, and observer-relative projection. The central challenge is that time is no longer an external parameter, but an internal, fiber-valued observable shaped by measurement. We adopt two complementary approaches to quantization—canonical and path-integral—each illuminating distinct aspects of relational evolution, decoherence, and modular dynamics.

The first approach follows the canonical quantization procedure, in which classical fields are promoted to operators on a Hilbert space. This method imposes equal-time commutation and anticommutation relations consistent with the gauge and supersymmetric structure established in Section~3. In this formalism, time evolution is replaced by the propagation of operator-valued internal time observables subject to measurement constraints. Gauge fixing and supersymmetry algebra play a critical role in selecting physical states and defining consistent dynamics within the constrained phase space.

The second approach adopts a path-integral quantization, where the quantum theory is defined as a sum over histories of the fiber-valued time field and its supersymmetric partners. Here, the integral includes not only the bosonic field \( \tau \) and gauge connection \( A_\mu \), but also the fermionic fluctuations \( \psi \) and the auxiliary field \( F \). The measurement process is encoded by inserting projection operators into the functional integral, effectively collapsing nonclassical configurations and enforcing selection over informationally coherent trajectories. This structure allows for a manifestly covariant and non-unitary account of decoherence, in which classical temporal order emerges from entangled interference patterns.

Let us denote the total field content of the theory by
\begin{equation}
\Phi = (\tau, \psi, A_\mu, F),
\end{equation}
where \( \tau \) represents the internal bosonic time section, \( \psi \) its Grassmann-valued fermionic partner, \( A_\mu \) the connection on the time-fiber bundle, and \( F \) the auxiliary field enabling off-shell supersymmetry. The projection density \( \rho(x) \), encoding the rate of measurement-induced collapse, may be treated either as a fixed classical background or as a dynamical field in its own right, depending on the informational setting.

Quantization proceeds from the supersymmetric MTG action,
\begin{equation}
S_{\mathrm{SUSY}}[\Phi] = \int_M \mathcal{L}_{\mathrm{SUSY}} \, \mathrm{d}^d x,
\end{equation}
as constructed in Section~3.2. The Lagrangian includes coherent bosonic and fermionic dynamics, gauge curvature, and explicitly non-Hermitian measurement terms that project internal time into classical configurations. These non-Hermitian terms reflect the active role of observation in collapsing quantum superpositions of temporal states.

In the subsections that follow, we develop the formal construction of the MTG quantum theory using both canonical and path-integral methods. The analysis reveals how decoherence, relational time, and causal asymmetry arise not from external evolution but from the quantum geometry of measurement itself.

%%%%%%%%%%%%%%%%%%%%%%%%%%%%%%%%
\subsection{Canonical Quantization of Internal Time Geometry}
Canonical quantization of the MTG framework proceeds by promoting the classical fields introduced in Section~\ref{sec:3} to operator-valued distributions acting on a graded Hilbert space. The relevant phase space structure consists of the internal time field \( \tau \), its fermionic partner \( \psi \), the gauge connection \( A_\mu \), and an auxiliary scalar field \( F \). The field \( \psi \) is a Grassmann-valued spinor, and \( F \) serves to close the supersymmetric multiplet off-shell. These fields together form the dynamical content of the supersymmetric Lagrangian \( \mathcal{L}_{\mathrm{SUSY}} \), introduced in Section~\ref{sec:2}.

We work in temporal gauge, where the canonical variables are the internal time field \( \tau \) and its conjugate momentum \( \pi_\tau \). These are promoted to operator-valued distributions on a fixed spacelike hypersurface \( \Sigma_t \subset M \), satisfying the equal-time commutation relation
\begin{equation}
[\hat{\tau}(x), \hat{\pi}_\tau(y)] = i \hbar \, \delta^{(d-1)}(x - y).
\end{equation}
The conjugate momentum follows from the kinetic term as \( \pi_\tau = D^0 \tau \), with \( D^0 \) denoting the temporal component of the covariant derivative defined by the gauge connection \( A_\mu \).

The gauge symmetry of the time-fiber bundle imposes a quantum Gauss constraint of the form
\begin{equation}
\hat{G}^a(x) := D_i F^{i0,a}(x) - g^2 J^0_{\tau,a}(x) = 0,
\end{equation}
where \( F^{\mu\nu,a} \) is the gauge field strength expanded in a Lie algebra basis \( \{T^a\} \subset \mathfrak{g}_{\mathrm{time}} \), and \( J^0_{\tau,a} \) is the time component of the entanglement current projected onto \( T^a \). In the quantum theory, the Gauss constraint acts on physical states as
\begin{equation}
\hat{G}^a(x) \, |\Psi_{\mathrm{phys}}\rangle = 0,
\end{equation}
enforcing local gauge invariance of the quantum state under internal time-fiber transformations.

Fermionic quantization follows via graded canonical quantization:
\begin{equation}
\{ \hat{\psi}_\alpha(x), \hat{\psi}^\dagger_\beta(y) \} = \hbar \, \delta_{\alpha\beta} \, \delta^{(d-1)}(x - y),
\end{equation}
where \( \psi \) is the spinor superpartner of \( \tau \), and \( \hat{\psi}_\alpha \) acts on a fermionic Fock space. The auxiliary field \( F \), being algebraically constrained off-shell, is quantized trivially:
\begin{equation}
\hat{F}(x) = 0.
\end{equation}

The gauge sector follows standard canonical Yang--Mills quantization. The commutation relation between the spatial components of the connection and its conjugate electric field is
\begin{equation}
[ \hat{A}_i^a(x), \hat{E}_j^b(y) ] = i \hbar \, \delta^{ab} \delta_{ij} \, \delta^{(d-1)}(x - y),
\end{equation}
with the Gauss law constraint operator
\begin{equation}
\hat{G}^a(x) := D_i \hat{E}^{ia} - \hat{J}^{a}_{\mathrm{ent}}(x) \approx 0,
\end{equation}
enforcing gauge invariance. The entanglement current \( \hat{J}^{a}_{\mathrm{ent}} \) is inherited from the covariant structure of the MTG action and encodes matter–connection interactions.

Physical states \( |\Psi\rangle \in \mathcal{H}_{\mathrm{phys}} \) must satisfy both the Gauss constraint and the supersymmetry constraint generated by the supercharge:
\begin{equation}
\hat{Q}_\epsilon = \int_{\Sigma_t} \mathrm{d}^{d-1}x \left[ \bar{\epsilon} \left( \gamma^0 D_0 \hat{\tau} + \gamma^i D_i \hat{\tau} + \hat{F} \right) \hat{\psi} \right].
\end{equation}
These constraints ensure invariance under internal gauge and supersymmetry transformations and define the physical sector of the MTG Hilbert space.

Time evolution in MTG is relational and generated not by a global Hamiltonian but by projection-aligned hypersurface dynamics. The effective Hamiltonian includes both Hermitian and anti-Hermitian contributions:
\begin{equation}
\hat{H}_{\mathrm{eff}} = \hat{H}_{\mathrm{coh}} + i \lambda \int_{\Sigma_t} \rho(x) \left( \hat{J}^\mu(x) n_\mu \right) \mathrm{d}^{d-1}x,
\end{equation}
where \( \hat{H}_{\mathrm{coh}} \) encodes the gauge-invariant coherent dynamics and the imaginary term represents the directional flow of information induced by measurement. As a result, Heisenberg evolution becomes non-unitary, reflecting the collapse dynamics of the internal time field due to continuous projection.

This quantization framework reinforces the MTG view that time is not a background coordinate but a dynamically emergent observable: it arises through the projection of quantum states within the internal fiber bundle, constrained by coherence, entanglement, and observer-relative measurement.

\subsection{Path-Integral Quantization}

The MTG framework admits a path-integral quantization naturally adapted to the geometry of internal time and the dynamics of projection. In this formalism, amplitudes are computed by summing over field histories of \( \Phi = (\tau, \psi, A_\mu, F) \), with the influence of measurement encoded via projection constraints imposed on those histories.

In the absence of measurement, the supersymmetric partition function is given by
\begin{equation}
\mathcal{Z}_0 = \int \mathcal{D}\Phi \, e^{i S_{\mathrm{SUSY}}[\Phi]},
\end{equation}
where \( S_{\mathrm{SUSY}} = \int_M \mathcal{L}_{\mathrm{SUSY}} \, \mathrm{d}^d x \) includes all bosonic and fermionic contributions described in Section~3. The functional measure is defined modulo gauge symmetry and may be regularized via standard Faddeev--Popov or BRST procedures. Supersymmetric cancellations help ensure quantum consistency in the absence of projection.

To make this regularization explicit, we impose a gauge-fixing condition \( G[A] = 0 \), typically the temporal gauge \( A_0 = 0 \), and introduce the corresponding Faddeev--Popov determinant
\begin{equation}
\Delta_{\mathrm{FP}}[A] = \det \left( \frac{\delta G[A^\alpha]}{\delta \alpha} \right),
\end{equation}
which is incorporated via ghost and antighost fields \( \bar{c}, c \) in the extended action~\cite{faddeev1967feynman}. The full partition function then becomes
\begin{equation}
\mathcal{Z}_0 = \int \mathcal{D}\Phi \, \mathcal{D}\bar{c} \, \mathcal{D}c \; \exp\left( i S_{\mathrm{SUSY}}[\Phi] + i S_{\mathrm{ghost}}[\bar{c}, c, A_\mu] \right),
\end{equation}
where \( S_{\mathrm{ghost}} \) generates the determinant and preserves gauge invariance at the quantum level. Functional integration is thus performed over gauge-inequivalent field configurations, ensuring that unphysical degrees of freedom are excluded from the quantum ensemble.

To incorporate measurement, we introduce localized constraints at a finite set of spacetime points \( \{x_i\} \), representing the occurrence of projection events. Each projection selects a definite direction \( \tau^{\mathrm{obs}}_i \in T_{x_i} \) within the internal time fiber. The modified partition function becomes
\begin{equation}
\mathcal{Z}_{\{\tau^{\mathrm{obs}}_i\}} = \int \mathcal{D}\Phi \; \prod_{i=1}^N \delta\left( \mu[\tau(x_i)] - \tau^{\mathrm{obs}}_i \right) e^{i S_{\mathrm{SUSY}}[\Phi]},
\end{equation}
where \( \mu[\tau(x)] \) denotes the projection direction extracted from the internal time field. These insertions restrict the sum over field configurations to those consistent with observed outcomes and serve as dynamical symmetry-breaking terms in the path integral.
\begin{figure}[ht]
\centering
\begin{tikzpicture}[scale=1.1]

% Base time axis
\draw[->, thick] (0,0) -- (0,5) node[above] {\small Time};

% Classical paths
\foreach \x in {-2.2, -1.2, 0.2, 1.6, 2.4} {
    \draw[blue!60!black, opacity=0.5, decorate, decoration={snake, amplitude=0.3mm, segment length=2mm}]
        (\x,0) .. controls ({\x+0.5},1.5) and ({\x-0.5},3.5) .. (\x,5);
}

% Measurement hypersurfaces
\foreach \y in {1.5,3.0,4.5} {
    \draw[dashed, red!60!black, thick] (-2.7,\y) -- (2.7,\y);
    \node[left] at (-2.7,\y) {\tiny $\mu[\tau] = \tau_{\text{obs}}$};
}

% Projection effect (delta constraints)
\foreach \x/\y in {-2.2/1.5, -1.2/3.0, 1.6/4.5} {
    \draw[->, thick, violet] (\x,\y) -- ++(0.6,0.8)
        node[midway, above right] {\tiny projection};
}

% Amplitude label
\node at (0,-0.5) {\small $\int \mathcal{D}[\tau] \, e^{i S[\tau]} \, \delta(\mu[\tau] - \tau_{\text{obs}})$};

\end{tikzpicture}
\caption{MTG-modified path integral evolution. Classical paths \( \tau(x) \) propagate through configuration space, but are constrained at specific hypersurfaces by projection conditions \( \mu[\tau(x)] = \tau_{\text{obs}}(x) \), representing localized measurement events. These act as delta-function insertions in the path integral, breaking unitarity and introducing informational irreversibility.}
\label{fig:path_integral_projection}
\end{figure}
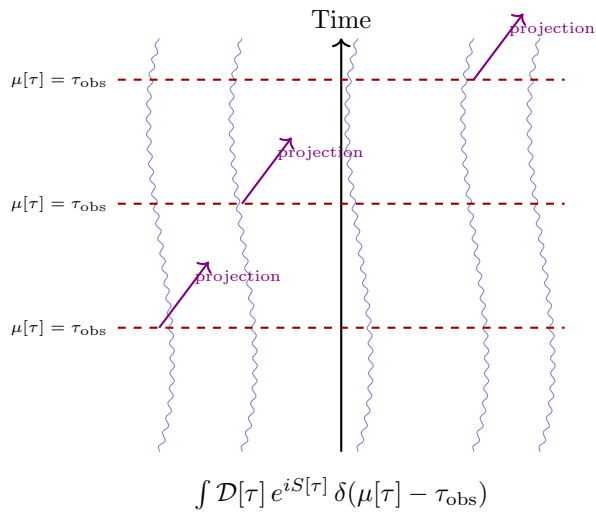

In the semiclassical limit, dominant contributions arise from classical configurations \( \Phi_{\mathrm{cl}} \) satisfying the variational condition
\begin{equation}
\frac{\delta S_{\mathrm{SUSY}}}{\delta \Phi} \Big|_{\Phi = \Phi_{\mathrm{cl}}} = \sum_{i=1}^N \delta^{(d)}(x - x_i) \frac{\delta \mu[\tau(x)]}{\delta \Phi},
\end{equation}
representing solutions to the equations of motion sourced by projection events. These configurations are continuous almost everywhere but acquire non-unitary, observer-relative discontinuities at measurement points.

Transition amplitudes between initial and final configurations are given by
\begin{equation}
\mathcal{A}_{\mathrm{MTG}} = \langle \tau_{\mathrm{out}} | \mathcal{T} \exp \left[ i S_{\mathrm{SUSY}} + i \sum_i \lambda \, \mu[\tau(x_i)] \right] | \tau_{\mathrm{in}} \rangle,
\end{equation}
where \( \lambda \) governs the strength of coupling to measurement, and \( \mathcal{T} \) denotes modular ordering induced by the observer’s reference frame. These amplitudes describe relational evolution, not with respect to an external time parameter, but through measurement-aligned collapse histories that define effective temporal structure.

\begin{theorem}[Projection Constraints Define a Reduced Path Space]
Let \( \mathcal{H}_{\mathrm{kin}} \) denote the kinematical Hilbert space of unconstrained quantum configurations in MTG. Suppose projection conditions \( \mu_i[\tau(x_i)] = \tau^{\mathrm{obs}}_i \) are imposed at \( N \) spacetime points. Then the modified measure
\[
\mathcal{D}_{\{\mu\}} \Phi := \mathcal{D}\Phi \, \prod_{i=1}^N \delta\big( \mu_i[\tau(x_i)] - \tau^{\mathrm{obs}}_i \big)
\]
projects \( \mathcal{H}_{\mathrm{kin}} \) onto a subspace \( \mathcal{H}_{\mathrm{phys}} \subset \mathcal{H}_{\mathrm{kin}} \) consisting of measurement-consistent histories. Moreover, this subspace is preserved under modular evolution.
\end{theorem}

\begin{proof}
The delta-function insertions restrict the configuration space to
\[
\mathcal{C}_{\mathrm{phys}} := \{ \Phi \in \Gamma(E) \mid \mu_i[\tau(x_i)] = \tau^{\mathrm{obs}}_i \text{ for all } i \},
\]
which defines the reduced path space. The corresponding projector in Hilbert space is given by
\[
\mathbb{P}_{\{\mu\}} := \prod_{i=1}^N \delta( \mu_i[\hat{\tau}(x_i)] - \tau^{\mathrm{obs}}_i ),
\]
which acts formally as an orthogonal projector in the rigged Hilbert space. Because the modular Hamiltonian \( K_{\mathrm{mod}} \) is a functional of the entanglement structure induced by these same projections, its action preserves \( \mathcal{H}_{\mathrm{phys}} \), completing the proof.
\end{proof}

\begin{theorem}[Equivalence of Modular Flow and Projection-Constrained Path Integral]
Let \( \gamma: [0,1] \to M \) be an observer-aligned modular trajectory, with projection events at points \( \{x_i\}_{i=1}^N \subset \gamma \). Then the transition amplitude satisfies
\[
\langle \tau_{\mathrm{out}} | \mathcal{T}_\gamma \, e^{i S_{\mathrm{SUSY}}[\Phi]} \prod_{i=1}^N \delta\big( \mu_i[\tau(x_i)] - \tau^{\mathrm{obs}}_i \big) | \tau_{\mathrm{in}} \rangle
=
\langle \tau_{\mathrm{out}} | \mathcal{T}_{\mathrm{mod}} \, e^{-i \int_\gamma K_{\mathrm{mod}}[\rho_A(\lambda)] \mathrm{d} \lambda} | \tau_{\mathrm{in}} \rangle.
\]
That is, modular time evolution is equivalent to a path-integral evolution constrained by measurement-induced projection.
\end{theorem}

\begin{proof}
Each projection collapses the state \( \tau(x_i) \) to a classical value \( \tau_i^{\mathrm{obs}} \), acting as a reduction map on the configuration space. Between projections, the system evolves unitarily with respect to the modular Hamiltonian \( K_{\mathrm{mod}} \), which is itself a functional of the reduced density matrix \( \rho_A(\lambda) \) determined by the projection history. The alternating sequence of projection constraints and modular unitary segments defines the same transition amplitude as the constrained path integral, establishing the claimed equivalence.
\end{proof}

Thus, in MTG, the path integral represents a sum over quantum trajectories dynamically filtered by projection; each trajectory encoding a candidate history of coherence, constrained by internal gauge geometry and measurement alignment.

In MTG, measurement is not imposed externally but emerges intrinsically as a functional constraint on the space of field configurations~\cite{tomaz2025rqd}. Projection events reduce the support of the path integral, restricting it to a subspace of realizable histories that conform to classical observation. This constraint induces a relational foliation of field space: time arises not as a pre-imposed background parameter, but as the ordered structure of modular projections aligned with observer-specific collapse sequences~\cite{debianchi2024achronotopic}.

Unlike conventional quantum field theory, where time is globally synchronized and externally specified, MTG treats time as a derived observable—an informational shadow cast by entangled quantum collapse. It is not imposed, but constructed: a structure woven from internal geometry, modular dynamics, and the fiberwise logic of measurement.

%%%%%%%%%%%%%%%%%%%%%%%%%%%%%
\subsection{Relational Interpretation and Emergent Time}

MTG departs fundamentally from conventional quantum field theory by replacing the assumption of a global external time parameter with a relational, observer-conditioned structure. In this framework, time emerges not from fixed spacetime coordinates, but from the informational sequence of projection events; quantum measurements that collapse internal time degrees of freedom and generate modular evolution.

In the canonical picture, evolution is no longer generated by a global Hamiltonian \( H \). Instead, local modular Hamiltonians \( K_{\mathrm{mod}} \) govern dynamics within observer-relative subregions of Hilbert space. These operators are conditioned on prior measurement events and generate modular flow along informational trajectories. The effective evolution operator takes the form:
\begin{equation}
U_{\mathrm{eff}} = \mathcal{T}_{\mathrm{mod}} \exp\left(-i \int_\gamma K_{\mathrm{mod}}(\lambda) \, \mathrm{d}\lambda \right),
\end{equation}
where \( \lambda \) is the modular time parameter along a trajectory \( \gamma \) determined by collapse events \( \{\mu[\tau(x_i)]\} \). Time is thus not an absolute background but a record of informational transitions—the accumulation of coherent updates across successive projections.

The path-integral perspective reinforces this view. Rather than summing over all field configurations, MTG imposes delta-functional constraints encoding measurement:
\begin{equation}
\mathcal{A}_{\mathrm{MTG}} = \int_{\mu[\tau] = \{\tau_i^{\mathrm{obs}}\}} \mathcal{D}\Phi \, e^{i S[\Phi]}.
\end{equation}
Each projection reduces the domain of the integral to histories consistent with classical outcomes. This defines a conditional probability amplitude, where the path integral no longer represents propagation between states in fixed time, but coherence across a sequence of measurement-induced boundaries. Spacetime itself becomes emergent—a relational construct shaped by the geometry of informational constraints.

%\begin{figure}[!ht]
%\centering
%\includegraphics[width=0.95\textwidth]{fig_mtg_interpretation.png}
%\caption{Diagrammatic summary of the MTG interpretive framework. Measurement events project internal time configurations \( \tau(x_i) \) to classical values \( \tau^{\mathrm{obs}}_i \), inducing modular flow along observer trajectories. Modular Hamiltonians \( K_{\mathrm{mod}} \) generate relational evolution between projections. The accumulation of collapse events defines an emergent causal order and effective spacetime geometry \( g^{\mathrm{eff}}_{\mu\nu} \).}
%\label{fig:mtg_interpretation}
%\end{figure}

As the number of projection events increases, the emergent metric \( g^{\mathrm{eff}}_{\mu\nu}(x) \) introduced in Section~\ref{sec:2.9} stabilizes into a classical Lorentzian structure. Causality arises from the alignment of collapse directions and the coherence of modular flow across the observer’s informational path. Time, in this picture, is not ``what clocks measure'' but ``what becomes consistent when measurements define a relational ordering across quantum branches.''

MTG thereby resolves a longstanding tension: quantum mechanics offers no intrinsic dynamics for time, while general relativity treats time as a geometric given. In MTG, both time and geometry emerge together, encoded in the same projection dynamics that shape coherence, causality, and classical experience.

The result is a reformulation of quantum evolution governed not by absolute parameters, but by internal consistency across measurement events. Modular Hamiltonians encode observer-relative flow; effective geometry is sculpted from projection histories; and the arrow of time accumulates statistically through successive reductions of quantum uncertainty. This interpretive framework sets the stage for embedding MTG in broader physical contexts, including cosmology, quantum gravity, and holography, where time and geometry are deeply informational in origin.

\section{Embedding in the Standard Model}\label{sec:5}
To demonstrate compatibility with known particle physics, we now construct an embedding of the Standard Model (SM) into the Measurement-Induced Temporal Geometry (MTG) framework. In this formulation, matter and gauge fields are promoted to sections of fiber-covariant bundles, allowing them to couple directly to both time-fiber curvature and measurement-induced projection. This reinterprets familiar particle dynamics as governed jointly by Standard Model interactions and internal temporal geometry.

The key structure is a unified principal bundle \( P \to M \) whose structure group combines the internal time symmetry with the Standard Model gauge group:
\begin{equation}
G_{\mathrm{unified}} = G_{\mathrm{time}} \times G_{\mathrm{SM}},
\end{equation}
where \( G_{\mathrm{time}} \) is the symmetry group of the internal time fibers—typically taken as a compact Abelian or noncompact subgroup such as \( U(1)^n \) or \( SL(2,\mathbb{R}) \)—and
\begin{equation}
G_{\mathrm{SM}} = SU(3)_C \times SU(2)_L \times U(1)_Y
\end{equation}
is the Standard Model gauge group. All matter fields are realized as sections of associated bundles transforming under representations of both factors in \( G_{\mathrm{unified}} \).

Each Standard Model field \( \Psi_i \) thus becomes a section of a tensor product bundle:
\begin{equation}
\Psi_i \in \Gamma\left(E_{\mathrm{time}} \otimes E_{\mathrm{SM}, i}\right),
\end{equation}
where \( E_{\mathrm{time}} \) encodes the internal temporal representation, and \( E_{\mathrm{SM}, i} \) is the usual vector bundle for particle species \( i \). The corresponding fiber-covariant derivative includes both time and SM gauge connections:
\begin{equation}
\mathcal{D}_\mu \Psi_i = \partial_\mu \Psi_i + A_\mu^{\mathrm{time}} \cdot \Psi_i + A_\mu^{\mathrm{SM}} \cdot \Psi_i,
\end{equation}
with \( A_\mu^{\mathrm{time}} \in \mathfrak{g}_{\mathrm{time}} \) and \( A_\mu^{\mathrm{SM}} \in \mathfrak{g}_{\mathrm{SM}} \).

This unified connection ensures that all matter and gauge fields experience both standard interactions and curvature induced by internal time geometry. In regions where projection occurs, the internal time components of \( \Psi_i \) are subject to collapse, yielding decoherence and observer-relative causal structure consistent with the MTG formalism of Section~\ref{sec:4}.

In what follows, we analyze the embedding by first specifying the temporal representations associated with Standard Model matter fields, then deriving their covariant dynamics under the unified gauge structure, and finally examining how measurement-induced projection modifies scattering processes, mass generation, and vacuum stability within the fiber-extended framework.
%%%%%%%%%%%%%%%%%%%%%
\subsection{Gauge Coupling to Standard Model Fields}
To embed the Standard Model within the MTG framework, we extend the internal gauge structure to unify temporal symmetries with Standard Model interactions. This is achieved by enlarging the principal bundle over spacetime to include the internal time group alongside the familiar gauge symmetries. The unified structure group takes the form
\begin{equation}
G_{\mathrm{unified}} = G_{\mathrm{time}} \times SU(3)_C \times SU(2)_L \times U(1)_Y,
\end{equation}
where \( G_{\mathrm{time}} \) corresponds to the symmetry group of the internal time fibers \( T_x \) in the bundle \( \pi: E \to M \). Depending on the geometry of these fibers, one may take \( G_{\mathrm{time}} \) to be a compact Abelian group such as \( U(1)_{\mathrm{mod}} \), or a non-Abelian group such as \( SU(1,1) \), \( SL(2,\mathbb{R}) \), or \( Sp(1,\mathbb{H}) \), particularly in cases involving modular or quaternionic time structure.

Let \( \mathcal{P} \to M \) denote the principal \( G_{\mathrm{unified}} \)-bundle. Matter fields are then described as sections of fiber-covariant associated bundles,
\begin{equation}
\Psi \in \Gamma(E_{\mathrm{time}} \otimes E_{\mathrm{SM}}),
\end{equation}
where \( E_{\mathrm{time}} \) carries a representation \( \rho_{\mathrm{time}} \) of the internal time-fiber group, and \( E_{\mathrm{SM}} \) carries the usual representation of the Standard Model gauge group. The total connection on \( \mathcal{P} \) takes values in the direct sum Lie algebra
\begin{equation}
\mathfrak{g}_{\mathrm{unified}} = \mathfrak{g}_{\mathrm{time}} \oplus \mathfrak{su}(3)_C \oplus \mathfrak{su}(2)_L \oplus \mathfrak{u}(1)_Y,
\end{equation}
and decomposes in components as
\begin{equation}
A_\mu = A^{\mathrm{time}}_\mu + G^a_\mu T^{(3)}_a + W^i_\mu T^{(2)}_i + B_\mu Y,
\end{equation}
where \( T^{(3)}_a \in \mathfrak{su}(3)_C \), \( T^{(2)}_i \in \mathfrak{su}(2)_L \), and \( Y \in \mathfrak{u}(1)_Y \) are the Lie algebra generators in the representation appropriate to \( \Psi \). The corresponding field strengths decompose as
\begin{equation}
F_{\mu\nu} = F^{\mathrm{time}}_{\mu\nu} + G^a_{\mu\nu} T^{(3)}_a + W^i_{\mu\nu} T^{(2)}_i + B_{\mu\nu} Y.
\end{equation}
The temporal curvature \( F^{\mathrm{time}}_{\mu\nu} \in \mathfrak{g}_{\mathrm{time}} \) governs modular flow, entanglement transport, and coherence dynamics, while the Standard Model curvature components preserve their usual physical interpretation~\cite{casini2011entanglement}.

Standard Model fields embed fiber-covariantly within this framework~\cite{connes1996standard, chamseddine2007spectral}. Left-handed leptons and quarks appear as
\begin{align}
L &= (\nu_L, e_L) \in \Gamma\left(E_{\mathrm{time}} \otimes \mathbb{C}^1 \otimes \mathbb{C}^2_{-1/2}\right), \\
Q &= (u_L, d_L) \in \Gamma\left(E_{\mathrm{time}} \otimes \mathbb{C}^3 \otimes \mathbb{C}^2_{1/6}\right),
\end{align}
where the subscripts denote hypercharge \( Y \), and the tensor factors represent color, weak isospin, and internal time symmetry.

Right-handed fermions appear as \( SU(2)_L \) singlets but still transform under the time-fiber group:
\begin{align}
e_R &\in \Gamma\left(E_{\mathrm{time}} \otimes \mathbb{C}^1 \otimes \mathbb{C}^1_{-1}\right), \\
u_R &\in \Gamma\left(E_{\mathrm{time}} \otimes \mathbb{C}^3 \otimes \mathbb{C}^1_{2/3}\right), \\
d_R &\in \Gamma\left(E_{\mathrm{time}} \otimes \mathbb{C}^3 \otimes \mathbb{C}^1_{-1/3}\right).
\end{align}

The Higgs doublet is treated as a scalar section coupled to temporal geometry:
\begin{equation}
H \in \Gamma\left(E_{\mathrm{time}} \otimes \mathbb{C}^1 \otimes \mathbb{C}^2_{1/2}\right),
\end{equation}
allowing the electroweak potential to depend on internal coherence. Coherence minima and modular vacua can thus influence spontaneous symmetry breaking and the structure of phase transitions.
\begin{figure}[ht]
\centering
\begin{tikzpicture}[scale=1.2]

% Base manifold M
\draw[very thick] (-3.5,0) -- (3.5,0);
\node at (-3.8,0) {\small $M$};

% Fibers and SM fields
\foreach \x/\label in {-3/$x_1$, -1.5/$x_2$, 0/$x_3$, 1.5/$x_4$, 3/$x_5$} {
    \draw[gray!60] (\x,0) -- (\x,2.3);
    \node[below] at (\x,0) {\tiny \label};
}

% Time field tau in fiber
\foreach \x/\tauval in {-3/1.3, -1.5/1.6, 0/2.0, 1.5/1.4, 3/1.8} {
    \draw[->, thick, blue!70!black] (\x,0) -- (\x,\tauval)
        node[pos=1.0, right] {\tiny $\tau(x)$};
}

% SM gauge fields in base
\draw[->, thick, red!70!black] (-2.3,0.2) -- (-1.3,0.2);
\node[above] at (-1.8,0.25) {\tiny $A_\mu(x)$};

% SM fermion field in base
\draw[->, thick, orange!90!black] (0.2,0.2) -- (1.2,0.2);
\node[above] at (0.7,0.25) {\tiny $\psi(x)$};

% Lifting arrows into fiber
\draw[->, thick, violet!80!black] (-1.8,0.2) -- (-1.5,1.6);
\node at (-2.1,1.0) {\tiny Lift};
\draw[->, thick, violet!80!black] (1.0,0.2) -- (1.5,1.4);

% Labels in fiber space
\node at (0,2.5) {\small Total space $E$};
\draw[->, thick] (2.7,2.3) -- (2.7,0.1);
\node[right] at (2.7,1.2) {\small $\pi: E \to M$};

\end{tikzpicture}
\caption{Standard Model coupling via fiber lifting. Fields such as gauge potentials \( A_\mu(x) \) and fermions \( \psi(x) \), originally defined over the spacetime manifold \( M \), are promoted to live in the total space \( E \) of the internal time-fiber bundle. The internal time field \( \tau(x) \) and its connection \( \nabla \) enter the covariant dynamics via lifted interactions \( D_\mu \tau \), inducing measurement-coupled evolution.}
\label{fig:sm_fiber_lifting}
\end{figure}
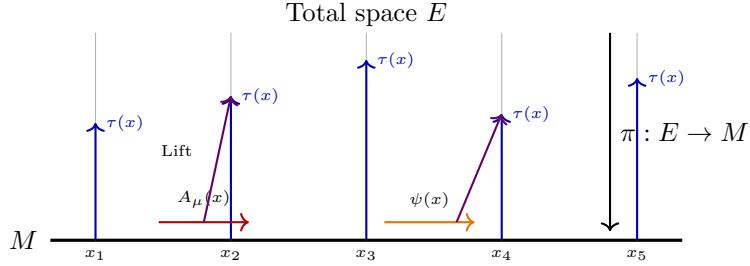
Neutrino fields are naturally accommodated. If right-handed neutrinos exist, they are given by
\begin{equation}
\nu_R \in \Gamma\left(E_{\mathrm{time}} \otimes \mathbb{C}^1 \otimes \mathbb{C}^1_0\right),
\end{equation}
admitting either Dirac or Majorana mass terms, depending on the symmetry properties of the temporal representation. Holonomies in \( G_{\mathrm{time}} \) may induce oscillation phases and generate effective mass splittings.

The gauge bosons themselves are components of the total connection \( A_\mu \), including gluons \( G^a_\mu \), weak bosons \( W^i_\mu \), the hypercharge field \( B_\mu \), and the temporal gauge field \( A^{\mathrm{time}}_\mu \). The total curvature \( F_{\mu\nu} \) encodes not only Standard Model field strengths but also coherence transport, entanglement flow, and measurement-induced deformation of the internal time geometry.

This embedding respects all gauge symmetries while enriching the SM field content with new couplings to internal time. It provides a geometric setting for interpreting decoherence, mass, and flavor dynamics as emergent phenomena tied to the curvature and projection structure of the fibered temporal manifold.
%%%%%%%%%%%%%%%%%%%%%%%%%%
\subsection{Fermion Representations and Coherence Dynamics}
Fermionic matter fields are promoted from conventional spacetime spinors to sections of fiber-covariant bundles that couple simultaneously to Standard Model gauge fields and the internal time-fiber connection. This extension allows temporal curvature and projection dynamics to modulate coherence, flavor structure, and mass generation within a unified geometrical setting.

Let \( \mathcal{S} \to M \) denote the spinor bundle associated with the spacetime spin structure, and let \( E \to M \) be the internal time-fiber bundle with structure group \( G_{\mathrm{time}} \). The full gauge symmetry is encoded in the unified group
\begin{equation}
G = G_{\mathrm{time}} \times SU(3)_C \times SU(2)_L \times U(1)_Y,
\end{equation}
and fermion fields are represented as sections of the associated fiber-covariant spinor bundle
\begin{equation}
\Psi \in \Gamma(\mathcal{S} \otimes \mathcal{V}),
\end{equation}
where \( \mathcal{V} \) carries a representation of \( G \), including both temporal and Standard Model charges. Each Standard Model fermion—quark or lepton, left- or right-handed—is extended to such a section with a prescribed internal time representation \( \rho_{\mathrm{time}} \).

The total gauge connection acting on \( \Psi \) decomposes as
\begin{equation}
A_\mu = A_\mu^{\mathrm{time}} + A_\mu^{\mathrm{SM}},
\end{equation}
where \( A_\mu^{\mathrm{SM}} = G^a_\mu T^a + W^i_\mu T^i + B_\mu Y \) is the conventional Standard Model connection and \( A_\mu^{\mathrm{time}} \in \mathfrak{g}_{\mathrm{time}} \) encodes the geometry of internal time. The corresponding fiber-covariant derivative is defined by
\begin{equation}
\mathrm{D}_\mu \Psi = \partial_\mu \Psi + A_\mu^{\mathrm{time}} \cdot \Psi + A_\mu^{\mathrm{SM}} \cdot \Psi,
\end{equation}
governing both gauge interactions and transport across the internal time fiber. The temporal curvature \( F_{\mu\nu}^{\mathrm{time}} \) enters directly into fermionic dynamics, modifying phase transport, coherence flow, and sensitivity to measurement.

The fermionic action takes the standard Dirac form using the covariant operator \( \slashed{\mathrm{D}} = \gamma^\mu \mathrm{D}_\mu \):
\begin{equation}
\mathcal{L}_{\mathrm{ferm}} = \bar{\Psi} i \slashed{\mathrm{D}} \Psi = \bar{\Psi} i \gamma^\mu (\partial_\mu + A_\mu^{\mathrm{time}} + A_\mu^{\mathrm{SM}}) \Psi.
\end{equation}
This Lagrangian is invariant under the full symmetry group \( G \), and encodes the dynamical influence of both Standard Model forces and temporal geometry. The time-fiber connection affects fermion evolution in several critical ways.

First, temporal curvature induces path-dependent coherence phases. Fermion fields accumulate nontrivial holonomy around loops in internal time, with possible implications for oscillation and interference phenomena. Second, measurement-induced projection acts as a localized collapse mechanism, selectively decohering components of \( \Psi \) based on their alignment with the observer frame \( n^\mu \), as introduced in the measurement Lagrangian. Third, a nontrivial time-fiber curvature \( F_{\mu\nu}^{\mathrm{time}} \) may induce apparent flavor symmetry distortions by altering coherence transport across the fiber, even when the Standard Model gauge sector remains dynamically intact. These distortions reflect geometric backreaction through the time-fiber connection, rather than explicit symmetry violation in the fermionic Lagrangian itself.

Because fermions are transported via the full MTG connection, their dynamics become sensitive to the informational geometry of coherence~\cite{baumgratz2014quantifying}. Holonomies in \( G_{\mathrm{time}} \), especially in topologically nontrivial settings, can lead to CPT-violating phases, modified mixing matrices, and observer-relative decoherence effects. These signatures are negligible in flat regions of the internal time bundle, but may become significant in high-curvature domains, such as early-universe cosmology, black hole interiors, or laboratory systems engineered for entanglement precision.

In this framework, Standard Model fermions are not merely gauge fields on spacetime; they are fibered entities dynamically entangled with the geometry of measurement. Their masses, flavors, and coherence properties are shaped not only by symmetry but by the flow of internal time. The next subsection incorporates this geometric structure into the Higgs sector, where measurement-induced collapse gives rise to effective mass generation.
%%%%%%%%%%%%%%%%%%%%%%%%%%%%%%%%%%%
\subsection{Higgs Field and Temporal Symmetry Breaking}
In the MTG framework, the Higgs field is elevated from a conventional scalar doublet to a fiber-covariant section that couples not only to Standard Model gauge fields but also to the internal time-fiber connection. This coupling introduces a geometric mechanism for mass generation, in which symmetry breaking arises from projection-induced localization in internal time~\cite{giacomini2019quantum, parker2009mass}.

The Higgs field is modeled as a scalar-valued section
\begin{equation}
H \in \Gamma(E_{\mathrm{time}} \otimes \mathbb{C}^1 \otimes \mathbb{C}^2_{1/2}),
\end{equation}
where \( E_{\mathrm{time}} \) carries a representation of \( G_{\mathrm{time}} \), and the remaining factors represent hypercharge and weak isospin as in the Standard Model. The field \( H \) is thus sensitive to both electroweak gauge symmetry and the geometry of internal time.

The corresponding fiber-covariant derivative is
\begin{equation}
\mathrm{D}_\mu H = \partial_\mu H + A_\mu^{\mathrm{time}} \cdot H + A_\mu^{\mathrm{SM}} \cdot H,
\end{equation}
ensuring that the kinetic term
\begin{equation}
\mathcal{L}_{\mathrm{H,kin}} = \langle \mathrm{D}_\mu H, \mathrm{D}^\mu H \rangle
\end{equation}
is invariant under the full symmetry group \( G_{\mathrm{time}} \times SU(2)_L \times U(1)_Y \).

The Higgs potential is modified to reflect dependence on the internal time field:
\begin{equation}
\mathcal{L}_{\mathrm{H,pot}} = - V(H), \quad V(H) = \lambda_H \left( \langle H, H \rangle - v^2(\tau) \right)^2.
\end{equation}
Here, the vacuum expectation value \( v(\tau) \) is promoted to a function of internal time, allowing for coherence-dependent mass generation. In measurement-active regions—where projection collapses temporal degrees of freedom—the function \( v(\tau) \) becomes sharply localized, dynamically triggering electroweak symmetry breaking. Holonomies or modular alignment in \( G_{\mathrm{time}} \) may induce further variation across regions or observer frames.

The complete Higgs sector Lagrangian becomes
\begin{equation}
\mathcal{L}_{\mathrm{Higgs}} = \langle \mathrm{D}_\mu H, \mathrm{D}^\mu H \rangle - \lambda_H \left( \langle H, H \rangle - v^2(\tau) \right)^2,
\end{equation}
with the symmetry-breaking scale \( v(\tau) \) now linked to the local coherence structure of internal time. Projection, therefore, does not merely constrain dynamics—it induces mass by localizing the vacuum structure.

Fermion masses arise through Yukawa interactions:
\begin{equation}
\mathcal{L}_{\mathrm{Yukawa}} = y_{ij} \bar{\Psi}_i H \Psi_j + \text{h.c.},
\end{equation}
where the coefficients \( y_{ij} \) may themselves reflect structure inherited from temporal geometry—such as holonomy phases, modular transitions, or topological indices. Once the Higgs field acquires a projection-aligned expectation value, fermions become massive relative to that coherence frame.

This approach contrasts with the conventional view in which mass emerges from minimizing a static potential. In MTG, mass arises relationally—from the collapse of coherence, the breaking of temporal symmetry, and the informational constraints of measurement. It provides a geometric, observer-relative account of symmetry breaking, compatible with decoherence and responsive to the entangled structure of internal time.

In the following subsection, we turn to gauge bosons and explore how similar projection-aligned mechanisms yield mass, propagate temporal alignment, and modify interaction structure in regions of curved internal time.
%%%%%%%%%%%%%%%%%%%%%%%%%%%%%%%
\subsection{Flavor Mixing from Temporal Holonomy}
In the Standard Model, the replication of fermion generations and the structure of flavor mixing matrices—such as the CKM and PMNS matrices—are empirical features lacking a fundamental explanation. Within the MTG framework, these features arise naturally from the topology, holonomy, and curvature of the internal time-fiber bundle. This geometric reinterpretation links flavor dynamics to the informational structure of projection and coherence.

Let \( \pi: E \to M \) denote the internal time-fiber bundle with typical fiber \( T_x \) a compact or noncompact manifold, such as a torus, hyperbolic surface, or quaternionic sphere. The structure group \( G_{\mathrm{time}} \) governs how fibers are patched across spacetime, and the total bundle \( E \) may possess nontrivial topology encoded in \( \pi_1(E) \), the fundamental group of the total space.

We propose that each fermion generation corresponds to a distinct topological sector of \( E \), labeled by a non-contractible loop class \( [\gamma_i] \in \pi_1(E) \). These classes capture transition functions between fiber charts and define equivalence sectors for internal time holonomy. The replication of fermions across three generations thus reflects the multiplicity of inequivalent holonomy sectors in temporal geometry. In this picture, generation number becomes a topological quantum number—emergent from the fibered structure of time.

Flavor mixing is interpreted as parallel transport across coherence sectors~\cite{tomaz2025rqd}. For a homotopy class \( \gamma_{ij} \) interpolating between sectors \( i \) and \( j \), the transition amplitude is governed by the holonomy of the internal time connection:
\begin{equation}
U_{ij} := \operatorname{Hol}_\nabla(\gamma_{ij}) = \mathcal{P} \exp \left( \int_{\gamma_{ij}} A^{\mathrm{time}}_\mu \, dx^\mu \right).
\end{equation}
To ensure that these holonomies \( U_{ij} \) define physically meaningful flavor mixing matrices—such as the CKM or PMNS matrices—they must correspond to unitary operators. This requires that the internal time connection \( A_\mu^{\text{time}} \) take values in a unitary Lie algebra, such as \( \mathfrak{u}(n) \) or \( \mathfrak{su}(n) \). We therefore assume that the structure group \( G_{\text{time}} \) is a compact Lie group, ensuring that holonomy transport preserves inner products and defines unitary parallel evolution.

Under this assumption, the matrix \( [U_{ij}] \) constructed from holonomy paths defines an observer-dependent flavor mixing matrix consistent with quantum interference phenomena. Non-unitary extensions—such as those arising from noncompact groups like \( \mathrm{SL}(2,\mathbb{R}) \)—could introduce damping, leaky coherence, or emergent sterile sectors, and may describe open-system dynamics beyond the standard oscillation paradigm.

The full flavor mixing matrix arises as a direct sum over such holonomy operators:
\begin{equation}
U_{\mathrm{CKM}} \sim \bigoplus_{i,j} \operatorname{Hol}_\nabla(\gamma_{ij}).
\end{equation}
This construction provides a geometric and gauge-invariant origin for mixing phases and angles, grounded in the entanglement structure of internal time.

The same mechanism governs neutrino oscillations. In the presence of nonzero temporal curvature \( F^{\mathrm{time}}_{\mu\nu} \), distinct flavor components accumulate different geometric phases along their propagation. The modular-time evolution of a neutrino state is given by
\begin{equation}
|\nu_\alpha(t)\rangle = \sum_\beta U_{\alpha\beta}(t) |\nu_\beta\rangle,
\end{equation}
where \( U_{\alpha\beta}(t) \) includes both conventional Hamiltonian evolution and MTG holonomy corrections. These corrections may alter oscillation frequencies, interfere with phase coherence, or shift coherence lengths, particularly in regions of high curvature or nontrivial measurement structure.

Projection-induced collapse further modulates flavor dynamics. Each projection locally collapses the internal time field \( \tau \), altering coherence and affecting flavor transport. In regions where projection events are sparse or aligned, coherence is preserved and flavor oscillations persist. In contrast, dense or irregular projections induce decoherence, suppressing oscillations and favoring eigenstate localization. The decoherence profile depends on the observer-relative projection density \( \rho(x) \), which encodes the measurement history across spacetime.

In this framework, mixing matrices become dynamical, renormalized by curvature and modulated by measurement. Flavor is no longer a static symmetry label but a coherence-dependent structure shaped by topology, holonomy, and projection. The number of effective generations may vary across temporal regimes, and novel phenomena—such as topology-driven oscillation suppression or curvature-induced mixing enhancements—may emerge.

These results complete the fiber-covariant embedding of Standard Model matter within MTG. The structure of particle families, mixing, and mass is unified under a geometric framework in which internal time governs coherence, projection, and informational transport. In the next section, we turn to the emergence of gravitational dynamics from these same principles, linking spacetime geometry to the informational flow induced by measurement.
%%%%%%%%%%%%%%%%%%%%
\subsection{Summary: Measurement Geometry and the Standard Model}
The MTG framework embeds the full structure of the Standard Model within a fiber-covariant geometry of internal time. All matter and gauge fields evolve not only under conventional interactions but also within a dynamically curved temporal bundle, where coherence, holonomy, and projection govern informational flow. This unification of gauge dynamics and temporal geometry provides a geometric foundation for mass generation, flavor structure, and decoherence.

The extended gauge group \( G = G_{\mathrm{time}} \times SU(3)_C \times SU(2)_L \times U(1)_Y \) defines a unified connection \( A_\mu \), whose curvature \( F_{\mu\nu} \) encompasses both conventional field strengths and internal time dynamics. This connection mediates the transport of quantum fields across spacetime and internal time fibers, modulating not only gauge evolution but also coherence propagation through observer-dependent geometry.

Fermions are promoted to fiber-covariant spinor fields, transforming under representations of both Standard Model and temporal symmetry groups. Their dynamics are governed by a generalized Dirac operator that includes parallel transport along the time-fiber connection. As a result, fermion wavefunctions accumulate holonomy-induced phases and become sensitive to measurement-induced projection. Coherence collapse selectively decoheres fiber components aligned with the observer frame \( n^\mu \), while nontrivial holonomy gives rise to flavor mixing. In this framework, generation structure is linked to topological sectors of the time bundle.

The Higgs sector is similarly embedded. The Higgs field becomes a bifundamental scalar, mediating between temporal and electroweak symmetry spaces. Its potential acquires coherence-dependent structure, and symmetry breaking arises locally from projection events that collapse internal time directions. Yukawa couplings contract over fiber indices, allowing mass to emerge dynamically through the interplay of coherence transport and measurement-induced collapse.

Flavor structure is reinterpreted geometrically. Generations correspond to distinct fiber sectors or homotopy classes within the time bundle, while mixing arises from parallel transport along coherence paths. When \( G_{\mathrm{time}} \) is compact and unitary, holonomy defines unitary flavor transformations consistent with observed CKM and PMNS structure. In noncompact settings, holonomy may induce non-unitary mixing, with possible phenomenological implications such as decoherence or sterile flavor leakage.

Measurement acts as a dynamical agent that modifies coherence, generates effective mass, and filters symmetry via observer alignment. In domains of strong projection or pronounced temporal curvature, deviations from Standard Model behavior may emerge. Observable effects include curvature-sensitive mixing angles, flavor decoherence, and phase shifts induced by informational geometry. Such signatures could be probed in neutrino oscillation experiments, early-universe processes, or engineered quantum systems.

By embedding the Standard Model within a measurement-sensitive temporal geometry, MTG offers a unified framework in which matter, mass, and flavor arise from the topological and informational structure of internal time. In the following section, this framework is extended to gravitational dynamics, where spacetime curvature itself emerges as an entropic signature of coherence loss under measurement.

\section{Gravity as Emergent Entropic Geometry}\label{sec:6}
In conventional field theory, gravity is introduced as a fundamental interaction mediated by a massless spin-2 field, with dynamics governed by the Einstein--Hilbert action. In the MTG framework, by contrast, gravity arises not as a fundamental force but as an emergent, entropic response to the informational geometry induced by quantum measurement. Spacetime curvature, in this view, reflects the statistical structure of projection events and the resulting decoherence of the internal time field \( \tau \).

Throughout the preceding sections, we have shown how projection events act on \( \tau \) to define causal order, temporal flow, and an effective spacetime metric \( g^{\mathrm{eff}}_{\mu\nu}(x) \) constructed from the integrated structure of classicalized measurement outcomes. This metric, while emergent, obeys dynamical laws that encode gravitational behavior as a consequence of entropic balance over ensembles of quantum histories.

In this section, we formalize the emergence of gravitational dynamics from measurement statistics. Beginning with the definition of a projection-induced entropy functional, we show that variation of this entropy under constraints yields equations of motion for the effective geometry. These equations take the form of modified Einstein field equations, where the source is not conventional matter stress-energy but an entropic flux arising from coherence collapse and measurement density.

This formulation preserves the relational and information-theoretic essence of MTG, while recovering classical gravity as a limiting case. It further predicts novel regimes---such as measurement-induced curvature singularities, coherence horizons, and entropy-driven inflationary expansion---that arise from the microscopic geometry of projection. We begin by defining the projection entropy functional that governs emergent gravitational response.
%%%%%%%%%%%%%%%%%%%%%%%%%%%%%%%
\subsection{Projection Entropy and Emergent Dynamics}\label{sec:6.1}
To derive gravitational behavior from the informational structure of measurement, the MTG framework introduces an entropy functional that quantifies coherence loss induced by projection. Each measurement of the internal time field \( \tau \) collapses the local fiber configuration to a classical direction \( \mu[\tau(x)] \), thereby reducing the system’s entanglement and internal degrees of freedom. The cumulative effect of these projection events defines an emergent, coarse-grained geometry, encoded by the effective metric \( g^{\mathrm{eff}}_{\mu\nu}(x) \), introduced in Section~\ref{sec:2}.

In earlier sections, this metric served as a kinematic background constructed from coherence and projection data. Here, it acquires dynamical significance: we treat \( g^{\mathrm{eff}}_{\mu\nu} \) as a variational field determined by extremizing the projection entropy subject to consistency with its projection-induced definition. This approach links gravitational dynamics directly to quantum decoherence, positioning spacetime curvature as a statistical response to the informational structure of measurement.

We now show that gravitational dynamics in MTG arise from extremizing an entropy functional that quantifies coherence loss under projection. The projection entropy is defined by
\begin{equation}
S_{\mathrm{proj}}[\rho, \tau] = \int_M \rho(x) \, S(\tau(x)) \, \mathrm{d}^d x,
\quad\text{where}\quad
S(\tau(x)) := -\operatorname{Tr}[\rho_{\text{fiber}}(x) \log \rho_{\text{fiber}}(x)],
\end{equation}
and \( \rho_{\text{fiber}}(x) \) is the reduced density matrix obtained by tracing out all but the internal time degrees of freedom at point \( x \). This entropy functional quantifies the local information loss induced by measurement, weighted by the projection density \( \rho(x) \), and encodes the cumulative effect of coherence collapse.

In earlier sections, the effective metric \( g^{\mathrm{eff}}_{\mu\nu}(x) \) was introduced as a kinematic construct derived from coherence vector fields and projection density. Here, it becomes dynamical: we treat it as an emergent field constrained to match the measurement-induced geometry. Gravitational equations in MTG then arise by extremizing the projection entropy subject to this consistency condition, linking spacetime curvature directly to informational loss.

To make the entropy functional well-defined, we interpret \( \tau(x) \) not merely as a classical field configuration, but as a parameter that determines a reduced density matrix \( \rho_{\text{fiber}}(x) \) supported on the fiber \( \pi^{-1}(x) \). Specifically, we assume that \( \rho_{\text{fiber}}(x) \) arises via a partial trace over all other spacetime points in a global quantum state \( \rho_{\text{global}}[\tau] \), conditioned on the local value of \( \tau(x) \):
\[
\rho_{\text{fiber}}(x) := \operatorname{Tr}_{M \setminus \{x\}} \left( \rho_{\text{global}}[\tau] \right).
\]
The local entropy \( S(\tau(x)) \) thus measures the information lost under projection at \( x \), reflecting how the internal time field becomes classicalized through decoherence. The global entropy functional \( S_{\mathrm{proj}} \) quantifies the total degradation of coherence across spacetime. Low entropy corresponds to tight alignment between the internal time field and the observer-induced projection direction; high entropy signals delocalized coherence and increased classical uncertainty.

To derive dynamics from this entropy, we apply a variational principle constrained by the geometry of measurement-induced projection. We impose that the effective metric \( g^{\mathrm{eff}}_{\mu\nu}(x) \), defined earlier in terms of coherence vectors and projection density, remains consistent with its original construction:
\begin{equation}
g^{\mathrm{eff}}_{\mu\nu}(x) = \eta_{\mu\nu} + \kappa \int_{\Sigma_x} \mu[\tau(y)]_\mu \, \mu[\tau(y)]_\nu \, \rho(y) \, \mathrm{d}\Sigma(y).
\end{equation}
Introducing a symmetric Lagrange multiplier field \( \lambda^{\mu\nu}(x) \), we define the total constrained functional
\begin{align}
\mathcal{A}[\rho, \tau, \lambda] := S_{\mathrm{proj}}[\rho, \tau]
- \int_M \lambda^{\mu\nu}(x) \Big[
g^{\mathrm{eff}}_{\mu\nu}(x) - \eta_{\mu\nu} - \kappa \int_{\Sigma_x} \mu[\tau(y)]_\mu \mu[\tau(y)]_\nu \rho(y) \, \mathrm{d}\Sigma(y)
\Big] \mathrm{d}^d x.
\end{align}
We then require stationarity:
\begin{equation}
\delta \mathcal{A}[\rho, \tau, \lambda] = 0,
\end{equation}
where the variation is taken with respect to \( \rho \) and \( \tau \), and the effective metric is understood as a dependent functional. This variational principle yields coupled equations of motion for \( \rho \), \( \tau \), and \( g^{\mathrm{eff}} \), enforcing consistency between geometric curvature and the informational flux generated by measurement.

The resulting dynamics resemble modified Einstein equations:
\begin{equation}
G_{\mu\nu}[g^{\mathrm{eff}}] = 8\pi G \, T^{\mathrm{meas}}_{\mu\nu},
\end{equation}
where \( G_{\mu\nu} \) is the Einstein tensor derived from \( g^{\mathrm{eff}} \), and \( T^{\mathrm{meas}}_{\mu\nu} \) is an effective stress-energy tensor generated by the entropy gradient of \( \rho(x) \). This stress-energy is not sourced by conventional matter, but by the statistical backreaction of projection-induced coherence loss.

In this framework, gravity emerges as a thermodynamic response to the degradation of internal quantum coherence. Projection events act as informational sources of curvature, encoding observer-relative measurement into the causal structure of spacetime. The entropy functional thus plays a role analogous to the Einstein--Hilbert action, with spacetime curvature reflecting the interplay between coherence localization and entropic dissipation.

This entropic origin of gravity resonates with other emergent paradigms, including Jacobson's thermodynamic derivation of Einstein equations, entanglement-induced gravity, and holographic duality~\cite{jacobson1995thermodynamics, van2010building}. However, MTG provides a distinct mechanism; geometry is not merely dual to entanglement, it is constructed from it through projection. The effective metric arises from relational coherence structure, filtered by observer-dependent measurement, and encodes curvature as a statistical trace of decoherence.
%%%%%%%%%%%%%%%%%%
\subsection{Static and Cosmological Solutions}
The entropic field equations derived in the MTG framework admit a broad class of solutions that depend on the structure of projection density \( \rho(x) \), internal time coherence, and observer-relative modular dynamics. In this section, we explore the physical consequences of these solutions in both static and cosmological settings, showing how classical curvature emerges from the informational flow of measurement.

In regions where measurement events are sparse and the internal time field \( \tau \) evolves slowly across spacelike hypersurfaces, coherence is well maintained. The projection density \( \rho(y) \) can be approximated as locally constant, and the coherence directions \( \mu[\tau(y)] \) vary negligibly across the integration region. Under these conditions, the emergent metric becomes a perturbative deformation of Minkowski space:
\begin{equation}
g^{\mathrm{eff}}_{\mu\nu}(x) = \eta_{\mu\nu} + \kappa \int_{\Sigma_x} \mu[\tau(y)]_\mu \, \mu[\tau(y)]_\nu \, \rho(y) \, \mathrm{d}\Sigma(y). \tag{125}
\end{equation}
Here \( \Sigma_x \subset M \) is an observer-dependent hypersurface in the causal past of \( x \), often taken to be a segment of the past light cone. The integrand defines a rank-one symmetric tensor field whose accumulation reflects coherence localization in the measurement geometry.

In this regime, the entropic stress-energy tensor \( T^{\mathrm{meas}}_{\mu\nu} \) behaves analogously to a vacuum polarization term, modulated by the decay of coherence in the time-fiber bundle. The resulting geometry resembles a weak-field solution sourced not by conventional mass or energy, but by the informational structure of measurement. These curvature perturbations may influence quantum coherence or temporal interference phenomena and could yield testable signatures in precision-controlled or engineered measurement regimes.

At cosmological scales, the decay of internal coherence gives rise to large-scale curvature through entropic backreaction. When \( \rho(t) \) is homogeneous and isotropic in modular time, the emergent metric acquires a Friedmann–Lemaître–Robertson–Walker (FLRW) form:
\begin{equation}
\mathrm{d}s^2 = -\mathrm{d}t^2 + a(t)^2 \, \mathrm{d}\vec{x}^2,
\end{equation}
with the scale factor \( a(t) \) satisfying a modified Friedmann equation,
\begin{equation}
\left( \frac{\dot{a}}{a} \right)^2 = \frac{8\pi G}{3} \rho_{\mathrm{meas}}(t),
\end{equation}
where \( \rho_{\mathrm{meas}}(t) \) is the entropy density generated by projection events per unit modular time. In this context, cosmological inflation arises from an early phase of high projection density, driving exponential expansion due to the energetic cost of coherence collapse. As the projection rate diminishes, the expansion slows, and the system transitions to classical behavior. Reheating corresponds to the thermal redistribution of entanglement across fiber sectors, without requiring a fundamental inflaton field.

A further implication of this entropic dynamics is the formation of measurement-induced horizons. In domains where projection density becomes sharply localized—either via coherent focusing or external manipulation—the entropy flux diverges. The resulting emergent geometry develops causal boundaries beyond which internal time coherence cannot propagate. These {\itshape coherence horizons} are informational in nature: they demarcate regions where projection has destroyed the possibility of modular flow, rather than forming geometric singularities in the usual sense~\cite{debianchi2024achronotopic}. Nevertheless, in extreme cases where \( \tau \) collapses to a constant eigenstate over an open region, modular flow breaks down entirely, signaling a temporal singularity—a boundary of causal structure rather than of spacetime.

Taken together, these solutions illustrate the range of gravitational behavior encoded in the MTG framework. Static observers generate local curvature via coherence collapse, while global entropic decay drives cosmological expansion. In high-curvature or high-projection regimes, causal structure itself becomes observer-relative and informationally defined. This unifies gravitational response with the statistics of measurement, providing a predictive reformulation of geometry grounded in the dynamics of internal time.

While the analysis here focuses on local weak-field and stationary regimes, the MTG formalism also admits cosmological solutions. In particular, the cumulative projection count across spatial hypersurfaces defines an emergent scale factor, capturing the statistical expansion of decoherent structure. This construction replaces the Friedmann equation with a relational growth law rooted in measurement history. A detailed formulation is presented in Section~\ref{sec:7}, where temporal asymmetry and gravitational entropy are derived from projection statistics.
%%%%%%%%%%%%%%%%%%%%%%%%%%%
\subsection{Holography and Surface Entropy}

One of the central insights from black hole thermodynamics and string theory is the holographic principle--the notion that the informational content of a spatial region is encoded on its boundary. In the MTG framework, this principle emerges naturally from the entropy generated by measurement-induced collapse and the modular geometry of internal time.

Each projection event reduces fiber entanglement, collapsing the internal time field \( \tau \) to a classical configuration \( \mu[\tau(x)] \). This process not only introduces local curvature but also imposes a constraint on surface entropy along hypersurfaces aligned with an observer’s modular time. As a result, the effective dynamics of a bulk region \( \mathcal{R} \subset M \) can be determined from the accumulated informational flux across its boundary \( \partial \mathcal{R} \), encoded entirely in projection data.

To formalize this, we consider a modular foliation of spacetime by hypersurfaces \( \Sigma_\lambda \), each labeled by an observer-relative modular parameter \( \lambda \). The entropy across a slice \( \Sigma_\lambda \) is given by
\begin{equation}
S_\lambda = \int_{\Sigma_\lambda} \rho(x) \, \mathcal{S}(\tau(x)) \, \mathrm{d}\Sigma,
\end{equation}
where \( \rho(x) \) denotes the measurement density and \( \mathcal{S}(\tau(x)) \) is the von Neumann entropy of the fiber configuration at \( x \). This surface entropy plays the role of an informational boundary term: rather than residing in the bulk Lagrangian, it drives geometry through its modular-time variation~\cite{takayanagi2025holographic}.

The evolution between successive slices \( \Sigma_\lambda \) is governed by a modular Hamiltonian \( K_{\mathrm{mod}}(\lambda) \), which encodes the entanglement flow across each hypersurface. In MTG, this modular operator also determines the evolution of the emergent geometry via the relation
\begin{equation}
\frac{\mathrm{d}}{\mathrm{d}\lambda} g^{\mathrm{eff}}_{\mu\nu}(\lambda) = -\frac{\delta K_{\mathrm{mod}}}{\delta g^{\mu\nu}(\lambda)}.
\end{equation}
This differential identity expresses the coupling between geometry and information: the flow of modular energy across a boundary induces curvature in the bulk to maintain coherence consistency~\cite{blanco2013relative, pastawski2015holographic}.

Given a set of projection data along the boundary \( \partial \mathcal{R} \)—including the collapsed directions \( \mu[\tau(x)] \), density \( \rho(x) \), and fiber orientation—it is possible to reconstruct the interior effective geometry. The emergent metric at any point \( x \in \mathcal{R} \) is given by
\begin{equation}
g^{\mathrm{eff}}_{\mu\nu}(x) = \eta_{\mu\nu} + \kappa \int_{\Sigma_x \subset \partial \mathcal{R}} \mu[\tau(y)]_\mu \, \mu[\tau(y)]_\nu \, \rho(y) \, \mathrm{d}\Sigma(y),
\end{equation}
where \( \Sigma_x \) is the causal portion of the boundary influencing point \( x \). This holographic reconstruction mirrors the logic of AdS/CFT duality, but without assuming a fundamental gravitational bulk—here, bulk geometry arises from entanglement reduction along the modular boundary~\cite{tomaz2025rqd}.

In highly curved settings such as black holes, the boundary \( \Sigma_\lambda \) defines a modular screen: a surface beyond which coherence cannot propagate due to overwhelming projection density. The entropy across such a screen obeys an area law,
\begin{equation}
S_{\mathrm{proj}} \leq \frac{\mathrm{Area}(\Sigma)}{4G},
\end{equation}
consistent with the Bekenstein--Hawking formula. However, unlike in conventional gravitational theories where this entropy counts microstates of a fundamental gravitational field, in MTG it reflects the information irretrievably lost through projection. The entropy arises from coherence collapse, not from hidden dynamical degrees of freedom~\cite{bakas2024modulartime}.

The MTG framework thus realizes the holographic principle as a consequence of measurement geometry. Projection-induced entropy across modular surfaces encodes the geometry of the interior, and modular flow governs the emergence of curvature via informational flux. Spacetime becomes a derived construct, arising not from gravitational fields but from fiber-wise information loss and relational entropy. This redefinition of surface entropy in terms of modular collapse offers a new route to holography, grounded in observer-aligned measurement dynamics and free of assumptions about quantum gravity microstates.

In the following section, we turn to experimental and observational consequences of this framework, examining how temporal coherence, projection-induced collapse, and modular geometry can be tested through particle physics, astrophysics, and quantum simulation.

\section{Cosmology from Measurement Dynamics}\label{sec:7}
The cosmological implications of the Measurement-Induced Temporal Geometry (MTG) framework extend the model's predictive reach from local measurement processes to the large-scale structure and evolution of the universe. In contrast to classical cosmology, which assumes a preexisting spacetime manifold equipped with a metric satisfying Einstein’s equations, MTG treats spacetime itself as an emergent statistical structure, dynamically reconstructed from a history of quantum measurements acting on an internal time field.

In this section, we reinterpret key features of standard cosmology; such as, the initial singularity, inflation, horizon structure, and dark energy, as consequences of projection geometry and entropy flow. The effective metric \( g^{\mathrm{eff}}_{\mu\nu} \) becomes time-dependent due to cosmologically varying measurement density \( \rho(x) \), and the large-scale dynamics are governed by the entropic Einstein equations derived in Section~\ref{sec:6}.

We begin by defining homogeneous and isotropic solutions of the MTG framework that reproduce Friedmann–Lema\^itre–Robertson–Walker (FLRW) behavior in the dense-measurement limit. We then explore how deviations from homogeneity and isotropy emerge naturally in sparse or topologically frustrated projection regimes, offering geometric explanations for cosmic variance, anisotropies, and the arrow of time.

Finally, we interpret the early universe; including the inflationary epoch and entropy production, in terms of coherence collapse and measurement-induced reheating, establishing an informational foundation for initial conditions and cosmic expansion.
%%%%%%%%%%%%%%%%%%%%%%%%%

\subsection{Origin of Time and Cosmic Expansion}

In the MTG framework, the classical notion of cosmological time—as a globally synchronized parameter flowing uniformly from a singular origin—is replaced by a relational, informational construct. Temporal flow is not assumed a priori; instead, it emerges dynamically as the internal time field \( \tau(x) \) undergoes localized projections into classical values through measurement. The act of projection defines the start of cosmic history, replacing initial conditions and simultaneity with an observer-dependent informational boundary.

Let \( \mu[\tau(x)] \) denote the projected classical direction of \( \tau \) at a spacetime point \( x \in M \). The origin of time corresponds to the first point \( x_0 \in M \) for which \( \mu[\tau(x_0)] \neq 0 \), marking the first coherence-to-classical transition in the fiber bundle \( \pi: E \to M \). Unlike the classical Big Bang singularity, this is not a divergence in curvature but a discontinuity in coherence, initiating modular causal structure through collapse.

Subsequent measurement events build a causal sequence by projecting \( \tau \) along observer trajectories, generating an effective metric through the MTG construction:
\begin{equation}
g^{\mathrm{eff}}_{\mu\nu}(x) = \eta_{\mu\nu} + \kappa \int_{\Sigma_x} \mu[\tau(y)]_\mu \mu[\tau(y)]_\nu \, \rho(y) \, \mathrm{d}\Sigma(y),
\end{equation}
where \( \rho(y) \) is the local projection density, and \( \mu[\tau(y)] \) encodes its classical direction. The accumulation of such projection data defines time relationally, without reliance on an external or universal clock.

To describe cosmic expansion, we consider a foliation \( \{\Sigma_t\}_{t \in \mathbb{R}} \) of spacelike hypersurfaces aligned with the observer’s informational record, induced by modular time flow. At each point \( x = (t, \vec{x}) \in \Sigma_t \), we define the projection density as
\begin{equation}
\rho(t, \vec{x}) := \mathbb{E}_{\rho_{\mathrm{fiber}}(x)} \left[ \mu[\tau(x)] \neq 0 \right],
\end{equation}
where the expectation is taken with respect to the reduced density matrix \( \rho_{\mathrm{fiber}}(x) \), obtained by tracing out all but the internal time degrees of freedom at \( x \) from the global quantum state \( \rho_{\mathrm{global}}[\tau] \). This defines the probability that a coherence collapse occurs at point \( x \), conditional on the observer’s available information.

The cumulative projection count across the hypersurface \( \Sigma_t \) is then given by
\begin{equation}
N(t) := \int_{\Sigma_t} \rho(t, \vec{x}) \, \mathrm{d}^3 x,
\end{equation}
which quantifies the total number of classicalization events at modular time \( t \). This function plays a central role in defining emergent cosmological dynamics in the MTG framework.

The parameter \( t \) here refers to an emergent foliation parameter aligned with the modular time flow induced by the projection congruence \( n^\mu \). Each hypersurface \( \Sigma_t \) is defined as a slice of constant internal time, and the function \( \rho(t, \vec{x}) \) measures the local rate of decoherence. Its integral over \( \Sigma_t \) defines a coherence-weighted volume functional, which serves as a proxy for classical structure formation.

We then define the emergent cosmological scale factor by
\begin{equation}\label{eq:a-meas-event}
a^3(t) := \alpha \, N(t),
\end{equation}
where \( \alpha \) is a normalization constant that sets the effective unit of decohered volume. In this formulation, the expansion of space is governed by the accumulation of classicalization events. This generalizes the standard scale factor, which is typically defined by the spatial volume element \( \sqrt{\det g_{ij}} \) in a fixed metric background. In the MTG framework, by contrast, \( a^3(t) \) captures the statistical growth of decoherent structure and need not coincide with any geometric determinant unless \( g_{\mu\nu}^{\mathrm{eff}} \) is dynamically induced.

The scale factor thus emerges not from classical geometry, but from the temporal and informational structure imposed by projection. In regions of high projection density, \( a(t) \) increases rapidly; in coherence-preserving regimes, it remains approximately constant. This replaces the Friedmann equation with a statistical growth law rooted in measurement history.

This interpretation aligns with the modular geometry introduced earlier, where projection-induced decoherence defines an effective classical boundary. In this framework, even gravitational entropy can be understood informationally: the entropy associated with a spatial region \( A \) may be expressed by normalizing the projection density over the entangling surface \( \gamma_A \). Defining
\[
\tilde{\rho}(x) := \frac{\rho(x)}{\int_{\gamma_A} \rho(x')\, d^{d-1}x'},
\]
we obtain the surface entropy functional
\begin{equation}
S_A := -\int_{\gamma_A} \tilde{\rho}(x) \log \tilde{\rho}(x) \, \mathrm{d}^{d-1}x,
\end{equation}
which defines the pointwise Shannon entropy associated with the normalized projection density over the entangling surface \( \gamma_A \). While this does not represent the full von Neumann entropy of the reduced state \( \rho_A \), it agrees with it when the state is diagonalized in the measurement basis defined by projection. In this limit, the Ryu--Takayanagi area law emerges as a coarse-grained approximation to the informational structure of collapse, with surface entropy governed by the local decoherence profile of internal time.

The arrow of time is thereby reconstructed from the monotonicity of projection count. Temporal direction is not imposed but inferred:
\begin{equation}
t_2 > t_1 \quad \Leftrightarrow \quad N(t_2) > N(t_1).
\end{equation}
This formulation eliminates the need for globally synchronized clocks and derives temporal asymmetry directly from the irreversibility of measurement.

In classical cosmology, the scale factor satisfies the Friedmann equation:
\[
\left( \frac{\dot{a}}{a} \right)^2 = \frac{8\pi G}{3} \rho_{\text{matter}} - \frac{k}{a^2} + \frac{\Lambda}{3}.
\]
This equation governs the evolution of spatial geometry via matter and energy content in general relativity. In the MTG framework, by contrast, the evolution of \( a(t) \) is not determined by a differential equation derived from Einstein's equations, but by a statistical accumulation of measurement events;~\eqref{eq:a-meas-event}.

This coherence-weighted integral replaces the role of the Friedmann eqution with a relational growth law defined by the informational structure of projection. Rather than being sourced by energy density, curvature and expansion are inferred from the spatial distribution of decoherence. In this sense, MTG provides an informational reconstruction of cosmological expansion; grounded not in gravitational dynamics, but in observer-relative measurement statistics.

\begin{theorem}[Monotonicity of Scale Factor under Measurement Entropy Growth]
Let \( \{\Sigma_t\}_{t \in \mathbb{R}} \) be a foliation of a globally hyperbolic spacetime \( M \) indexed by an observer’s measurement history. Suppose \( \rho(t, \vec{x}) \) is the projection density and
\begin{equation*}
a^3(t) := \int_{\Sigma_t} \rho(t, \vec{x}) \, \mathrm{d}^3 x
\end{equation*}
is the emergent scale factor. If the measurement entropy \( \mathcal{S}(t, \vec{x}) \) satisfies
\[
\frac{\partial \mathcal{S}}{\partial t}(t, \vec{x}) \geq 0,
\]
then \( a(t) \) is non-decreasing in \( t \), and strictly increasing wherever \( \mathcal{S}(t, \vec{x}) \) increases on a set of nonzero measure.
\end{theorem}

\begin{proof}
If \( \mathcal{S}(t, \vec{x}) \) increases with \( \rho(t, \vec{x}) \), then
\[
\int_{\Sigma_{t_2}} \rho(t_2, \vec{x}) \, \mathrm{d}^3 x \geq \int_{\Sigma_{t_1}} \rho(t_1, \vec{x}) \, \mathrm{d}^3 x
\]
for \( t_2 > t_1 \), with strict inequality on any region of increasing entropy. Thus \( a(t) \) is monotonic.
\end{proof}

This result links cosmic expansion directly to measurement entropy. The scale factor \( a(t) \) reflects not dynamical evolution under a classical field, but the accumulation of information via projection. Time begins not at a geometric singularity, but at the first informational break in coherence. Expansion, temporal orientation, and the arrow of time all emerge together from a quantum-informational foundation.
%%%%%%%%%%%%%%%%%%%%
\subsection{Inflation as a Phase of Saturated Projection}

In the MTG framework, cosmic inflation arises not from a fundamental scalar inflaton field, but from a transient epoch characterized by saturated, spatially coherent projection of the internal time field \( \tau(x) \). During this regime, the time-fiber bundle exhibits uniformity across horizon-scale patches, and the projection map \( \mu[\tau(x)] \) collapses consistently to nearly identical classical directions. This coherence amplifies the deformation of the effective metric, driving exponential expansion through the accumulation of classicalization events.

Assume that over a relational time interval \( t \in [t_i, t_f] \), the internal time field is nearly spatially constant on each observer-aligned hypersurface \( \Sigma_t \), such that \( \partial_i \tau(x) \approx 0 \) for all \( x \in \Sigma_t \). Under this condition, the projection vectors \( \mu[\tau(x)]_\mu \) are approximately constant across \( \Sigma_t \), and the effective metric reduces to
\begin{equation}
g^{\mathrm{eff}}_{\mu\nu}(x) = \eta_{\mu\nu} + \kappa \int_{\Sigma_x} \mu[\tau(y)]_\mu \mu[\tau(y)]_\nu \, \rho(y) \, \mathrm{d}\Sigma(y),
\end{equation}
where the integrand contributions reinforce one another coherently. This alignment causes the classical deformation of geometry to accumulate constructively across space, amplifying the emergent curvature.

During this coherence-dominated phase, the projection density approaches saturation, \( \rho(x) \approx \rho_0 \), and becomes spatially uniform over each slice. The total number of classicalization events,
\begin{equation}
N(t) = \int_{\Sigma_t} \rho(x) \, \mathrm{d}^3x,
\end{equation}
increases steadily, leading to exponential growth in the emergent scale factor defined by \( a^3(t) \sim N(t) \). The resulting expansion rate is governed by
\begin{equation}
a(t) \sim \exp(H t), \quad H := \frac{1}{3} \frac{d}{dt} \log N(t),
\end{equation}
where \( H \) reflects the informational growth rate of projection. No vacuum energy or fine-tuned potential is required; inflation emerges from the geometric reinforcement of coherent collapse~\cite{guth1981inflationary}.

This inflationary epoch ends when the coherence of \( \tau(x) \) begins to fragment. Increasing entanglement with matter fields, topological obstructions in the fiber structure, or stochastic variation in projection processes cause \( \partial_i \tau(x) \) to grow and \( \mu[\tau(x)] \) to vary across \( \Sigma_t \). As a result, metric contributions from different regions begin to interfere destructively, halting exponential growth. Decoherence sets in, localizing the internal time field into spatially resolved configurations. The projection density becomes irregular, classical fields emerge, and entropy production begins—effectively reheating the universe without invoking a distinct reheating mechanism.

Unlike the standard picture where inflation is driven by a scalar field \( \phi \) rolling along a nearly flat potential \( V(\phi) \), the MTG framework realizes inflation as a coherence phase of the internal time bundle. The inflationary epoch corresponds to a period of maximal projection alignment; its termination arises from the breakdown of this alignment. The number of e-folds is determined not by the shape of a potential but by the duration of saturated, coherent measurement. This mechanism reframes inflation as an informational phase transition—one in which geometry, entropy, and causal order emerge from the structure of internal time collapse.

We formalize this behavior in the following theorem:

\begin{theorem}[Exponential Growth Under Saturated Coherent Projection]
Let $\tau(x)$ be spatially coherent and let $\rho(x)$ be constant over a relational time interval $t \in [t_i, t_f]$. Then the emergent scale factor
\[
a^3(t) := \int_{\Sigma_t} \rho(t, \vec{x}) \, \mathrm{d}^3 x
\]
satisfies
\[
a(t) = a(t_i) \, \exp\left( H (t - t_i) \right), \quad H = \frac{1}{3} \nabla_\mu n^\mu,
\]
where $n^\mu$ is the observer-aligned projection direction and $\nabla_\mu n^\mu$ is the expansion scalar of the congruence.
\end{theorem}

\begin{proof}
With $\rho(t, \vec{x}) = \rho_0$ constant and $\tau$ spatially coherent, each hypersurface $\Sigma_t$ evolves under the volume flow defined by the congruence generated by $n^\mu$. The infinitesimal volume element evolves according to
\[
\mathrm{d}^3 x(t) = \mathrm{d}^3 x(t_i) \, \exp\left( \int_{t_i}^t \nabla_\mu n^\mu \, \mathrm{d}t' \right),
\]
where $\nabla_\mu n^\mu$ denotes the local expansion rate of the volume form. Integrating this across $\Sigma_t$ yields
\[
a^3(t) = \rho_0 \int_{\Sigma_t} \mathrm{d}^3 x(t) = a^3(t_i) \, \exp\left( \int_{t_i}^t \nabla_\mu n^\mu \, \mathrm{d}t' \right),
\]
and taking the cube root gives the desired exponential form.
\end{proof}

This result captures the core insight of MTG inflation: the exponential expansion of the universe arises not from a fundamental dynamical field but from the geometric accumulation of coherent, measurement-induced projection events. The cosmological scale factor tracks the informational structure of projection, and the dynamics of inflation are governed by coherence, not classical potential energy. Reheating is reinterpreted as the fragmentation of this coherence, initiating classical matter dynamics and entropy production.

In this view, the inflationary epoch of the early universe is not a singular dynamical anomaly but the natural consequence of a transient regime of maximal alignment in the geometry of internal time. Measurement, coherence, and modular flow replace scalar fields and fine-tuned potentials, yielding a conceptually simpler and operationally grounded account of cosmic inflation.

\begin{figure}[ht]
\centering
\begin{tikzpicture}[scale=1.1]

% Hypersurfaces
\foreach \y in {0, 1.0, 2.0, 3.2} {
    \draw[gray!60] (-3,\y) -- (3,\y);
    \node[left] at (-3,\y) {\scriptsize $\Sigma_t$};
}

% Tau vectors aligned on early hypersurfaces
\foreach \x in {-2.5,-1.5,-0.5,0.5,1.5,2.5} {
    \draw[->, thick, blue!80] (\x,0) -- ++(0,0.9);
    \draw[->, thick, blue!80] (\x,1.0) -- ++(0,0.9);
}

% Tau vectors decohering on top hypersurface
\foreach \x/\dx/\dy in {-2.5/-0.2/0.9, -1.5/0.3/0.8, -0.5/-0.1/1.0, 0.5/0.2/0.9, 1.5/-0.3/1.1, 2.5/0.1/1.0} {
    \draw[->, thick, blue!40!gray, decorate, decoration={snake, amplitude=0.6pt, segment length=6pt}]
        (\x,3.2) -- ++(\dx,\dy);
}

% Labels for early coherence and decoherence
\node[blue!80] at (3.2,0.5) {\footnotesize Coherent $\tau(x)$};
\node[blue!40!gray] at (3.6,3.4) {\footnotesize Decohering $\tau(x)$};

% Arrows indicating expansion (scale factor)
\foreach \y in {0.5, 1.5, 2.6} {
    \draw[<->, red!70!black, thick] (-2.8,\y) -- (2.8,\y);
}

% Label for expansion
\node[red!70!black] at (0,2.8) {\footnotesize Expansion of $a(t)$};

% Optional inset a(t) plot
\begin{scope}[shift={(4.5,0)}]
\draw[->] (0,0) -- (2.2,0) node[right] {\scriptsize $t$};
\draw[->] (0,0) -- (0,1.5) node[above] {\scriptsize $a(t)$};
\draw[thick, violet] plot[smooth] coordinates {(0,0.1) (0.5,0.2) (1,0.4) (1.5,0.9) (2,1.2)};
\node at (1.2,1.4) {\tiny \textcolor{violet}{$a(t) \sim e^{Ht}$}};
\draw[dashed] (1.5,0.9) -- (1.5,0);
\end{scope}

\end{tikzpicture}
\caption{Measurement-induced inflation in MTG. Early alignment of the internal time field \( \tau(x) \) across spacelike hypersurfaces \( \Sigma_t \) leads to saturated projection density \( \rho(x) \) and exponential growth of the emergent scale factor \( a(t) \). Decoherence at later stages fragments coherence and initiates classical reheating.}
\label{fig:inflation_MTG_tikz}
\end{figure}
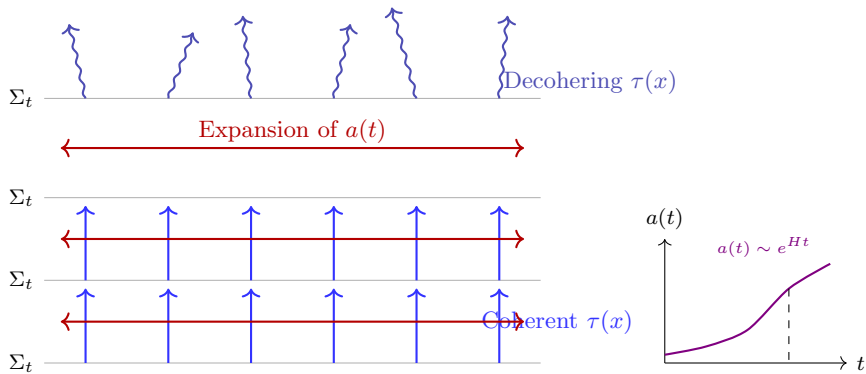

%%%%%%%%%%%%%%%%%%%%
\subsection{Structure Formation from Temporal Fluctuations}

In conventional cosmology, primordial density perturbations are typically traced to quantum fluctuations of a scalar inflaton field, magnified by exponential expansion and stretched to cosmological scales. In contrast, the MTG framework attributes the origin of structure to intrinsic fluctuations in the projection of the internal time field \( \tau(x) \), generated by the stochastic and decoherent nature of quantum measurement. These fluctuations do not arise from field-theoretic excitations, but from the probabilistic geometry of fiber collapse.

Even in a phase where \( \tau(x) \) is coherent and nearly homogeneous across spacelike hypersurfaces, the act of measurement introduces irreducible quantum uncertainty. Local deviations in the projection map \( \mu[\tau(x)]_\mu \) result from fiber misalignment, ambient decoherence, or curvature in the internal time bundle. These deviations perturb the emergent metric by contributing localized variations to the deformation term:
\begin{equation}
\delta g^{\mathrm{eff}}_{\mu\nu}(x) \sim \kappa \int_{\Sigma_x} \delta\left( \mu[\tau(y)]_\mu \mu[\tau(y)]_\nu \right) \rho(y) \, \mathrm{d}\Sigma(y).
\end{equation}
Here, fluctuations in the metric arise not from classical matter inhomogeneities but from variation in the causal imprint of projection across the internal fiber geometry.

The statistical structure of these fluctuations can be analyzed through two-point correlators of the projected time field. These are computed via the MTG path integral, summing over all admissible configurations of \( \tau \) and \( \rho \), subject to the projection-induced constraints:
\begin{equation}
\langle \delta \tau(\vec{x}) \, \delta \tau(\vec{x}') \rangle = \int \mathcal{D}\tau \, \mathcal{D}\rho \, e^{i S[\tau, \rho]} \, \delta \tau(\vec{x}) \, \delta \tau(\vec{x}'),
\end{equation}
where \( \delta \tau(x) \) denotes deviation from the coherent background. In Fourier space, the power spectrum is defined by
\begin{equation}
\langle \delta \tau_{\vec{k}} \, \delta \tau_{-\vec{k}} \rangle = (2\pi)^3 \delta^{(3)}(0) \, P(k),
\end{equation}
and under broad assumptions of local coherence and projection-limited noise, this spectrum is approximately scale-invariant:
\begin{equation}
P(k) \sim k^{n_s - 1}, \quad \text{with } n_s \approx 1.
\end{equation}
The resulting spectrum aligns with observations of the cosmic microwave background, despite its non-field-theoretic origin.

%\begin{figure}[!ht]
    %\centering
    %\includegraphics[width=0.95\textwidth]{fig_mtg_inflation.png}
    %\caption{Measurement-induced inflation in MTG. A coherent internal time field \( \tau(x) \) aligned across spacelike hypersurfaces leads to saturated projection density \( \rho(x) \) and exponential expansion of the emergent scale factor \( a(t) \). As coherence fragments, reheating occurs via decoherence and classical matter formation.}
    %\label{fig:inflation_MTG}
%\end{figure}

Crucially, the metric perturbations produced by this mechanism are intrinsically gauge-invariant. They originate from the structure of internal time, rather than from coordinate-dependent scalar fields. Scalar modes correspond to spatial variation in projection direction, tensor modes to shear across fiber sectors, and possible isocurvature modes may arise from disconnected topological domains.

In this view, the large-scale structure of the universe reflects the quantum geometry of temporal collapse. The seeds of galaxies, clusters, and anisotropies emerge from the stochastic fragmentation of coherence during early projection epochs. Their power spectrum, propagation, and observable imprint are all shaped by the informational content of measurement. MTG thus provides a physically grounded alternative to scalar inflation, rooting cosmological structure in the dynamics of time itself.

%%%%%%%%%%%%%%%%%%%%%
\subsection{Topological Projection Sectors}

The MTG framework offers a unified geometric account of both dark energy and dark matter, interpreting them as emergent phenomena arising from the topology and coherence structure of the internal time-fiber bundle. These effects are not attributed to unknown particles or exotic matter components, but instead reflect the incomplete collapse and topological complexity of the projection process. In this interpretation, deviations from standard gravitational dynamics signal the persistent presence of quantum structure within the temporal geometry of spacetime.

Dark energy originates in regions where the internal time field \( \tau(x) \) remains uncollapsed, preserving its coherent quantum character. In such domains, the projection density \( \rho(x) \) is effectively vanishing, and the projection map \( \mu[\tau(x)] \) is undefined or delocalized. The time-fiber curvature \( F_{\mu\nu} \), no longer constrained by classical measurement, retains its full quantum amplitude and contributes to the emergent metric via entropic backreaction. This residual curvature acts as an effective cosmological term, generating large-scale acceleration without requiring a fundamental cosmological constant. The resulting contribution to the emergent Einstein equations is governed by the expectation value of the curvature over coherence-preserved regions:
\begin{equation}
\Lambda_{\mathrm{eff}} \sim \left\langle F_{\mu\nu} F^{\mu\nu} \right\rangle_{\text{non-collapsed}}.
\end{equation}
This term captures the influence of internal quantum geometry on spacetime expansion and provides a natural origin for late-time cosmic acceleration.

Dark matter, in turn, emerges from topological obstructions within the internal time bundle. Global features such as non-contractible loops, fiber misalignments, or discrete holonomy sectors prevent \( \tau(x) \) from being globally trivialized, even in regions where projection has occurred. These topological defects carry no conventional energy–momentum, yet they modify the parallel transport of the internal time field and contribute to the deformation of the emergent metric \( g^{\mathrm{eff}}_{\mu\nu} \). Their influence is gravitational in character, producing effects analogous to non-luminous mass in galactic halos, filaments, and voids—despite being invisible to detectors limited to Standard Model interactions.

Together, these mechanisms reveal that dark energy and dark matter are complementary expressions of the same measurement-induced ontology. Persistent coherence manifests as a source of cosmic acceleration, while topological misalignment manifests as effective mass. Both arise from the structure of the projection process and the geometry of internal time, without invoking new particles, fields, or hidden sectors.

The next result formalizes the connection between residual fiber curvature and the emergence of a cosmological constant in the MTG gravitational field equations.

\begin{theorem}[Effective Cosmological Constant from Residual Fiber Curvature]
Let $M$ be a globally hyperbolic spacetime manifold equipped with an internal time-fiber bundle $\pi: E \to M$, and let $A^{\mathrm{time}}_\mu$ be a connection on $E$ whose curvature is defined by $F_{\mu\nu} := \partial_\mu A_\nu - \partial_\nu A_\mu + [A_\mu, A_\nu]$. Suppose $\rho(x)$ is a projection density function and let $\Omega \subset M$ be a region where $\rho(x) \approx 0$, so that projection is sparse or absent. Then the effective cosmological constant in the emergent Einstein equations over $\Omega$ is given by
\begin{equation}
\Lambda_{\mathrm{eff}} = \kappa \, \langle F_{\mu\nu} F^{\mu\nu} \rangle_{\Omega},
\end{equation}
where $\langle \cdot \rangle_{\Omega}$ denotes a coarse-grained average over $\Omega$ and $\kappa > 0$ is the MTG coupling parameter converting entanglement curvature into metrical response. If the curvature $F_{\mu\nu}$ is covariantly constant or slowly varying throughout $\Omega$, then the emergent Einstein equation in that region takes the standard form
\begin{equation}
G_{\mu\nu}[g^{\mathrm{eff}}] + \Lambda_{\mathrm{eff}} \, g^{\mathrm{eff}}_{\mu\nu} = 8\pi G \, T^{\mathrm{meas}}_{\mu\nu}.
\end{equation}
\end{theorem}

\begin{proof}
The gravitational sector of MTG is governed by a variational principle applied to the entropy functional \( S[\rho, \tau] \), which encodes the informational content of the internal time field and its projection structure. In regions where the projection density \( \rho(x) \) vanishes, the internal time field \( \tau(x) \) remains in a superposed quantum state, uncollapsed by measurement. The curvature \( F_{\mu\nu} \) of the temporal connection \( A^{\mathrm{time}}_\mu \) retains its quantum fluctuations, and its contribution to the entropy becomes dynamically significant.

Varying the total action with respect to the emergent metric \( g^{\mathrm{eff}}_{\mu\nu} \), one finds that the stress-energy contribution from uncollapsed regions is governed by the variation of the Yang–Mills-type curvature term:
\begin{equation}
\int_{\Omega} \mathrm{Tr}(F_{\alpha\beta} F^{\alpha\beta}) \, \sqrt{-g} \, \mathrm{d}^4x,
\end{equation}
where \( \Omega \subset M \) denotes a coherence-preserved domain. The functional derivative \( \delta S / \delta g^{\mu\nu} \) introduces a tensorial source proportional to the effective cosmological term \( \Lambda_{\mathrm{eff}} g_{\mu\nu} \), with
\begin{equation}
\Lambda_{\mathrm{eff}} := \kappa \left\langle F_{\mu\nu} F^{\mu\nu} \right\rangle_{\Omega}.
\end{equation}
If \( F_{\mu\nu} \) is approximately covariantly constant within \( \Omega \), then this contribution behaves identically to a cosmological constant in the Einstein equations, completing the proof.
\end{proof}

This result shows that cosmic acceleration arises naturally from persistent temporal coherence. The cosmological constant is no longer a fundamental parameter requiring fine-tuning, but a statistical consequence of curvature in regions uncollapsed by measurement. Likewise, topological sectors of the internal time bundle produce effective gravitational influence through their modification of the projection geometry and holonomy structure. These contributions explain dark energy and dark matter as geometric expressions of the same quantum-informational framework that gives rise to spacetime, causality, and temporal flow in the MTG model.
%%%%%%%%%%%%%%%%%%
\subsection{Cosmology as Measurement Geometry}

The MTG framework reinterprets cosmology not as the dynamical evolution of fields on a fixed spacetime background, but as the informational emergence of classical geometry from quantum time through measurement. In this view, the entire observable structure of the universe---including the arrow of time, inflation, large-scale structure, and the dark sector---arises from the coherence, topology, and projection dynamics of the internal time-fiber bundle \( \pi: E \to M \).

Cosmological time is not predefined as a global coordinate but arises relationally from the cumulative act of measurement. Each projection of the internal time field \( \tau(x) \) localizes a previously coherent temporal direction, and an observer constructs a relational temporal parameter \( t \) by slicing spacetime into hypersurfaces \( \Sigma_t \), where \( t_2 > t_1 \) if and only if
\[
\int_{\Sigma_{t_2}} \rho(x) \, \mathrm{d}^3x > \int_{\Sigma_{t_1}} \rho(x) \, \mathrm{d}^3x.
\]
This statistical ordering of projection events defines a monotonic arrow of time without invoking external clocks or global initial conditions.

In the early universe, inflation corresponds to a phase of saturated and spatially coherent projection. During this regime, the internal time field \( \tau(x) \) becomes nearly uniform across horizon-sized patches, and the projection density \( \rho(x) \) reaches a high, nearly constant value. The resulting exponential growth of the emergent scale factor \( a(t) \) is not driven by a potential energy term, but by the rapid accumulation of classicalization events. As coherence fragments, the time-fiber bundle decoheres, spatial variations in \( \tau \) emerge, and energy is redistributed into distinguishable matter fields, initiating reheating. This process requires no inflaton field or decay mechanism, but results purely from the fragmentation of modular alignment.

Structure formation arises from quantum fluctuations in the projection process itself. Even when \( \tau(x) \) is globally coherent, local misalignments, decoherence events, and curvature fluctuations introduce variation in the projected vectors \( \mu[\tau(x)]_\mu \). These variations imprint spatial inhomogeneities on the emergent metric \( g^{\mathrm{eff}}_{\mu\nu} \), producing both scalar and tensor perturbations. The resulting correlation structure is inherently gauge-invariant and naturally yields a nearly scale-invariant power spectrum, matching observed features of the CMB and large-scale structure without appealing to quantum field fluctuations of a hypothetical inflaton.

The dark sector is likewise reinterpreted in geometric and informational terms. Regions where projection has not occurred---where \( \rho(x) \approx 0 \) and \( \tau(x) \) remains coherent---retain nonzero curvature \( F_{\mu\nu} \), producing a nonvanishing contribution to the emergent gravitational equations. This contribution mimics a cosmological constant, yet it arises not from vacuum energy, but from persistent coherence in the internal time bundle. Similarly, topological defects in the fiber structure, including non-contractible loops, monodromies, and alignment vortices, source \( g^{\mathrm{eff}}_{\mu\nu} \) without coupling to Standard Model fields. These topological projection sectors behave phenomenologically as dark matter, influencing gravitational potentials without requiring additional particles.

What unifies these phenomena is coherence. Expansion reflects the accumulation of classical information; structure arises from coherence fluctuations; entropy and matter emerge from decoherence; and gravitational anomalies reflect residual curvature and topological frustration in internal time. The entire cosmological history becomes an emergent signature of how quantum time transitions to classical geometry through the observer-relative act of projection.

In this formulation, the universe is not a dynamical evolution through spacetime, but an informational geometry constructed from a record of measurements. Spacetime grows because coherence is collapsed. Matter exists because projection localizes. The dark sector manifests as the residue of unresolved quantum geometry. MTG thus reframes cosmology as a subfield of quantum information geometry: a theory in which spacetime, structure, and dynamics are not fundamental, but emergent features of measurement within a temporally fibered quantum universe.

\section{Modular Information Geometry}\label{sec:8}
The MTG framework provides a geometric and informational reformulation of the holographic principle. In traditional holographic dualities such as AdS/CFT, a gravitational bulk spacetime is encoded in the quantum dynamics of a lower-dimensional boundary theory. Within MTG, this correspondence arises not from a duality between distinct theories, but from a unified structure in which both bulk geometry and boundary entanglement emerge from the projection dynamics of an internal time-fiber bundle \( \pi: E \to M \). The act of measurement; specifically, the collapse of the internal time field \( \tau(x) \) into classical outcomes \( \mu[\tau(x)] \), generates an effective bulk metric \( g^{\mathrm{eff}}_{\mu\nu}(x) \), while simultaneously imprinting entanglement structure on the boundary.

In this formulation, each projection event contributes not only to the local classicalization of temporal geometry, but also to the organization of modular flow across observer-accessible regions. The modular Hamiltonian that governs boundary dynamics emerges as the generator of parallel transport along the internal fibers of \( \tau \), linking measurement outcomes across causal surfaces. Entanglement entropy, in turn, becomes a coarse-grained measure of projection density and coherence collapse over boundary-aligned wedges. The effective geometry of spacetime, the flow of modular information, and the causal structure of boundary theory all trace their origin to the same measurement-induced processes that define internal time.

This section develops the MTG interpretation of holography as a boundary manifestation of relational measurement geometry. We demonstrate how internal time-fiber holonomy defines modular structure on the boundary, how entanglement entropy emerges from projected coherence across minimal surfaces, and how bulk geometry is reconstructed through alignment and transport of modular data. In this view, holography is not a correspondence between independent theories, but a manifestation of a deeper informational unity: the geometry of spacetime and the entanglement structure of the boundary are both encoded in the projection history of quantum time.

%%%%%%%%%%%%%%%%%%%%%%%%%%%%%%%%%%%%%
\subsection{Modular Flow from Time-Fiber Projection}

Let \( M \) be a \((d+1)\)-dimensional Lorentzian spacetime with conformal boundary \( \partial M \), as arises in asymptotically AdS geometries. In the MTG framework, the internal geometry of time is encoded in a fiber bundle \( \pi: E \to M \), where each fiber \( \mathcal{T} \) represents a space of internal temporal configurations—typically a complex torus, a homogeneous domain such as \( SU(1,1)/U(1) \), or a quaternionic disk. Measurement-induced projection of the internal time field \( \tau(x) \) determines both bulk geometry and its boundary signature. Restricting the time-fiber bundle to the boundary yields a bundle \( E|_{\partial M} \to \partial M \), where the collapsed field \( \tau|_{\partial M} \) encodes observable modular structure accessible to boundary observers.

Boundary observers interact not with the full bulk configuration but with classical data imprinted by collapse events. For a spatial subregion \( A \subset \partial M \), the reduced density matrix \( \rho_A \) reflects access to a partial subsystem and generates a modular Hamiltonian \( K_A \), defined by \( \rho_A = Z^{-1} e^{-K_A} \), governing the evolution of operators \( \mathcal{O}_A(s) \) via modular flow: \( \frac{d}{ds} \mathcal{O}_A(s) = i[K_A, \mathcal{O}_A(s)] \). In conventional quantum field theory, \( K_A \) is abstractly defined through algebraic structures. Within MTG, however, it acquires a concrete geometric interpretation: it is the boundary encoding of informational flux carried by the collapse of \( \tau \) through the time-fiber bundle.

Let \( \mu[\tau(x)] \in T^*_x M \) denote the covector field obtained from the classicalization of the internal time field \( \tau \) at a boundary point \( x \in A \), and let \( \rho(x) \) represent the local density of measurement events. The product \( \mu[\tau(x)]^\mu \rho(x) \) defines a projection current that encodes the informational direction and strength of collapse. The modular Hamiltonian \( K_A \) associated with the boundary region \( A \subset M \) can then be expressed as the flux of this current through \( A \):
\begin{equation}
K_A \sim \int_A \mu[\tau(x)]^\mu \, \rho(x) \, n_\mu(x) \, \mathrm{d}^d x,
\end{equation}
where \( n_\mu(x) \) is the future-directed unit normal to \( A \).  If \( A = \partial \Omega \) bounds a spacetime region \( \Omega \subset \partial M \), and the fields \( \mu[\tau(x)] \) and \( \rho(x) \) are sufficiently smooth, the divergence theorem allows the modular Hamiltonian to be expressed as a bulk integral over \( \Omega \):
\begin{equation}
K_A \sim \int_\Omega \nabla_\mu J^\mu(x) \, \mathrm{d}^{d+1}x,
\end{equation}
where \( J^\mu(x) = \mu^\mu[\tau(x)] \rho(x) \) is the projection current defined previously. This volume formulation interprets modular energy as the integrated divergence of coherence flux, sourced by the local structure of measurement-induced collapse. It connects entanglement evolution on the boundary to informational flow in the bulk, governed by the geometry of the internal time-fiber bundle.

This construction reinterprets modular flow as an emergent, informationally grounded symmetry: it arises not from imposed dynamics, but from the causal structure of projection geometry. The boundary modular Hamiltonian becomes the observable trace of time-fiber collapse, and entanglement evolution across \( A \) is governed by how coherence in \( \tau \) fragments under observer-aligned measurement. The modular Hamiltonian is thus elevated from an abstract algebraic generator to a geometric encoding of quantum-to-classical transition in the internal time domain.

Through this correspondence, MTG explains modular flow, entanglement wedges, and holographic structure as phenomena rooted in the measurement geometry of internal time. The boundary inherits modular time as a relational construct, shaped by the projection history encoded in \( \mu[\tau(x)] \) and \( \rho(x) \). Modular Hamiltonians become flux functionals of the projection current, and their flow reflects the coherence gradients induced by measurement across \( \partial M \). Entanglement dynamics, rather than emerging from abstract algebraic automorphisms, are interpreted as the causal propagation of internal time collapse. In this setting, time is not an external parameter but the integrated direction of coherence loss relative to the observer's frame.
%%%%%%%%%%%%%%%%%%%%%%%%%%%%%%%%
\subsection{Entanglement Entropy and Minimal Surface Projections}
In conventional holography, the entanglement entropy of a boundary subregion \( A \subset \partial M \) is given by the Ryu--Takayanagi prescription, which assigns to \( A \) a codimension-2 surface \( \gamma_A \subset M \) that extremizes the area functional while being anchored to \( \partial A \) and homologous to it. The entropy is then given by \( S_A = \mathrm{Area}(\gamma_A)/(4 G_N) \), providing a direct correspondence between the geometry of spacetime and the entanglement structure of the boundary theory.

The MTG framework reinterprets this correspondence by replacing area minimization with projection flux extremization. Rather than selecting surfaces based on the classical spacetime metric, MTG identifies the relevant codimension-2 surface \( \widetilde{\gamma}_A \subset M \) as the one that extremizes the flux of the projection current \( J^\mu(x) = \mu^\mu[\tau(x)] \rho(x) \) through the surface. This flux encodes the informational flow of coherence collapse across \( \widetilde{\gamma}_A \), making entanglement entropy a functional of projection geometry rather than purely metric area.

To each boundary subregion \( A \), MTG associates a surface \( \widetilde{\gamma}_A \subset M \) that minimizes the total projection flux among all codimension-2 hypersurfaces homologous to \( A \). This flux is computed by contracting the projection current \( \mu^\mu[\tau(x)] \rho(x) \) with the surface volume form. 

In the special case where the projection vector field \( \mu^\mu[\tau(x)] \) is everywhere normal to \( \widetilde{\gamma}_A \), and the projection density \( \rho(x) \) is constant across the surface, the flux functional reduces to
\[
\int_{\widetilde{\gamma}_A} \rho \, \mu^\mu[\tau(x)] \, n_\mu(x) \, \mathrm{d}\Sigma(x) = \rho \int_{\widetilde{\gamma}_A} \mathrm{d}\Sigma(x),
\]
which is proportional to the area of \( \widetilde{\gamma}_A \) in the effective metric \( g^{\mathrm{eff}}_{\mu\nu} \). In this regime, the MTG prescription recovers the standard Ryu--Takayanagi formula, demonstrating that classical holographic entropy arises as a limiting case of coherent measurement geometry.

More generally, however, the geometry of \( \widetilde{\gamma}_A \) encodes an extremal condition not of spatial area, but of coherence-preserving modular flow. These surfaces identify loci in the bulk where modular alignment is least disrupted by projection, marking boundaries between informationally stable and decohering regions of the internal time bundle.

The entanglement entropy \( S_A \) is thus reinterpreted as a functional over the projection geometry of \( \widetilde{\gamma}_A \). Rather than being defined solely by geometric area, it becomes a local integral of entropy density determined by projection collapse:
\begin{equation}
S_A = \int_{\widetilde{\gamma}_A} \mathcal{S}(x) \, \mathrm{d}\Sigma(x), \quad \text{where} \quad \mathcal{S}(x) := -\operatorname{Tr}\left[ \rho_{\text{fiber}}(x) \log \rho_{\text{fiber}}(x) \right].
\end{equation}
Here \( \rho_{\text{fiber}}(x) \) is the reduced density matrix of the internal time fiber at point \( x \), obtained by tracing out all other degrees of freedom. The integrand \( \mathcal{S}(x) \) quantifies the local information loss due to measurement, and the surface integral reflects the total entropy associated with modular decoherence across \( \widetilde{\gamma}_A \).

In regimes where projection is dense and homogeneous, \( \mathcal{S}(x) \) becomes constant, and the entropy functional reduces to an area law. In more general settings, \( \mathcal{S}(x) \) varies with coherence and projection structure, encoding the informational gradient across the entanglement wedge.

This construction naturally links boundary entanglement to internal projection geometry. As the field \( \tau(x) \) undergoes localized collapses, its coherent structure fragments, inducing modular information flow across the bulk. The surface \( \widetilde{\gamma}_A \) then geometrizes this flow, marking the hypersurface along which the accumulation of classicalization contributes minimally to the entropy gradient. The resulting entropy is not merely a count of bulk degrees of freedom but a reflection of the history and geometry of collapse.

In this way, MTG generalizes the holographic entanglement principle. Entropy arises not from static surfaces in a fixed geometry, but from the dynamic propagation of modular coherence and its collapse-induced boundary imprint. Minimal surfaces become projection surfaces, and the RT formula appears as a limiting case in a broader correspondence between measurement, geometry, and information.
%%%%%%%%%%%%%%%%%%%%%%%%
\subsection{Bulk Reconstruction from Boundary Projection Data}
In conventional holographic dualities such as AdS/CFT, the gravitational bulk is encoded nonlocally in the degrees of freedom of a boundary conformal field theory. Within the MTG framework, this correspondence is reinterpreted in terms of projection geometry: the emergent bulk metric \( g^{\mathrm{eff}}_{\mu\nu} \), as well as the internal time field \( \tau(x) \), are reconstructed from a boundary record of measurement-induced collapse. Rather than arising from preassigned boundary conditions or abstract operator algebras, the MTG reconstruction scheme relies on observer-relative data—the measurement density \( \rho(x) \), the projected time vectors \( \mu[\tau(x)]^\mu \), and the holonomy of the internal time-fiber connection—defined directly on the conformal boundary \( \partial M \).

Let \( A \subset \partial M \) be a region of the boundary associated with an observer or an accessible subdomain. The projection data on \( A \) include the scalar field \( \rho(x) \) encoding the local rate of temporal collapse, the projection vectors \( \mu[\tau(x)]^\mu \) specifying the classicalized direction of internal time, and the connection one-form \( A_\mu^{\mathrm{time}} \) valued in the Lie algebra \( \mathfrak{g}_{\mathrm{time}} \). Together, these fields determine the causal, coherent, and modular alignment structure of the internal time bundle along the boundary.

Reconstructing the bulk geometry begins with analyzing how projected vectors are transported along \( \partial M \). The internal time connection defines a parallel transport operator whose holonomy captures how coherence evolves around boundary loops. For any curve \( \gamma_{xy} \subset \partial M \) connecting boundary points \( x \) and \( y \), the holonomy of the connection \( \nabla^{\mathrm{time}} \) is given by the path-ordered exponential
\begin{equation}
\operatorname{Hol}_{\gamma_{xy}}(\nabla^{\mathrm{time}}) = \mathcal{P} \exp\left( \int_{\gamma_{xy}} A^{\mathrm{time}}_\mu \, dx^\mu \right),
\end{equation}
encoding the misalignment and curvature of the internal time bundle. This boundary holonomy data can then be lifted into the bulk by extending both \( \rho(x) \) and \( \mu[\tau(x)]^\mu \) into a neighborhood of \( A \) using parallel transport consistent with the internal connection.

The emergent metric at a bulk point \( x \in M \) is defined by an integral over the surrounding causal region \( \Sigma_x \), incorporating the transported projection vectors and measurement field:
\begin{equation}
g^{\mathrm{eff}}_{\mu\nu}(x) = \eta_{\mu\nu} + \kappa \int_{\Sigma_x} \mu[\tau(y)]_\mu \mu[\tau(y)]_\nu \, \rho(y) \, \mathrm{d}\Sigma(y).
\end{equation}
This expression reflects how internal time collapses along causal hypersurfaces induce local geometric deformation, making spacetime structure a statistical consequence of modular alignment.

The reconstruction domain, the bulk region where this procedure is well-defined—is identified with the entanglement wedge \( \mathcal{E}_A \). In MTG, the entanglement wedge consists of all bulk points for which the projection vectors can be extended coherently from \( A \), without encountering destructive interference or topological obstructions in the time-fiber bundle. If \( x \in \mathcal{E}_A \), then \( \tau(x) \) and its associated curvature structure can be reconstructed from data on \( A \), and the emergent metric \( g^{\mathrm{eff}}_{\mu\nu}(x) \) is determined entirely by boundary measurement geometry.

Local observables in the bulk are reconstructed as functionals of boundary projection data transported through the internal time bundle. Let \( x \in \mathcal{E}_A \) be a point in the entanglement wedge associated with boundary region \( A \subset \partial M \). Define a lifted internal time configuration \( \tau_x \in \Gamma(E|_U) \) over a neighborhood \( U \ni x \), obtained by parallel transport of boundary data \( \tau|_{\partial M} \) under the time-fiber connection \( \nabla^{\mathrm{time}} \). Then the observable \( \mathcal{O}(x) \) is given by
\begin{equation}
\mathcal{O}(x) := \mathcal{F}_A[\tau_x, \rho, A^{\mathrm{time}}],
\end{equation}
where \( \mathcal{F}_A \) is a gauge-covariant functional encoding the dependence on projected internal time, measurement density, and fiber curvature structure. This construction ensures that bulk fields remain relationally defined by observer-accessible data and retain covariance under fiber-bundle automorphisms. The time evolution of these observables is not governed by a universal Hamiltonian but by the modular Hamiltonian \( K_A \), which generates relational modular flow via
\begin{equation}
\frac{d}{ds} \mathcal{O}(x; s) = i [K_A, \mathcal{O}(x; s)],
\end{equation}
with \( s \) the modular parameter induced by the projection current across \( A \). This relation confirms that bulk dynamics, including causality and evolution, are derived from the observer's informational geometry, not from external spacetime coordinates. The conditions under which this reconstruction is valid can be formalized as follows:

\begin{theorem}[Bulk Reconstruction from Projection Congruences]
Let \( M \) be a spacetime with boundary \( \partial M \), and let \( \pi: E \to M \) be a smooth internal time-fiber bundle equipped with a connection \( \nabla^{\mathrm{time}} \) and curvature \( F^{\mathrm{time}}_{\mu\nu} \). Let \( A \subset \partial M \) be a boundary region with projection data consisting of:
\begin{itemize}
  \item A smooth density field \( \rho: A \to \mathbb{R}_{\geq 0} \),
  \item A projection vector field \( \mu[\tau]: A \to T^* \partial M \),
  \item A connection one-form \( A^{\mathrm{time}}_\mu \in \mathfrak{g}_{\mathrm{time}} \).
\end{itemize}
Assume:
\begin{enumerate}
  \item The projection vectors \( \mu[\tau(x)] \) extend smoothly into a bulk region \( U \subset M \) via parallel transport under \( \nabla^{\mathrm{time}} \).
  \item The holonomy group of \( \nabla^{\mathrm{time}} \) over \( A \) acts trivially on fibers in \( U \).
  \item The field \( \rho(x) \) extends smoothly into \( U \).
\end{enumerate}
Then for any point \( x \in U \subset \mathcal{E}_A \), the effective metric \( g^{\mathrm{eff}}_{\mu\nu}(x) \) and internal time field \( \tau(x) \) are reconstructible via the integral:
\[
g^{\mathrm{eff}}_{\mu\nu}(x) = \eta_{\mu\nu} + \kappa \int_{\Sigma_x \cap U} \mu[\tau(y)]_\mu \mu[\tau(y)]_\nu \, \rho(y) \, \mathrm{d}\Sigma(y),
\]
where \( \Sigma_x \) is a causal hypersurface through \( x \) and \( \kappa \) is the coupling parameter of the projection geometry.
\end{theorem}

\begin{proof}
Let \( x \in U \subset M \), and let \( \gamma_{xy} \subset M \) be any smooth path from \( x \) to a point \( y \in A \subset \partial M \). Since the holonomy of \( \nabla^{\mathrm{time}} \) is trivial over \( U \), parallel transport along \( \gamma_{xy} \) defines a unique lift of the boundary projection vector \( \mu[\tau(y)] \) to a covector \( \mu[\tau(x)] \in T^*_x M \). The triviality of holonomy ensures that the result is independent of the path homotopy class within \( U \), and smoothness of the connection guarantees that the resulting field \( \mu[\tau(x)] \) is smooth in \( x \).

By assumption, the projection density \( \rho(x) \) also extends smoothly into \( U \), so the integrand in
\[
g^{\mathrm{eff}}_{\mu\nu}(x) = \eta_{\mu\nu} + \kappa \int_{\Sigma_x \cap U} \mu[\tau(y)]_\mu \mu[\tau(y)]_\nu \, \rho(y) \, \mathrm{d}\Sigma(y)
\]
is well-defined and continuous. The surface \( \Sigma_x \) may be chosen as a spacelike or null hypersurface anchored in the causal past of \( x \), and the domain of integration restricted to \( \Sigma_x \cap U \) ensures consistency with the boundary-propagated projection data.

Therefore, both \( g^{\mathrm{eff}}_{\mu\nu}(x) \) and \( \tau(x) \) are constructible from data on \( A \), and the reconstruction is smooth and gauge-consistent within the region \( U \subset \mathcal{E}_A \).
\end{proof}

\begin{figure}[ht]
\centering
\begin{tikzpicture}[scale=.8]

% Boundary circle
\draw[thick] (0,0) circle (3);
\node at (3.6,0) {\small $\partial M$ (Boundary)};
\node at (0,0) {\small Bulk $M$};

% Region A on boundary
\draw[very thick, red] (45:3) arc (45:135:3);
\node[red] at (90:3.4) {\small Region $A$};

% Draw entanglement wedge
\fill[blue!10, opacity=0.5]
    (45:3) .. controls (1.0,1.4) and (-1.0,1.4) .. (135:3)
    -- (135:3) arc (135:45:3);

\draw[dashed, blue!70!black]
    (45:3) .. controls (1.0,1.4) and (-1.0,1.4) .. (135:3);

%\node at (0,1.8) {\small Entanglement Wedge};

% Projection vectors mu[tau(x)] with curvature
\foreach \angle/\scale in {60/1.0,75/1.3,90/1.1,105/0.9,120/1.2}
{
    \draw[->, thick, blue!80!black, decorate, decoration={snake,amplitude=0.6pt, segment length=6pt}]
        (\angle:3) -- ++({\angle-180}:1.3*\scale)
        node[midway, above right=-2pt, rotate=\angle-180] {\tiny $\mu[\tau(x)]$};
}

% Modular flow arrows inside wedge
\foreach \x/\y in {0/1.1}
{
    \draw[->, thick, violet!70!black, dashed]
        (\x,\y) -- ++(0,-0.8)
        node[midway, right] {\tiny modular flow};
}

% Caption label
\node[align=center] at (0,-3.4) {
\footnotesize Bulk reconstruction: $K_A$ from boundary $\mu[\tau(x)]$ and $\rho(x)$
};

\end{tikzpicture}
\caption{Reconstruction of bulk modular geometry from boundary projections. Collapse vectors \( \mu[\tau(x)] \) along a boundary region define modular Hamiltonians that generate flow into the entanglement wedge. Fiber alignment and projection density govern bulk curvature.}
\label{fig:bulk-reconstruction}
\end{figure}

%%%%%%%%%%%%%%%%%%%%%%%%%%%%%%
\subsection{Holography as Measurement Geometry}
The MTG framework reframes holography as a consequence of measurement geometry rather than a correspondence between dual theories. Boundary entanglement, modular flow, and bulk spacetime geometry all emerge from the projection dynamics of the internal time-fiber bundle \( \pi: E \to M \). Modular Hamiltonians are reinterpreted as flux integrals of projection current; entanglement entropy arises from the integration of collapse-induced entropy density across projection-extremal surfaces; and bulk observables are reconstructed from parallel transport of internal time data through coherence-preserving domains~\cite{hamilton2006local}. The entanglement wedge \( \mathcal{E}_A \) becomes the maximal region where boundary measurement history determines classical geometry, with causal structure governed not by a background spacetime but by the alignment and coherence of measurement-induced projections. In this formulation, holography is not an abstract duality; it is a relational lifting of modular structure through the geometry of collapse, encoding time, entropy, and causality as emergent features of the observer's informational interaction with the quantum state.

Let \( \mu[\tau(x)]_\mu \in T^*_x M \) denote the covector field associated with the classicalization of the internal time field \( \tau \), and let \( \rho(x) \) represent the local density of measurement events. Using the effective metric to raise indices, define the coherence vector
\begin{equation}
\mu^\mu[\tau(x)] := g^{\mu\nu}_{\mathrm{eff}}(x) \mu_\nu[\tau(x)],
\end{equation}
and introduce the projection current
\begin{equation}
J^\mu(x) := \mu^\mu[\tau(x)] \, \rho(x),
\end{equation}
which encodes both the direction and intensity of classicalization. The modular Hamiltonian associated to a boundary subregion \( A \subset \partial M \) is then given by the flux of this current across \( A \):
\begin{equation}
K_A \sim \int_A J^\mu(x) \, n_\mu(x) \, \mathrm{d}^d x,
\end{equation}
where \( n_\mu(x) \) is the future-directed unit normal to \( A \). This formulation recasts modular flow not as an abstract automorphism of operator algebras, but as a directional transport of coherence generated by projection dynamics along the internal time bundle.

Entanglement entropy in MTG acquires both an informational and geometric origin. The Ryu--Takayanagi surface \( \gamma_A \subset M \), a codimension-2 hypersurface homologous to \( A \subset \partial M \), is reinterpreted as the surface that extremizes the modular projection current through the bulk. When projection vectors are aligned and \( \rho(x) \) is uniform, this reproduces the standard minimal-area prescription. More generally, however, entanglement entropy becomes a functional of the measurement record,
\begin{equation}
S_A \sim \int_{\gamma_A} \rho(x) \log \rho(x) \, \mathrm{d}^{d-1} x,
\end{equation}
capturing the entropy cost of classicalizing modular information across \( \gamma_A \).

This perspective naturally links boundary entanglement to internal projection geometry. As the internal time field \( \tau(x) \) undergoes localized collapses, its coherent structure fragments, generating modular information flow across the bulk. The surface \( \widetilde{\gamma}_A \) then geometrizes this flow, marking the hypersurface along which classicalization accumulates with minimal disruption to modular alignment. The resulting entropy is not merely a count of degrees of freedom, but a reflection of the informational history encoded in projection dynamics.

In this framework, the Ryu--Takayanagi formula emerges as a coarse-grained limit of a more fundamental measurement-theoretic structure. Minimal surfaces are not fundamental entities but statistical boundaries—interfaces where projection-induced decoherence balances internal modular coherence. Their existence reflects the interplay between internal time curvature, fiber holonomy, and measurement density—structures absent from classical gravity but essential to MTG’s holographic dictionary.

Bulk geometry is reconstructed from boundary projection data by extending projection vectors \( \mu[\tau(x)]^\mu \) into the bulk via parallel transport under the internal time connection \( \nabla^{\mathrm{time}} \). The emergent metric takes the form
\begin{equation}
g^{\mathrm{eff}}_{\mu\nu}(x) = \eta_{\mu\nu} + \kappa \int_{\Sigma_x} \mu[\tau(y)]_\mu \mu[\tau(y)]_\nu \, \rho(y) \, \mathrm{d}\Sigma(y),
\end{equation}
where \( \Sigma_x \) is a hypersurface anchored in the observer’s causal domain. The coherence and alignment of these transported projection vectors determine the geometry and curvature within the entanglement wedge \( \mathcal{E}_A \), which is defined as the bulk region where such a transport remains consistent and free of destructive decoherence.

Local bulk observables are then expressed as functionals of the lifted internal time field,
\begin{equation}
\mathcal{O}(x) \sim \mathcal{F}_A[\tau|_{\partial M}], \quad x \in \mathcal{E}_A,
\end{equation}
with modular time evolution governed by the boundary Hamiltonian \( K_A \):
\begin{equation}
\frac{d}{ds} \mathcal{O}(x; s) = i [K_A, \mathcal{O}(x; s)].
\end{equation}
This confirms that temporal dynamics in the bulk are not externally imposed but arise from informational flow defined by the structure of measurement on the boundary.

The AdS/CFT dictionary is thus reinterpreted as a lifting of measurement-induced projection geometry from the boundary to the bulk. Spacetime emerges not from fundamental fields, but as the classical imprint of internal time collapse aligned by modular geometry. In this sense, holography becomes a relational statement about the coherence of quantum measurement histories, encoded geometrically in the fiber bundle structure of time.

\section{String Theory Embedding}\label{sec:9}
This section has reframed the central tenets of holography---modular Hamiltonians, entanglement entropy, and bulk reconstruction---within the fiber-geometric ontology of Measurement-Induced Temporal Geometry (MTG). Rather than treating bulk spacetime and its boundary dual as abstractly related dynamical systems, MTG embeds both in a unified structure: the internal time-fiber bundle \( \pi: E \to M \), which encodes measurement-induced collapse events, coherence flow, and the emergence of causal geometry.

The boundary modular Hamiltonian \( K_A \) for a subregion \( A \subset \partial M \) is interpreted as a flux of coherence-collapse, sourced by the projection vector \( \mu[\tau(x)] \) and modulated by the spatial gradient of projection density \( \nabla_\mu \rho(x) \). This yields a flow operator of the form
\[
K_A \sim \int_A \mu[\tau(x)]^\mu \nabla_\mu \rho(x) \, \mathrm{d}^d x,
\]
where the integrand encodes how sharply and in what direction classicalization proceeds through the boundary under measurement. In this picture, modular flow is not an abstract automorphism of the algebra of observables, but the physical manifestation of time-fiber alignment under entanglement-preserving projections.

Entanglement entropy, traditionally associated with the area of a minimal surface in the bulk via the Ryu--Takayanagi formula, is likewise reinterpreted as a projection-induced quantity. The relevant codimension-2 hypersurface \( \gamma_A \) is no longer selected solely by metric extremization, but by the condition that it balances or extremizes modular projection flux through the internal time geometry. Entropy is then computed as a functional of local projection statistics, with
\[
S_A \sim \int_{\gamma_A} \rho(x) \log \rho(x) \, \mathrm{d}^{d-1} x,
\]
reflecting the degree of classicalization across \( \gamma_A \) rather than its geometric area per se. In this interpretation, holographic entropy arises not from intrinsic spatial structure, but from the depth of measurement collapse within the bulk.

The emergent bulk geometry itself is reconstructed from boundary projection data through the extension of fiber alignment and holonomy transport. Provided the time-fiber connection \( \nabla^{\mathrm{time}} \) is smooth and the holonomy group acts trivially within a bulk subregion \( U \subset \mathcal{E}_A \), both the effective metric \( g^{\mathrm{eff}}_{\mu\nu}(x) \) and the internal time field \( \tau(x) \) may be fully determined from the boundary data \( \{ \mu[\tau], \rho, A^{\mathrm{time}} \} \) restricted to \( A \subset \partial M \). This identifies the entanglement wedge as the maximal domain in which quantum coherence on the boundary is sufficient to determine classical geometry in the bulk. Temporal evolution of bulk observables is governed by the modular Hamiltonian as a geometric consequence of boundary projection flow, with
\[
\frac{d}{ds} \mathcal{O}(x; s) = i [K_A, \mathcal{O}(x; s)],
\]
showing that modular time is nothing other than coherence-directed collapse extended through the fiber congruence.

Altogether, MTG recasts the holographic dictionary as a lifting of informational collapse: from the density and orientation of projections on the boundary, one reconstructs the emergent geometry, curvature, and modular dynamics of the bulk. Entanglement entropy arises from projection density and coherence variation; modular flow is generated by measurement-aligned time vectors; and bulk curvature reflects the obstruction to globally coherent projection. Spacetime, in this framework, is not a backdrop but a relational signature of quantum measurement geometry, organized through the internal time field \( \tau \). The entire architecture of holographic duality emerges from the geometry of becoming classical. In the final section, we will examine how this paradigm influences our understanding of quantum gravity, and explore its implications for cosmology, black hole interiors, and the unification of geometry with quantum information.

%%%%%%%%%%%%%%%%%%%%%%
\subsection{Worldsheet Embedding of the Internal Time Bundle}

In the string-theoretic embedding of Measurement-Induced Temporal Geometry (MTG), the internal time field \( \tau(x) \) arises naturally from the compactification structure of higher-dimensional string backgrounds. Rather than being introduced as a primitive variable, \( \tau(x) \) is identified with moduli fields associated with the geometry and fluxes of internal cycles in a compactification manifold \( X \). In type II theories, \( X \) may be a Calabi--Yau threefold; in M-theory or heterotic strings, it may be a \( G_2 \)- or \( Spin(7) \)-holonomy manifold. The internal time fiber over spacetime \( M \) is realized as a moduli-valued section induced by dimensional reduction, where the value of \( \tau(x) \) at each point corresponds to the integral of higher-dimensional fields over selected internal p-cycles \( \Sigma \subset X \).

For instance, in type IIA or IIB compactifications, a natural candidate for \( \tau(x) \) is the complexified Kähler modulus associated with a compact two-cycle \( \Sigma \in H_2(X, \mathbb{Z}) \), defined by
\begin{equation}
\tau(x) \sim \int_{\Sigma} (B + i J),
\end{equation}
where \( B \) is the Neveu--Schwarz two-form and \( J \) is the Kähler form. In M-theory or heterotic settings, \( \tau \) may instead emerge from integrals of the C-field over three- or four-cycles, or from Wilson lines and complex structure moduli. The total internal time bundle \( \pi: E \to M \) thereby acquires a geometric embedding in the compactification manifold and reflects the structure of its moduli space.

This internal time modulus \( \tau(x) \) couples to the two-dimensional worldsheet conformal field theory (CFT) via marginal operators that deform the background geometry. The worldsheet \( \Sigma_2 \), embedded in the product \( M \times X \), supports a nonlinear sigma model whose interaction terms with the internal space take the form
\begin{equation}
S_{\mathrm{int}} = \frac{1}{4\pi\alpha'} \int_{\Sigma_2} (B_{IJ} + i J_{IJ}) \, \partial X^I \bar{\partial} X^J + \cdots,
\end{equation}
where \( X^I \) are coordinates on \( X \). These couplings represent exactly marginal deformations of the internal CFT and generate moduli flows that modify the worldsheet beta functions, thereby deforming the target-space background. In the MTG setting, such deformations translate into variations in the emergent spacetime geometry, which is governed by the measurement-induced metric
\begin{equation}
g^{\mathrm{eff}}_{\mu\nu}(x) = \eta_{\mu\nu} + \kappa \int_{\Sigma_x} \mu[\tau(y)]_\mu \mu[\tau(y)]_\nu \, \rho(y) \, \mathrm{d}\Sigma(y),
\end{equation}
where \( \mu[\tau] \) is the projection of the internal modulus into a classical configuration and \( \rho(y) \) is the local density of measurement events. Worldsheet couplings thus contribute dynamically to spacetime curvature through the projection geometry of internal time.

Within this framework, measurement-induced projection corresponds to the physical stabilization of moduli in string theory. Collapse of \( \tau(x) \) into a definite classical value reflects the localization of a quantum modulus into a stabilized vacuum configuration. Several mechanisms provide such stabilization. First, open string condensation, such as brane--antibrane annihilation or tachyon condensation, can dynamically fix moduli by eliminating off-shell fluctuations. Second, supersymmetry-breaking processes—including gaugino condensation and flux-induced superpotentials—lift flat directions and favor energetically isolated vacua, effectively selecting a projection of \( \tau \). Third, worldsheet instantons generate nonperturbative contributions that exponentially suppress large excursions in moduli space, thereby localizing the wavefunction of \( \tau \) and inducing decoherence-like behavior from string-theoretic dynamics alone.

In each of these mechanisms, the projection map \( \mu[\tau(x)] \) is realized through physical moduli stabilization. Once stabilized, \( \tau(x) \) defines the local direction of modular flow, encodes observer-relative collapse structure, and determines the entropic signature of the emergent metric. The effective geometry of spacetime inherits these properties, making classical causal structure a statistical imprint of string-theoretic coherence collapse.

This identification embeds MTG within the ultraviolet-complete framework of string theory. The internal time fiber is realized as a modulus field arising from compactification; its projection corresponds to dynamical stabilization; and its transport, curvature, and holonomy correspond to geometric features of the fibered background. As a result, the MTG emergent metric inherits consistency conditions such as modular invariance, anomaly cancellation, and duality covariance from the underlying worldsheet theory. Classical spacetime is thus interpreted as the informational resolution of internal string geometry through relational projection. Time itself becomes a section of a compactified fiber, measured and collapsed by observers, with its coherence and classicalization dictating both gravitational dynamics and causal order. This worldsheet realization grounds the MTG framework in string theory and connects the emergence of geometry to the informational flow of projection in a physically consistent, holographically enriched setting.
%%%%%%%%%%%%%%%%%%%%%%
\subsection{Brane Intersections and Projection Events}
In the MTG framework, measurement-induced projection---the collapse of the internal time field \( \tau(x) \) into a classical value---finds a physical realization in string theory through the dynamics of D-brane intersections. These events, which include the  formation of the theoretical tachyonic open string modes and their subsequent condensation, provide a concrete mechanism by which the coherent quantum geometry of \( \tau(x) \) becomes localized into observer-dependent classical structure~\cite{sen2002tachyon}. The interaction between branes thus replaces the abstract notion of quantum measurement with a geometric and dynamical transition.

Consider a configuration in which a D3-brane extends along the noncompact spacetime directions \( (x^0, x^1, x^2, x^3) \), while a D7-brane wraps both these directions and additional internal cycles within a Calabi--Yau compactification manifold \( X \). Their intersection defines a codimension-2 locus in spacetime where open strings can stretch between the branes. These strings support massless and tachyonic modes, the latter arising when the brane separation drops below a critical distance. The appearance of tachyonic excitations signals an instability that leads to condensation, effectively collapsing the quantum degrees of freedom associated with stretched strings. This process corresponds, in the MTG language, to the classicalization of the internal time field at the intersection point.

The tachyon field \( T(x) \) defined along the intersection mediates this collapse. As the tachyon condenses, quantum fluctuations in \( \tau(x) \) are exponentially suppressed, and the field localizes to a definite classical configuration. This behavior is captured by the expression
\begin{equation}
\mu[\tau(x)] := \lim_{T(x) \to \infty} e^{-T(x)} \tau(x),
\end{equation}
which formalizes the projection as a limit where coherence vanishes and the time fiber reduces to a classical section. From the perspective of the internal time bundle \( \pi: E \to M \), this transition replaces a locally fluctuating fiber with a sharp projection vector, aligned with the brane worldvolume and shaped by the supersymmetry and moduli of the intersecting branes.

Within this setup, measurement acquires a concrete physical interpretation. It is the localized interaction of branes that selects a preferred configuration of the internal time field. The observer’s frame is encoded in the geometry and alignment of the measurement brane, and the act of collapse corresponds to the annihilation of off-diagonal elements in the stretched string Hilbert space. The result is a realignment of the fiber geometry, interpreted as a topological transition in the space of internal time directions.

Following this projection, the field \( \tau(x) \) becomes confined to the effective theory supported on the resulting brane worldvolume. Letting \( \mathcal{W} \subset M \) denote this region, the classical internal time field becomes a well-defined section \( \tau_{\mathrm{cl}}(x) \in \Gamma(E|_{\mathcal{W}}) \), satisfying \( \mu[\tau(x)] = \tau_{\mathrm{cl}}(x) \). This localized field defines the modular flow and causal orientation accessible to the observer associated with \( \mathcal{W} \). Meanwhile, outside this region, \( \tau \) remains coherent, allowing for future measurement interactions and further collapses.

The effective metric \( g^{\mathrm{eff}}_{\mu\nu}(x) \) generated by this process reflects the geometric imprint of prior brane interactions. Each collapse modifies the local projection congruence, alters the curvature of the fiber geometry, and reconfigures the internal causal order. The cumulative effect of such events across the bulk spacetime determines the full structure of emergent gravity and modular alignment. Each brane intersection contributes not only a geometric deformation but an entropic transition: the reduction of internal coherence and the emergence of classical order correspond to a loss of von Neumann entropy and a gain in classical distinguishability.

These transitions are topologically nontrivial. They may induce defect formation, generate holonomy in the time-fiber bundle, or lead to the appearance of modular anomalies. The resulting brane worldvolume serves as a historical record of measurement, encoding collapse data in its residual fields, gauge content, and moduli space structure. In this way, string theory provides a UV-complete realization of projection events, in which quantum measurement becomes a derived feature of higher-dimensional geometry. MTG thus integrates seamlessly into the brane dynamics of string theory, transforming the classicalization of time into a consequence of moduli stabilization, tachyon condensation, and observer-brane coupling. This framework grounds the informational transition from quantum to classical in the geometry of intersecting branes and prepares the way for a fully covariant and predictive theory of Planck-scale quantum gravity.

%%%%%%%%%%%%%%%%%%%%%%
\subsection{Supersymmetric Completion and Spontaneous Breaking}

To ensure ultraviolet consistency and control over quantum corrections, the Measurement-Induced Temporal Geometry (MTG) framework admits a natural supersymmetric extension. In this formulation, the internal time-fiber bundle is uplifted to a super-fiber bundle over a superspace \( (M, \mathcal{F}) \), where \( M \) is the underlying spacetime manifold and \( \mathcal{F} \) is the sheaf of superfunctions incorporating both bosonic and fermionic coordinates. The internal time field \( \tau(x) \), formerly treated as a scalar section of the bundle \( \pi: E \to M \), is now realized as the lowest component of a chiral superfield \( \Phi(x, \theta) \), where
\[
\Phi(x, \theta) = \tau(x) + \sqrt{2} \, \theta \psi(x) + \theta^2 F(x).
\]
Here \( \psi(x) \) is the fermionic superpartner of \( \tau(x) \), and \( F(x) \) is a complex auxiliary field. The fiber connection \( A_\mu \) on \( E \) is likewise promoted to a superconnection \( \mathcal{A}_M = (A_\mu, \lambda, D) \), which may be embedded in a vector supermultiplet or realized via the supercovariant derivative algebra. Supersymmetric coupling to supergravity follows standard formulations, and the projection density \( \rho(x) \) may appear as a composite operator or auxiliary field within the supergravity multiplet, contributing to the emergent geometry in a manifestly supersymmetric way.

In the unmeasured phase, the MTG action preserves supersymmetry. The supersymmetry variations of the chiral superfield satisfy
\[
\delta \Phi = \sqrt{2} \, \epsilon \psi + \cdots, \quad \delta \psi = \sqrt{2} \, \epsilon F + i \sqrt{2} \, \sigma^\mu \bar{\epsilon} \partial_\mu \tau,
\]
and in the coherent regime where \( F(x) = 0 \), the symmetry remains unbroken. Non-renormalization theorems protect the internal time field from quantum corrections, allowing moduli to remain flat and coherence to be preserved. Supersymmetry thus stabilizes the structure of \( \tau(x) \) in the absence of projection.

However, when a measurement occurs—signaled by the collapse \( \mu[\tau(x_0)] \neq 0 \) at some spacetime point \( x_0 \)—the superfield acquires a localized classical configuration. The auxiliary field \( F(x) \) obtains a nonzero expectation value at the collapse point,
\begin{equation}
\langle F(x_0) \rangle \neq 0 \quad \Rightarrow \quad \delta \psi(x_0) = \sqrt{2} \, \langle F(x_0) \rangle \, \epsilon \neq 0,
\end{equation}
indicating spontaneous breaking of supersymmetry. The fermion \( \psi(x) \), now a massless Goldstino, reflects the decoherence of \( \tau(x) \) and marks the transition from a superfield to a classically observed configuration. This mechanism mirrors the appearance of goldstini in metastable vacua of string-derived effective field theories.

The consequences of this spontaneous breaking propagate through supergravity couplings. Projection-induced collapse sources curvature and torsion via supersymmetric higher-derivative invariants. For instance, in old minimal supergravity, one finds terms of the form
\[
\int d^4\theta \, \Phi^\dagger e^V \Phi \supset \rho(x) F_{\mu\nu} F^{\mu\nu} + \text{fermionic terms},
\]
where \( F_{\mu\nu} \) is the curvature of the internal time-fiber connection. These contributions enter the effective stress tensor and backreact on the emergent geometry. In flux compactifications of type IIB string theory, where \( \tau(x) \) may correspond to complex structure moduli, the projection \( \mu[\tau(x)] \) selects a flux vacuum via
\[
G_3 = F_3 - \tau H_3, \quad W_{\mathrm{flux}} = \int_X G_3 \wedge \Omega,
\]
and collapse becomes equivalent to vacuum selection in a non-BPS flux configuration. Supersymmetry is thus broken at the level of the internal time field, and its projection modifies the local geometry and matter structure accordingly.

This mechanism gives measurement a dynamical and geometric interpretation within supersymmetric field theory. Projection is not an extrinsic operation but a transition within the superfield landscape, selecting a superselection sector of the full theory and breaking symmetry locally. The resulting Goldstino \( \psi(x) \) encodes collapse signatures, while \( F \) and \( D \) fields act as classical sources for curvature and torsion. Supersymmetric MTG embeds the process of decoherence within the covariant structure of supergravity and string theory, offering a UV-complete and anomaly-consistent realization of temporal projection.

In this formulation, measurement, symmetry breaking, and spacetime emergence are all understood as facets of a unified geometry. Internal time is no longer a fixed parameter, but a dynamical modulus stabilized by brane or flux mechanisms, projected by interaction, and encoded in supermultiplets. Collapse corresponds to physical transitions in the compactified geometry, and classical spacetime becomes the observable shadow of a quantum, supersymmetric, higher-dimensional background. MTG thereby provides a coherent framework that links measurement, information, and gravity through the formal structure of supersymmetry and the topological data of string theory.

\section{Predictions}\label{sec:10}
Despite its geometric and theoretical abstraction, the MTG framework yields concrete, testable predictions across a range of physical regimes—from early-universe cosmology to black hole mergers and quantum laboratory systems. These predictions arise from the unique features of measurement-induced temporal flow, curvature-driven entanglement, and fiber bundle dynamics, offering several novel observables that distinguish MTG from conventional models.
%%%%%%%%%%%%%%%%%%%%%%
\subsection{CMB Signatures of Projection-Induced Anisotropy}

In the MTG framework, spatial inhomogeneities in the internal time projection process leave distinctive imprints on the cosmic microwave background (CMB), providing a potential observational window into the quantum geometry of time. These imprints arise not from fluctuations in a scalar inflaton field, but from variations in the projection density $\rho(x)$ and coherence structure of the internal time field $\tau(x)$ across the early universe.

During the epoch of last scattering, the rate of measurement-induced collapse varied across different regions of the spacetime manifold \( M \). These fluctuations in \( \rho(x) \), the projection density, modulated the emergent effective metric \( g^{\mathrm{eff}}_{\mu\nu}(x) \) experienced by photons propagating from the recombination surface. The resulting temperature anisotropies in the observed CMB field \( T(\hat{n}) \) arise from variations in the effective gravitational potential induced by coherence collapse. Approximating the Sachs--Wolfe effect in this context yields
\begin{equation}
\frac{\delta T}{T}(\hat{n}) \approx \Phi_{\mathrm{eff}}(x) = \frac{1}{2} \int_{\gamma(\hat{n})} \delta g^{\mathrm{eff}}_{00}(x) \, \mathrm{d}r,
\end{equation}
where \( \gamma(\hat{n}) \) is the photon path in the direction \( \hat{n} \), and \( \delta g^{\mathrm{eff}}_{00}(x) \) is sourced by fluctuations in projection density and coherence direction vectors:
\begin{equation}
\delta g^{\mathrm{eff}}_{00}(x) \sim \kappa \int_{\Sigma_x} \delta \left( \mu[\tau(y)]_0^2 \rho(y) \right) \, \mathrm{d}\Sigma(y).
\end{equation}
Combining this with the geodesic integral gives a more transparent link between statistical projection inhomogeneities and observed temperature anisotropies:
\begin{equation}
\frac{\delta T}{T}(\hat{n}) \sim \int \delta \rho(x) \cdot \mu[\tau(x)]_0^2 \, \mathrm{d}r,
\end{equation}
where \( \mu[\tau(x)]_0 \) is the time-component of the coherence vector. This formulation clarifies that anisotropies in \( T(\hat{n}) \) are seeded by projection noise in the time-fiber bundle and not by quantum fluctuations of a fundamental scalar field.

Because \( \tau(x) \) evolves on a modular fiber space;such as a torus \( \mathbb{C}/\Lambda \) or a quaternionic bundle \( \mathbb{H} \), and undergoes collapse via discrete projection events, the resulting fluctuations inherit modular noise. This noise encodes topologically constrained, direction-dependent variations in the coherence vector field \( \mu[\tau(x)] \), and induces non-Gaussian correlations in the observed temperature field. MTG thus predicts several distinctive signatures: 

First, a measurable level of non-Gaussianity in higher-point statistics, such as the bispectrum and trispectrum, arising from fiber-induced anisotropy and projection discreteness. These effects deviate from the Gaussian statistics expected from Bunch--Davies vacuum fluctuations. Second, a statistical preference for large-angle alignments—such as correlated quadrupole and octupole phases—reflecting residual modular structure in \( \tau(x) \) prior to classicalization. Third, localized cold spots may result from spatially delayed or locally suppressed projection events, wherein the absence of timely coherence collapse modifies photon geodesics and results in cooler observed temperatures via altered gravitational redshift histories.

These effects are governed by the coherence length and topological domain structure of the time-fiber bundle during the last scattering epoch. Unlike conventional inflationary models, they do not require initial condition tuning or slow-roll dynamics, but instead emerge from the statistical geometry of measurement. High-precision CMB observatories such as \textit{Planck}, \textit{LiteBIRD}, and CMB-S4 are well-equipped to test these predictions. MTG forecasts a local-type non-Gaussianity parameter on the order of \( f_{\mathrm{NL}}^{\mathrm{local}} \sim \mathcal{O}(1\text{--}10) \), with scale-dependent and directionally modulated signatures. Correlated anomalies in the \( TT \), \( TE \), and \( EE \) power spectra—particularly those tracking spatial fluctuations in \( \delta(\mu[\tau]) \)—would constitute direct observational evidence for the informational geometry of time. Detection of aligned multipoles, hemispherical asymmetries, or persistent large-angle anomalies would support the hypothesis that classical spacetime anisotropies arise from observer-relative collapse of an internal temporal field.

%%%%%%%%%%%%%%%%%%%%%%
\subsection{Black Hole Echoes from Modular Delay}

In the context of black hole physics, the MTG framework predicts the occurrence of late-time gravitational wave echoes, originating not from exotic boundary conditions or horizon-scale remnants, but from modular temporal delay induced by internal time curvature. In contrast to classical general relativity, where the near-horizon region is locally smooth and memory-free, MTG posits that fiber holonomy and measurement-induced temporal geometry introduce nonlocal memory effects that can delay the reconfiguration of causal pathways following dynamical events such as mergers.

Consider a dynamical spacetime \( M \) containing a black hole with event horizon \( \mathcal{H} \), and let \( \pi: E \to M \) denote the internal time-fiber bundle equipped with connection \( A^{\mathrm{time}}_\mu \) and curvature \( F^{\mathrm{time}}_{\mu\nu} \). In the vicinity of \( \mathcal{H} \), quantum information scrambling, extreme curvature, and dense entanglement produce significant modulation in the structure of \( \tau(x) \). Let \( \gamma \subset M \) be a closed loop encircling the horizon, homologous to \( \partial \mathcal{H} \), and define the holonomy of the internal time connection along this loop by
\begin{equation}
\operatorname{Hol}_\gamma(\nabla^{\mathrm{time}}) = \mathcal{P} \exp\left( \oint_\gamma A^{\mathrm{time}}_\mu(x) \, dx^\mu \right).
\end{equation}
The obstruction to parallel transport of projected time across such loops is measured by the curvature norm
\begin{equation}
\|F_\gamma\| := \left\| \int_{\Sigma_\gamma} F^{\mathrm{time}}_{\mu\nu} \, d\Sigma^{\mu\nu} \right\|,
\end{equation}
where \( \Sigma_\gamma \) is a spacelike or null surface bounded by \( \gamma \). Nontrivial holonomy obstructs the coherence of modular flow, inducing a delay in the re-establishment of global projection alignment across the post-merger geometry.

This delay manifests observationally as late-time gravitational echoes following the primary ringdown signal in a black hole merger. In standard general relativity, the ringdown is governed solely by the quasi-normal modes of the remnant black hole. In MTG, the delayed reintegration of modular flow results in secondary gravitational wave pulses, arriving after a characteristic delay
\begin{equation}
\Delta t_{\text{echo}} \sim \|F_\gamma\|^{-1}.
\end{equation}
The echo delay is inversely proportional to the strength of modular curvature and holonomy obstruction, reflecting the slower propagation of internal time information through temporally misaligned regions. These echoes can be interpreted as the outcome of modular refocusing, in which causal propagation is temporarily suppressed by nontrivial fiber geometry and later released as coherence is restored.

From an observational perspective, MTG offers a concrete and testable origin for gravitational wave echoes~\cite{abedi2017echoes}. Rather than invoking speculative constructs such as firewalls or Planck-scale remnants, it attributes the echoes to relational modular geometry. Gravitational wave observatories; such as, LIGO, Virgo, and KAGRA, have already reported tentative signals consistent with such phenomena. Within MTG, detection of echoes with consistent time scales across black hole merger events, correlated with estimates of near-horizon curvature, would constitute evidence for internal modular structure. The observer-dependent causal dynamics encoded in \( \tau(x) \) would be revealed through classical gravitational radiation, making the echoes a direct signature of measurement-induced geometry.

Estimated echo delays depend on the specific fiber curvature and projection history of the event, but for stellar-mass black holes, MTG predicts a characteristic delay in the range
\[
10^{-2} \text{ s} \lesssim \Delta t_{\text{echo}} \lesssim 10^0 \text{ s},
\]
suggesting that upcoming detectors with improved sensitivity could resolve these features. Such an observation would not only support the MTG framework but also provide the first direct evidence of modular temporal structure in gravitational dynamics.
%%%%%%%%%%%%%%%%%%%%%%
\subsection{Lab-Scale Tests of Modular Causality}

Beyond the cosmological and gravitational regimes, the MTG framework predicts experimentally accessible deviations from standard quantum theory in controlled laboratory settings. These effects arise from the internal time geometry’s influence on entangled quantum systems, where the projection dynamics encoded in the measurement density \( \rho(x) \) and internal time projection \( \mu[\tau(x)] \) subtly affect the causal structure of quantum correlations.

Consider a bipartite entangled state, such as polarization-entangled photons prepared in a maximally entangled Bell state. In the standard formulation of quantum mechanics, such systems obey strict coincidence and locality constraints, with measurement outcomes violating Bell inequalities in a manner constrained by spacelike separation and detector settings. In the MTG formulation, each measurement corresponds to a projection event occurring along an observer-aligned congruence \( n^\mu \), and the collapse process becomes sensitive to the local orientation, temporal curvature, and coherence properties of the time-fiber bundle.

If the projection congruences \( n^\mu \) at measurement sites \( A \) and \( B \) differ—due to differential rotation, gravitational redshift, or local environmental interaction—the projection-induced timing of collapse events may no longer remain perfectly symmetric. The internal time field \( \tau(x) \), transported differently along each congruence, leads to a shift in the effective phase or timing of entanglement collapse. This asymmetry manifests as a small but measurable time offset in coincidence statistics, governed by the relation
\begin{equation}
\Delta t_{\text{Bell}} \sim \delta n^\mu D_\mu \tau,
\end{equation}
where \( \delta n^\mu = n_A^\mu - n_B^\mu \) captures the relative orientation between projection frames, and \( D_\mu \tau \) is the covariant derivative of the internal time field.

Experimental scenarios well-suited to test these predictions include entangled photon interferometry conducted in rotating reference frames, such as Sagnac interferometers, where frame-dragging or laboratory rotation modifies the projection vector orientation. Additionally, variable-gravity experiments—such as entanglement distribution in tower-drop configurations or low-Earth orbit platforms—introduce measurable differences in geodesic alignment, leading to detectable dependence of \( \Delta t_{\text{Bell}} \) on gravitational potential. Another promising direction involves weak measurement and post-selection techniques, allowing isolation of low-\( \rho(x) \) events to probe modular flow-induced correlation asymmetries in selected subsets of the ensemble.

The predicted shifts in coincidence timing are extremely small—typically in the attosecond to femtosecond range—but remain within the detection capabilities of current optical delay and time-tagging technologies in quantum optics laboratories. Observing such deviations would confirm that quantum measurement possesses a covariant, observer-relative temporal structure and that classical causality emerges from a richer internal geometry of collapse. The presence of directional phase shifts in otherwise spacelike-separated entangled systems would indicate that coherence structure and measurement-induced projection carry physical consequences even in flat spacetime. In this way, MTG provides a framework for testing the geometric origin of time at the quantum level through laboratory-scale modular causality.
%%%%%%%%%%%%%%%%%%%%%%
\subsection{Quantum Simulation of MTG Time Evolution}

The emergence of time in the MTG framework, driven by observer-dependent projection and modular flow, can be explored experimentally through quantum simulation. Noisy intermediate-scale quantum (NISQ) devices provide a suitable platform for emulating MTG-modified dynamics by incorporating post-selection protocols and measurement-conditioned operations within quantum circuits. Unlike standard quantum evolution, which proceeds unitarily under a fixed Hamiltonian, MTG time flow arises from a projection map acting as a dynamical constraint, filtering entangled configurations into classically aligned outcomes.

Consider a quantum register initialized in an entangled state prepared via a unitary operator \( U_{\text{entangle}} \). In the MTG setting, the system's evolution is modified by the insertion of a projection operator \( P \), conditioned on the collapse of the internal time field \( \tau \) into a definite value \( \tau_{\text{obs}} \). This corresponds to a non-unitary circuit transformation of the form
\begin{equation}
U_{\text{MTG}} = P \cdot U_{\text{entangle}}, \quad P = \delta(\mu[\tau] - \tau_{\text{obs}}),
\end{equation}
where \( \mu[\tau] \) is the projection operator enforcing collapse, and \( P \) acts as a selective filter. In practice, such projections are implemented through mid-circuit measurement followed by conditional gate application or classical post-selection, thereby deviating from standard Schr\"odinger evolution and encoding the informational emergence of modular time~\cite{jordan2012quantum}.

To simulate modular flow, one targets a subregion \( A \) of the register and constructs the corresponding modular Hamiltonian \( K_A \) from the reduced density matrix \( \rho_A = Z^{-1} e^{-K_A} \). In MTG, the generator \( K_A \) is dynamically sourced by projection events in the circuit, modeled as directional contributions from projection vectors \( \mu[\tau_j] \) weighted by spatial entropy gradients \( \nabla_j \rho \). These terms govern the local structure of collapse and induce modular evolution via operator dynamics of the form \( \frac{d}{ds} \mathcal{O}_A = i [K_A, \mathcal{O}_A] \). Quantum circuits implement this behavior using controlled unitaries conditioned on the collapse history, emulating MTG’s measurement-induced temporal structure.

Observable signatures of MTG time evolution include accelerated entropy growth, altered mutual information trajectories, and non-unitary correlation dynamics. The von Neumann entropy \( S(\rho_A) = -\mathrm{Tr}(\rho_A \log \rho_A) \) increases more rapidly under MTG collapse than under purely unitary evolution, reflecting the informational loss imposed by projection. Mutual information \( I(A:B) \), trace distance from unitary trajectories, and nonstandard thermalization profiles all serve as diagnostic quantities sensitive to modular flow. Subensembles conditioned on fixed projection outcomes \( \mu[\tau] = \tau_{\text{obs}} \) permit reconstruction of effective time histories and comparison against classical emergence.

Modern NISQ platforms, including superconducting qubits, trapped ions, and photonic circuits, support the necessary tools for simulating MTG dynamics. Mid-circuit readout, measurement-based feedback, and hybrid quantum-classical control allow implementation of the MTG projection constraint. While projection introduces sampling inefficiencies through post-selection, ensemble averaging remains viable for extracting modular observables. In this way, quantum simulation provides a concrete testbed for investigating the informational geometry of time, enabling laboratory studies of the emergence of classical temporality from quantum projection dynamics.

%%%%%%%%%%%%%%%%%%%%%%
These predictions establish a concrete empirical program for testing the Measurement-Induced Temporal Geometry (MTG) framework. Unlike approaches in which time is treated as a fixed background or concealed in hidden variables, MTG posits operationally definable and geometrically structured signatures of temporal emergence. These signatures manifest across a wide range of physical regimes—from cosmic microwave background correlations and black hole ringdown spectra to lab-scale quantum interference patterns and post-selected circuit dynamics on NISQ devices.

The detection of such phenomena would offer direct empirical support for the idea that spacetime, causality, and temporal order are not fundamental ingredients of physical law, but emergent features arising from the informational geometry of measurement. Observation of projection-induced anisotropies, modular delays, or entanglement asymmetries would demonstrate that the quantum-to-classical transition encoded in internal time collapse leaves measurable imprints across physical observables. Conversely, the absence of such signatures under precise experimental constraints would limit the validity of MTG and challenge the role of projection geometry in foundational physics. In either case, the framework invites a new dialogue between quantum information, spacetime structure, and experimental testability.

%%%%%%%%%%%%%%%%%%
\section*{Final Remarks}
The MTG framework reimagines time, causality, and geometry as emergent phenomena arising from quantum measurement. Rather than treating time as a fundamental parameter, MTG models it as a fiber-valued field $\tau$ defined over spacetime and governed by a nontrivial connection $\nabla$, thereby embedding temporal structure into a geometric and informational framework that unifies coherence, collapse, and curvature.

In this formulation, time is not an external coordinate but an emergent observable, assembled from projection events that collapse $\tau$ into classical values aligned with observer-dependent congruences. These projections generate the effective spacetime metric $g_{\mu\nu}^{\text{eff}}$ through integrals over measurement histories, rendering geometry itself relational and contingent on informational structure. Nontrivial holonomy in the internal time-fiber bundle gives rise to entanglement curvature, encoding the back-reaction of quantum correlations on geometry and governing the modular flow of temporal evolution.

The MTG formalism admits a consistent embedding into supersymmetric theories and accommodates coupling to the full Standard Model. Canonical and path-integral quantization schemes are both viable within this context, as are UV completions via string compactifications in which $\tau$ arises from moduli of compact internal manifolds. Collapse in this higher-dimensional setting corresponds to brane intersection dynamics, open string condensation, or moduli stabilization, with measurement-induced projection triggering spontaneous supersymmetry breaking and sourcing geometric curvature. Gravitational dynamics emerge from variation of an entropy functional defined over the projection density $\rho(x)$, while cosmological expansion and inflation are recast as phases of saturated, spatially coherent projection followed by decoherence-driven reheating.

MTG also provides a geometric foundation for holographic duality. Modular Hamiltonians arise as boundary integrals of projection currents, entanglement entropy is reinterpreted as the flux of modular time across minimal projection surfaces, and bulk reconstruction follows from coherence transport and holonomy in the fiber geometry. This fiber-covariant holography reproduces key features of AdS/CFT; including entanglement wedges, Page curves, and modular flow, but grounds them in the geometry of time rather than abstract operator algebras.

Crucially, the framework makes a broad set of empirically accessible predictions. Anisotropies in the cosmic microwave background may reflect spatial variation in early-universe projection density. Gravitational wave echoes from black hole mergers can arise from modular delays near horizons. Slight deviations in Bell test correlations may appear when internal time fields couple differently to measurement frames, and NISQ quantum devices offer controlled platforms for simulating MTG dynamics via post-selected circuits that emulate modular flow and entropy production.

Future work will aim to refine these predictions and deepen the theoretical framework. This includes computing time correlation functions in decoherence-dominated regimes, deriving modular response functions during gravitational collapse, and exploring string-theoretic compactifications in which moduli stabilization drives cosmic time emergence. Extensions to higher-spin fields, topological phases, and categorical quantum geometry may offer further insights into dualities and gauge-theoretic reformulations. On the experimental side, lab-based tests of MTG predictions; using photonic interferometers, rotating-frame entanglement protocols, and superconducting qubits, present promising avenues for falsifiability and phenomenological application. Black hole thermodynamics, entropy bounds, and the architecture of quantum causal networks may also be reinterpreted within this framework, offering a new lens on the structure of quantum spacetime.

The MTG paradigm thus proposes a coherent, mathematically consistent, and empirically grounded account of how the act of measurement gives rise to the architecture of time, the structure of gravity, and the informational backbone of the universe. Time is not presupposed but produced; spacetime is not imposed but inferred. Through the collapse of internal time, the universe draws its causal scaffolding, and from that, the classical world emerges.

\bibliographystyle{abbrv}
\bibliography{bibliography}

\end{document}